%% file: paper.tex
\documentclass[letterpaper]{article} 
\usepackage{aaai2026}  
\input{preambule}

\usepackage{times}  
\usepackage{helvet}  
\usepackage{courier}  
\usepackage[hyphens]{url}  
\usepackage{graphicx} 
\usepackage{natbib}  
\usepackage{caption} 
\usepackage{algorithm}
\usepackage[noend]{algorithmic}

\usepackage{newfloat}
\usepackage{listings}
\DeclareCaptionStyle{ruled}{labelfont=normalfont,labelsep=colon,strut=off} 
\floatstyle{ruled}
\newfloat{listing}{tb}{lst}{}
\floatname{listing}{Listing}
\title{On Dynamic Programming Theory for Leader–Follower Stochastic Games}
\author{
    Jilles Steeve Dibangoye\textsuperscript{\rm 1},
    Thibaut Le Marre\textsuperscript{\rm 2},
    Ocan Sankur\textsuperscript{\rm 3},
    François Schwarzentruber\textsuperscript{\rm 4}}
\affiliations{
    \textsuperscript{\rm 1}University of Groningen, j.s.dibangoye@rug.nl\\
    \textsuperscript{\rm 2}ENS de Lyon, CNRS, University Claude Bernard Lyon 1, Inria, LIP, UMR 5668, France, t.v.d.le.marre@rug.nl\\
    \textsuperscript{\rm 3}ENS de Lyon, CNRS, University Claude Bernard Lyon 1, Inria, LIP, UMR 5668, France, francois.schwarzentruber@ens-lyon.fr\\
    \textsuperscript{\rm 4} Univ Rennes, Inria, CNRS, IRISA, France, ocan.sankur@cnrs.fr\\

}

\begin{document}

\maketitle

\begin{abstract}
\input{body/abstract}
\end{abstract}

\input{body/introduction}

\input{body/background}

\input{body/lossless-reduction}

\input{body/formal-results}
\input{body/point-based-value-iteration}
\input{body/experiments}
\input{body/conclusion}

\bibliography{aaai2026}

\input{appendix/appendix}

\end{document}

%% file: preambule.tex
\usepackage{microtype}
\usepackage{graphicx}
\usepackage{subfigure}
\usepackage{booktabs} 

\usepackage{cancel}
\usepackage{xgames}
\usepackage{fikz}
\usepackage{pgfplots}
\pgfplotsset{compat=1.7}
\usepackage{tikz}
\usetikzlibrary{shapes}
\usetikzlibrary{decorations.pathmorphing}
\usetikzlibrary{arrows} 
\usetikzlibrary{automata}
\usetikzlibrary{fit}
\usetikzlibrary{shadows}
\usetikzlibrary{trees}
\usetikzlibrary{decorations}
\usetikzlibrary{decorations.text}
\usetikzlibrary{positioning}
\usetikzlibrary{patterns}
\usetikzlibrary{backgrounds}
\usetikzlibrary{matrix}
\usetikzlibrary{arrows.meta}
\usetikzlibrary{calc, shapes}
\usetikzlibrary {petri} 
\usetikzlibrary{chains,scopes}
\usetikzlibrary{spy}
\usepackage{array,multirow,makecell}
\usepackage[mathcal]{euscript}
\usepackage{stmaryrd}
\usepackage{bbold}
\usepackage{paralist}
\usepackage{sidecap}

\usepackage{xspace}

\def\eg{{\em e.g.,}\xspace}

\usepackage{placeins}

	\definecolor{sthlmLightBlue}{RGB}{214,237,252} 
		
	\definecolor{sthlmBlue}{RGB}{0,110,191} 

	\definecolor{sthlmLightGreen}{RGB}{213,247,244} 
		
	\definecolor{sthlmGreen}{RGB}{0,134,127} 

	\definecolor{sthlmLightGrey}{RGB}{213,217,225} 
		
	\definecolor{sthlmGrey}{RGB}{245,243,238} 
		
	\definecolor{sthlmDarkGrey}{RGB}{51,51,51} 

	\definecolor{sthlmLightOrange}{RGB}{255,215,210} 
		
	\definecolor{sthlmOrange}{RGB}{221,74,44} 

	\definecolor{sthlmLightPurple}{RGB}{241,230,252} 
		
	\definecolor{sthlmPurple}{RGB}{93,35,125} 

	\definecolor{sthlmLightRed}{RGB}{254,222,237} 
		
	\definecolor{sthlmRed}{RGB}{196,0,100} 

	\definecolor{sthlmYellow}{RGB}{252,191,10} 

\usepackage{colortbl}%
  \newcommand{\myrowcolour}{\rowcolor[gray]{0.925}}
\usepackage{booktabs}

\usepackage{amsmath}
\usepackage{amssymb}
\usepackage{mathtools}
\usepackage{amsthm}

\usepackage[capitalize,noabbrev]{cleveref}

\newtheorem{theorem}{Theorem}[section]
\newtheorem{proposition}[theorem]{Proposition}
\newtheorem{lemma}[theorem]{Lemma}
\newtheorem{corollary}[theorem]{Corollary}
\newtheorem{definition}{Definition}
\newtheorem{example}{Example}

\usepackage[textsize=tiny]{todonotes}

\newcommand{\colorL}{{\pl1{L}}}
\newcommand{\colorF}{{\pl2{F}}}
\newcommand{\colorI}{i}

\DeclareMathOperator*{\argmax}{\arg\!\max}

\tikzset{
  riverflow/.style={
    -, >=stealth,
    line cap=round,
    line join=round,
    line width=10pt,  
    opacity=0.3       
  }
}

\newcommand\Choose{\mathtt{Choose}}
\newcommand\unique{\mathtt{Unique}}

\newcommand\isInteresting[1]{\mathtt{isInteresting}[#1]}
\newcommand\isKeptInPruning[1]{\mathtt{isKept}[#1]}

%% file: body/abstract.tex
Leader–follower general-sum stochastic games (LF-GSSGs) model sequential decision-making under asymmetric commitment, where a leader commits to a policy and a follower best responds, yielding a strong Stackelberg equilibrium (SSE) with leader-favourable tie-breaking. This paper introduces a dynamic programming (DP) framework that applies Bellman recursion over credible sets—state abstractions formally representing all rational follower best responses under partial leader commitments—to compute SSEs. We first prove that any LF-GSSG admits a lossless reduction to a Markov decision process (MDP) over credible sets. We further establish that synthesising an optimal memoryless deterministic leader policy is NP-hard, motivating the development of \(\varepsilon\)-optimal DP algorithms with provable guarantees on leader exploitability. Experiments on standard mixed-motive benchmarks—including security games, resource allocation, and adversarial planning—demonstrate empirical gains in leader value and runtime scalability over state-of-the-art methods.

%% file: body/introduction.tex
\section{Introduction}
\label{sec:intro}

Leader–follower general-sum stochastic games (LF-GSSGs) model sequential planning under asymmetric commitment, where a leader commits to a policy and a follower best responds. The solution concept of interest is the strong Stackelberg equilibrium (SSE), which assumes that the follower resolves ties in favour of the leader. LF-GSSGs arise in numerous applications requiring robust decision-making under strategic uncertainty, including adversarial patrolling, network security, infrastructure protection, and cyber-physical planning~\cite{Yin_Jiang_Tambe_Kiekintveld_Leyton-Brown_Sandholm_Sullivan_2012,AnKKSSTV12,1503282,BasilicoNG16,BasilicoCG16,chung2011search}.

Dynamic programming (DP) has proven foundational in planning under uncertainty for Markov decision processes (MDPs)~\cite{bellman,Puterman1994} and zero-sum stochastic games (zs-SGs)~\cite{shapley_1953_stochastic,HorBos-aaai19,ZHENG2022110231,HorakBKK23}, where recursive decompositions and Markovian policies enable tractable value computation. However, in LF-GSSGs, the asymmetry of commitment and general-sum structure fundamentally alter the nature of planning: the follower's response may depend on the entire interaction history, and the leader must anticipate all such best responses. Subgame decomposability fails, standard state-based recursions no longer hold, and Markov policies are not sufficient for optimality~\cite{10.5555/2900929.2900938,lopez2022stationary}.

These limitations manifest even in simple deterministic environments. In the centipede game of Figure~\ref{fig:centipede:game}, backward induction predicts that rational agents will terminate the game immediately, despite the existence of more rewarding cooperative trajectories~\cite{rosenthal1981games,binmore1987modeling,mckelvey1992experimental,megiddo1986remarks,reny1988rationality,kreps2020course}. This outcome illustrates the failure of Bellman-style reasoning in the presence of social dilemmas,\footnote{Situations where mutual cooperation yields the highest joint payoff, yet each agent has incentives to defect unilaterally.} a limitation that only intensifies under stochastic transitions and asymmetric commitment.

\begin{figure}[H]
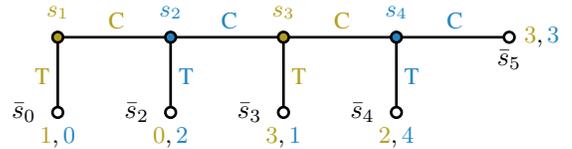

\centering
\begin{gametree}[index=action, xscale=1.5, node autolabel=$s_{\enum}$]
    \branch[root=w, player=1, offset={0.5,0}, sticky] from (0,0) to[v={1,0},h=1] {T(1,0); C};
    \branch[player=2] from (wC) to {T(0, 2); C};
    \branch[player=1] from (wCC) to {T(3, 1); C};
    \branch[player=2] from (wCCC) to[h=1] {T(2, 4); C[0](3, 3)};
    \node at (-0.3, -1) {$\bar{s}_0$};
    \node at (0.7, -1) {$\bar{s}_2$};
    \node at (1.7, -1) {$\bar{s}_3$};
    \node at (2.7, -1) {$\bar{s}_4$};
    \node at (4, -0.3) {$\bar{s}_5$};
\end{gametree}
\caption{Centipede game. The game begins in state \(s_1\). The leader (\pl1{}) acts in \(s_1\) and \(s_3\); the follower (\pl2{}) acts in \(s_2\) and \(s_4\). Terminal states \(\bar{s}_1, \dots, \bar{s}_5\) yield payoffs shown as (\pl1{leader}, \pl2{follower}).}
\label{fig:centipede:game}
\end{figure}

The computational landscape reflects these structural difficulties. While SSEs can be computed in polynomial time for normal-form games~\cite{conitzer2006computing}, the problem becomes NP-hard with succinct action representations~\cite{DBLP:conf/aaai/KorzhykCP10}, PSPACE-complete in STRIPS-like planning domains~\cite{DBLP:conf/icaps/BehnkeS24}, and even NEXPTIME-complete in multi-objective arenas~\cite{DBLP:journals/tocl/BruyereFRT24}. Existing methods either encode the problem as large mixed-integer programs~\cite{10.5555/2900929.2900938,letchford2010computing,3038794.3038831} or rely on extensive linear programs~\cite{conitzer2006computing}, both of which scale poorly with horizon or state space. Simplifying assumptions—such as myopic or omniscient followers~\cite{Denardo1967,Whitt1980}, or stationary Markov policies~\cite{lopez2022stationary}—limit practical applicability and fail to preserve generality.

This work introduces a new value-based dynamic programming framework for LF-GSSGs, grounded in a structural reduction to what we call a \emph{credible Markov decision process} (credible MDP). In this reformulation, states correspond to \emph{credible sets}—finite collections of occupancy states induced by a fixed leader policy and all rational follower responses. Each occupancy state is a distribution over joint histories of environment states and actions, induced by a follower response to a partial leader policy. Transitions between credible sets are deterministic and governed by the leader’s decision rules, while accounting for all follower responses that are \emph{admissible} under SSE semantics—that is, responses that could be extended into an optimal policy with leader-favourable tie-breaking, under an extension of the leader policy consistent with the current prefix.

\paragraph{Contributions.} The reduction is shown to be lossless: optimal leader value is preserved without requiring explicit follower enumeration. To characterise computational hardness, it is further proven that synthesising an optimal memoryless deterministic leader policy in an LF-GSSG is \emph{NP}-hard. Nonetheless, the value function over credible sets exhibits \emph{uniform continuity}, enabling Bellman-style recursion and point-based approximation with provable guarantees on leader exploitability. The resulting dynamic programming algorithms achieve \(\varepsilon\)-optimality without full policy enumeration and scale to large-horizon settings. Empirical evaluations on mixed-motive benchmarks—including security games, resource allocation, and adversarial planning—demonstrate improvements in leader value and runtime over MILP-based and DP-inspired baselines.

%% file: body/background.tex
\section{Background}
\label{sec:background}

This section formalises leader--follower general-sum stochastic games (LF-GSSGs), outlines the interaction model, and recalls key concepts from dynamic programming relevant to our framework. A summary of the notation used throughout the paper is provided in Appendix~\ref{appendix:notations}.

\begin{definition}
A \emph{leader--follower general-sum stochastic game} (LF-GSSG) is a tuple
\(
M \doteq (I, S, A_{\colorL}, A_{\colorF}, p, r_{\colorL}, r_{\colorF}, s_0, \gamma),
\)
where \( I = \{\colorL, \colorF\} \) is the set of players, comprising the leader \( \colorL \) and the follower \( \colorF \); \( S \) is a finite set of environment states; \( A_{\colorL} \) and \( A_{\colorF} \) are the finite action spaces of the leader and the follower, respectively. The transition function \( p \colon S \times A_{\colorL} \times A_{\colorF} \times S \to [0,1] \) defines the probability of moving to state \( s' \) from \( s \) under joint action \( (a_{\colorL}, a_{\colorF}) \). The one-step reward functions \( r_{\colorL}, r_{\colorF} \colon S \times A_{\colorL} \times A_{\colorF} \times S \to \mathbb{R} \) determine the immediate rewards received by each player. The game begins in a known initial state \( s_0 \in S \) and proceeds over an infinite number of stages, with rewards geometrically discounted by a factor \( \gamma \in [0,1) \).
\end{definition}

Since exact infinite-horizon planning is computationally prohibitive, we approximate it using a finite planning horizon \( \ell < \infty \). This is theoretically justified: for any \( \varepsilon > 0 \), the total expected return under a policy truncated at stage \( \ell \) differs from its infinite-horizon value by at most \( \varepsilon \), provided that
\[
\textstyle
\ell \geq \left\lceil \log_\gamma \left( \frac{(1 - \gamma)\varepsilon}{m} \right) \right\rceil,
\]
where \( m > 0 \) is a uniform upper bound on one-step reward magnitudes. This follows from geometric tail decay and enables tractable approximation to arbitrary precision.

\paragraph{Information structure and example.}
The interaction unfolds over discrete time steps \( t \in \{0, 1, \dots, \ell - 1\} \). At each stage \( t \), both players observe the current state \( s_t \), but neither observes the other’s action unless explicitly encoded in the state. Each player is informed only of the current state and their own past actions. We use the centipede
game as a running example throughout the paper

\begin{example}
The centipede game, shown in Figure~\ref{fig:centipede:game}~\cite{rosenthal1981games}, involves two players who alternately choose whether to continue or terminate the game by taking a larger share of an increasing pot. If a player chooses to take, the game ends immediately with that player receiving the larger share. Despite its simplicity, the centipede game illustrates a breakdown of backward induction due to strategic interdependence: classical dynamic programming terminates too early, failing to exploit deeper cooperative potential. This failure motivates our development of a more robust planning framework for LF-GSSGs.
\end{example}

\paragraph{Histories and decision rules.}
Each player makes decisions based on their \emph{private history}. For player \( \colorI \in \{\colorL, \colorF\} \), the history at time \( t \) is
\(
h_{\colorI,t} \doteq (s_0, a_{\colorI,0}, s_1, \dots, a_{\colorI,t-1}, s_t)
\),
capturing the sequence of observed states and the player’s own actions. Let \( H_{\colorI,t} \) denote the set of such histories. A \emph{decision rule} is a mapping \( \delta_{\colorI,t} \colon H_{\colorI,t} \to \Delta(A_{\colorI}) \); the leader may use stochastic rules, while the follower is restricted to deterministic ones~\cite{conitzer2006computing}. Let \( \Delta_{\colorI,t} \) be the set of all such rules. A \emph{policy} is a sequence \( \pi_{\colorI} = (\delta_{\colorI,0}, \dots, \delta_{\colorI,\ell-1}) \), and the joint policy is \( \pi = (\pi_{\colorL}, \pi_{\colorF}) \).

\paragraph{Value functions under a joint policy.}
Fix a joint policy \( \pi = (\pi_{\colorL}, \pi_{\colorF}) \). For player \( \colorI \), the action-value function at stage \( t \) is defined for environment state \( s \), joint history \( h = (h_{\colorL}, h_{\colorF}) \), and joint action \( a = (a_{\colorL}, a_{\colorF}) \) as
\[
\textstyle
q_{\colorI,t}^{\pi}(s,h,a) \doteq \mathbb{E}_{\pi} \left[ \sum_{k=t}^{\ell - 1} \gamma^{k - t} r_{\colorI}(s_k, a_{\colorL,k}, a_{\colorF,k}, s_{k+1}) \right],
\]
where the expectation is taken over future trajectories consistent with policy \( \pi \), given that history \( (h_{\colorL}, h_{\colorF}) \) is observed at stage \( t \) and joint action \( (a_{\colorL}, a_{\colorF}) \) is executed. The environment state \( s \) is included explicitly (though determined by the histories) to simplify belief updates in later constructions.

The value function at stage \( t \) is then
\[
\textstyle
v_{\colorI,t}^{\pi}(s, h_{\colorL}, h_{\colorF}) \doteq \mathbb{E}\left[ q_{\colorI,t}^{\pi}(s, h_{\colorL}, h_{\colorF}, a_{\colorL}, a_{\colorF}) \right],
\]
with the expectation taken over \( a_{\colorL} \sim \delta_{\colorL,t}(h_{\colorL}) \), \( a_{\colorF} = \delta_{\colorF,t}(h_{\colorF}) \), and \( s_{t+1} \sim p(s_t, a_{\colorL}, a_{\colorF}, \cdot) \). The terminal value is \( v_{\colorI,\ell}^{\pi} \equiv 0 \)~\cite{Puterman1994}.

\paragraph{Follower best-response set.}
Given a fixed leader policy \( \pi_{\colorL} \), the set of follower best responses is
\[
\textstyle
\mathtt{BR}(\pi_{\colorL}) \doteq \argmax_{\pi_{\colorF} \in \Pi_{\colorF}} \; v_{\colorF,0}^{(\pi_{\colorL}, \pi_{\colorF})}(s_0, s_0, s_0).
\]
To evaluate leader performance against optimal responses, define the leader-evaluated action-value function as
\(
q^{\pi_{\colorL}}_{\colorI,t}(s,h,a) \doteq \max_{\pi_{\colorF} \in \mathtt{BR}(\pi_{\colorL})} \; q^{(\pi_{\colorL}, \pi_{\colorF})}_{\colorI,t}(s,h,a),
\)
and the corresponding value function:
\[
v^{\pi_{\colorL}}_{\colorI,t}(s,h_{\colorL},h_{\colorF}) \doteq \max_{\pi_{\colorF} \in \mathtt{BR}(\pi_{\colorL})} \; \mathbb{E}[q^{(\pi_{\colorL}, \pi_{\colorF})}_{\colorI,t}(s,h_{\colorL},h_{\colorF},a_{\colorL},a_{\colorF})].
\]

\paragraph{Strong Stackelberg equilibrium.}
A leader policy \( \pi_{\colorL}^\star \) induces a \emph{strong Stackelberg equilibrium} (SSE) if
\[
\textstyle
\pi_{\colorL}^{\star} \in \argmax_{\pi_{\colorL} \in \Pi_{\colorL}} \max_{\pi_{\colorF} \in \mathtt{BR}(\pi_{\colorL})} \; v_{\colorL,0}^{(\pi_{\colorL}, \pi_{\colorF})}(s_0, s_0, s_0).
\]

\paragraph{Limitations of state-of-the-art approaches.}
State-of-the-art methods for computing SSEs in LF-GSSGs fall into two main categories: solving large-scale MILPs~\cite{10.5555/2900929.2900938,letchford2010computing,3038794.3038831} and extensive linear programs~\cite{conitzer2006computing}. These approaches rely on normal-form representations and scale poorly, while offering limited structural insight. Dynamic programming, highly effective in zero-sum and common-payoff settings, fails to extend to general-sum games. A natural but flawed idea is to compute, at each state and stage, a leader action that maximises expected payoff given a best-responding follower. The following proposition shows that replacing a single-stage leader decision rule using backward induction—even under the full follower value function—can invalidate equilibrium guarantees.

\begin{proposition}[Proof in Appendix~\ref{appendix:prop:wrong:pi}]
\label{prop:wrong:pi}
Let \( \pi_{\colorL} \) be a Markov leader policy. Define a new policy \( \pi'_{\colorL} \) that differs from \( \pi_{\colorL} \) only at stage \( t < \ell \), where \( \pi'_{\colorL} \) selects a decision rule \( \delta'_{\colorL,t}(s) \) solving:
\[
\begin{aligned}
&\max_{\delta_{\colorL}(s) \in \Delta(A_{\colorL})} \max_{a_{\colorF} \in A_{\colorF}} \;
\mathbb{E}_{a_{\colorL} \sim \delta_{\colorL}(s)} \left[ q^{\pi_{\colorL}}_{\colorL,t}(s,s,s,a_{\colorL},a_{\colorF}) \right] \\
&\text{s.t.} \quad \mathbb{E}_{a_{\colorL} \sim \delta_{\colorL}(s)} \left[ q^{\pi_{\colorL}}_{\colorF,t}(s,s,s,a_{\colorL},a_{\colorF}) \right] = v^{\pi_{\colorL}}_{\colorF,t}(s,s,s).
\end{aligned}
\]
Then \( \pi'_{\colorL} \) may fail to induce a strong Stackelberg equilibrium, even if \( \delta'_{\colorL,t}(s) = \delta_{\colorL,t}(s) \) for all \( s \) and all other stages.
\end{proposition}

This result motivates the \emph{myopic follower} approximation, which replaces the full follower value function with immediate reward \cite{fujiwara2006stackelberg,lopez2022stationary,JMLR:v24:22-0203}. Proposition~\ref{prop:wrong:pi} shows that even when using the correct long-term value function, backward induction may yield policies that violate SSE optimality. The myopic simplification is thus even less principled and less reliable in general.
Such failure arises in domains like \textsc{Centipede}, where locally rational leader deviations may trigger early termination and suboptimal outcomes (e.g., \( (1,0) \) instead of \( (2,4) \)). This breakdown stems from two structural barriers: (i) the leader faces a \emph{multi-objective} optimisation problem, in which actions optimal against one best response may be dominated under another; and (ii) the problem is \emph{non-Markovian} from the follower's perspective, making correct interpretation of leader actions dependent on history.

\begin{center}
\begin{tikzpicture}
\node[draw=black, thick, rounded corners, inner sep=6pt] {
  \begin{minipage}{0.95\linewidth}
  \emph{How can a leader-centric perspective support a reliable and scalable dynamic programming framework for computing strong Stackelberg equilibria in LF-GSSGs?}
  \end{minipage}
};
\end{tikzpicture}
\end{center}

%% file: body/lossless-reduction.tex
\section{Credible Markov Decision Processes}
\label{sec:credible:mdp}

We present a lossless reduction of any leader--follower general-sum stochastic game (LF-GSSG) into a Markov decision process (MDP) from the leader’s perspective. The core idea is to model the leader's planning task as a sequential decision process in which each state summarises the game's evolution under all rational follower responses. We call the resulting model the \emph{credible MDP}.

The credible MDP is constructed step by step. We first define \emph{joint decision-rule histories} and the \emph{occupancy states} they induce. These occupancy states are grouped into \emph{credible sets}, each indexed by a partial leader policy. We then define the deterministic transitions and cumulative rewards associated with these sets. To reduce redundancy, we introduce a filtering mechanism based on value-dominance within belief classes. This enables forward planning while preserving all information relevant to equilibrium computation. A formal proof that the credible MDP constitutes a \emph{lossless reduction} is provided in Section~\ref{sec:theoretical:results}.

\subsection{Occupancy States and Decision-Rule Histories}

At each stage~$t$, the decision process is indexed by joint decision-rule histories $\theta_t = (\theta_{\colorL,t}, \theta_{\colorF,t})$,
where $\theta_{i,t} = (\delta_{i,0}, \dots, \delta_{i,t-1})$ 
denotes the sequence of local decision rules
applied by agent $i \in I$ up to time step \(t-1\).
Each joint history \(\theta_t\) induces an \emph{occupancy state} \(o_{\theta_t}\), defined as the distribution over joint local environment histories given $\theta_t$:
\[
  o_{\theta_t} \colon (s_t, h_{\colorL,t}, h_{\colorF,t}) \mapsto \Pr(s_t, h_{\colorL,t}, h_{\colorF,t} \mid \theta_t).
\]
Although \( s_t \) is is already given in \( h_{\colorL,t} \) and \( h_{\colorF,t} \), we explicitly mention it in \( o_{\theta_t} \) to support belief-based filtering over states. This distribution is induced by the initial state, the system dynamics, and the policy history~$\theta_t$.
Given any occupancy state $o_\theta$ and a joint decision rule $\delta = (\delta_{\colorL}, \delta_{\colorF})$, we define a deterministic transition operator:
\begin{align*}
\tau &\colon (o_\theta, \delta) \mapsto o_{(\theta, \delta)},
\end{align*}
where $o_{(\theta, \delta)}$ is the next occupancy state resulting from applying 
$\delta$ to $o_\theta$ and propagating through the dynamics.
%
Formally, for any state \(s'\), joint history \((h_{\colorL},h_{\colorF})\), and joint action \((a_{\colorL},a_{\colorF})\) : \(o_{(\theta, \delta)}(s',(h_{\colorL}, a_{\colorL}, s'),(h_{\colorF}, a_{\colorF}, s')) = \sum_s
o_\theta(s,h_{\colorL},h_{\colorF})
\delta_{\colorL}(a_{\colorL}|h_{\colorL}) 
\delta_{\colorF}(a_{\colorF}|h_{\colorF})
p(s, a_{\colorL}, a_{\colorF}, s')\).
%
In addition, we define the \emph{value-so-far function} $\rho_i \colon o_t \mapsto \mathbb{R}$ for each agent $i \in I$, 
mapping each occupancy state $o_t$
to the expected cumulative discounted reward accrued up to stage~$t$:
\[
\rho_i(o_t) \doteq \textstyle \mathbb{E}_{o_t} \left[ \sum_{t'=0}^{t-1} \gamma^{t'} \, r_i(s_{t'}, a_{\colorL,t'}, a_{\colorF,t'}, s_{t'+1}) \right]
\]
where the expectation is taken with respect to the joint distribution over histories encoded in $o_t$. Note that $o_t$ encodes sufficient information to evaluate this cumulative reward.

\subsection{Credible Sets under Epistemic Pessimism}

We now define the credible state space for the leader's planning problem. At stage~$t$, the leader's local decision-rule history $\theta_{\colorL,t}$ is fully specified, while the follower's decision-rule history $\theta_{\colorF,t}$ is unknown and must be reasoned over. To account for all possible rational follower behaviours consistent with the game structure, we adopt a pessimistic closure over follower responses.
Specifically, we define the \emph{credible set} at stage~$t$ as the set of occupancy states reachable from the known leader decision-rule history and all follower decision-rule histories:
\[
  c_{\theta_{\colorL,t}} \doteq \left\{ o_{(\theta_{\colorL,t}, \theta_{\colorF,t})} \;\middle|\; \theta_{\colorF,t} \in \Theta_{\colorF,t} \right\},
\]
where $\Theta_{\colorF,t}$ is the set of all follower decision-rule histories of length~$t$. 
%
When the leader selects a new decision rule $\delta_{\colorL,t}$, her decision-rule history updates to $\theta_{\colorL,t+1} = \theta_{\colorL,t} \delta_{\colorL,t}$.
The corresponding successor credible set is defined by applying the transition operator $\tau$ to each element of $c_{\theta_{\colorL,t}}$ under every follower decision rule:
\[
c_{\theta_{\colorL,t+1}} \doteq \left\{ \tau(o_t, (\delta_{\colorL,t}, \delta_{\colorF,t})) \mid o_t \in c_{\theta_{\colorL,t}}, \delta_{\colorF,t} \in \Delta_{\colorF,t} \right\},
\]
where $c_{\theta_{\colorL,t+1}} \doteq T(c_{\theta_{\colorL,t}}, \delta_{\colorL,t})$ denotes the resulting set. This transition is deterministic from the leader's perspective and compactly encodes all possible follower responses compatible with epistemic pessimism.

\subsection{Marginal Beliefs and Their Role in Filtering}

Forward search over credible sets leads to a rapid growth in the number of reachable occupancy states. However, many of these states may be strategically redundant \emph{under Markov policies}: they differ in the distribution over local histories but induce the same belief over the current environment state~$s_t$.
In order to capture this redundancy, we define the \emph{marginal belief} induced by a joint decision-rule history  $\theta$ as
\[
  b_{o_{\theta}}(s_t) \doteq \textstyle 
  \sum_{(h_{\colorL,t}, h_{\colorF,t})} o_\theta(s_t, h_{\colorL,t}, h_{\colorF,t}).
\]
Since the game is fully observable and Markovian, the marginal belief~$b_{o_{\theta}}$ summarises all information relevant for predicting future transitions and rewards.
However, it is important to emphasise that \emph{belief equivalence does not imply full equivalence of occupancy states}. Two occupancy states may induce the same marginal belief but differ in their distributions over local histories, and thus in the cumulative reward they have already accrued. As a result, belief equivalence is insufficient for reducing the occupancy-state space directly.
Nevertheless, when comparing future continuations from belief-equivalent states, the belief is sufficient to reason about the \emph{value-to-go}. We exploit this fact to define a pruning procedure over credible sets that removes dominated occupancy states sharing the same marginal belief.

\subsection{Belief-Based Filtering over Credible Sets}

To mitigate the combinatorial blow-up during planning, we introduce a \emph{value-dominance filtering} operation under Markov policies. It retains only those occupancy states that are not strictly Pareto-dominated and not exact duplicates within the same marginal belief class. Let $\rho_{\colorL}(o)$ and $\rho_{\colorF}(o)$ denote the cumulative value-so-far for the leader and follower, respectively, under occupancy state $o$. 
Given a credible set $c$, 
we define
$B_c \doteq \{ b_o \mid o \in c \}$ 
and, for each $b \in B_c$, define the belief class 
\[c_b \doteq \{ o \in c \mid b_o = b \}.\]
For all $c_b$, for any occupancy state \(o\in c_b\), define \(c^o_b\) as: \[ \{ o' \in c_b \mid \forall i \in I, \rho_i(o') \geq \rho_i(o), \; \exists j \in I, \rho_j(o') > \rho_j(o) \}.\]
The set $c^o_b$ contains the occupancy states whose the margin belief is $b$ and where the values-so-far Pareto-dominates $o$.
We then define the filtered set of \(c\) as
\[\mathbb{F}(c) \doteq \cup_{b \in B_c} \unique(\{ o \in c_b \mid c^o_b = \emptyset \}),\] 
where $\unique(\cdot)$ retains exactly one representative%
\footnote{
To define $\unique(\{ o \in c_b \mid c^o_b = \emptyset \})$,
we first partition 
$\Omega := \{ o \in c_b \mid c^o_b = \emptyset \}$ 
in sets $\Omega^{\nu_\colorL, \nu_\colorF} = \set{o \in \Omega 
\mid \rho_\colorL(o) = \nu_\colorL, \rho_\colorF(o) = \nu_\colorF}$.
Then we consider the function $\Choose$ that picks an element in a set: 
$\Choose(X)$ is an element in $X$ for all non-empty sets $X$. 
Then $\unique(\{ o \in c_b \mid c^o_b = \emptyset \}) = \set{\Choose(\Omega^{\nu_\colorL, \nu_\colorF}) \mid \Omega^{\nu_\colorL, \nu_\colorF} \neq \emptyset \text{ and } \nu_\colorL, \nu_\colorF \in \mathbb R}$.
}
per unique cumulative value-so-far vector $(\rho_{\colorL}, \rho_{\colorF})$.
That is, for each marginal belief profile $b$, we retain all Pareto-undominated occupancy states and eliminate any redundant duplicates with identical accumulated returns. Applying this operator to any credible set $c$ yields the filtered set $\tilde{c} \doteq \mathbb{F}(c)$. This pruning step guarantees that all strategically relevant occupancy states are preserved: since the marginal belief captures the information required for computing continuation values, and $\rho_i$ accounts for the accumulated return, the filtered set $\tilde{c}$ supports sound and efficient forward planning within the credible MDP. A formal proof of the losslessness of the filtering operator is provided in Appendix~\ref{appendix:proof:lossless:filtering}.

\subsection{Reward Model and Terminal Sets}

 We  define a reward function $R$ over credible sets that captures the leader’s cumulative return at terminal stages. Intuitively, at intermediate stages $t < \ell$, the leader's policy is still under construction, and uncertainty about the follower's best response prevents assigning a meaningful reward. In contrast, at stage $\ell$, the leader's policy is fixed, and all follower responses consistent with epistemic pessimism are known. This enables a reliable assessment of the total reward.
 Let \( \tilde{F} \doteq \{ \tilde{c}_{\theta_{\colorL,\ell}} \in \tilde{C} \mid |\theta_{\colorL,\ell}| = \ell \} \) denote the terminal filtered credible sets. 
 The reward function $R$ is defined by $R(\tilde{c}) \doteq 0$ if $\tilde c \notin \tilde{F}$, and for all $\tilde{c}\in \tilde{F}$,
     \[R(\tilde{c}) \doteq \max \left\{ \rho_{\colorL}(o) 
     \, \middle|\, o \in \argmax_{o' \in \tilde{c}} \rho_{\colorF}(o') \right\}.\]
The reward is thus zero during transient planning and evaluated only once the leader’s policy is complete, ensuring strategic soundness at termination.

\subsection{Formal Definition}

\begin{definition}
\label{def:credible-mdp}
Given a leader--follower general-sum stochastic game \( M \) with horizon \( \ell \), the \emph{credible MDP} is defined as \( M' \doteq ( \tilde{C}, \Delta_{\colorL}, \tilde{T}, R, \tilde{F}, c_0, \ell ) \) where:
\begin{itemize}
  \item \( \tilde{C} \doteq \{\tilde{c}_{\theta_{\colorL,t}} \mid \theta_{\colorL,t}\in \Theta_{\colorL,t}, t \in \{0,1,\ldots,\ell\}\} \) is the set of filtered credible sets;
  \item \( \Delta_{\colorL} \) is the set of admissible leader decision rules;
  \item \( \tilde{T}\colon (\tilde{c}, \delta_{\colorL}) \mapsto \mathbb{F}(T(\tilde{c}, \delta_{\colorL})) \) is the deterministic transition function;
  \item \( R \colon \tilde{C} \to \mathbb{R} \) is the reward function defined above;
  \item \( \tilde{F} \subseteq \tilde{C} \) is the set of terminal filtered credible sets at stage~\( \ell \);
  \item \( c_0 \in \tilde{C} \) is the initial (filtered) credible set induced by \( s_0 \in M \);
  \item \( \ell \in \mathbb{N} \) is the finite planning horizon.
\end{itemize}
\end{definition}

\begin{example}
Figure~\ref{fig:centipede:game:reformulation} illustrates the structure of the credible MDP prior to belief-based filtering, using the centipede game under a deterministic leader. Nodes represent occupancy states summarised by their realised state and values-so-far.
\begin{figure}[H]
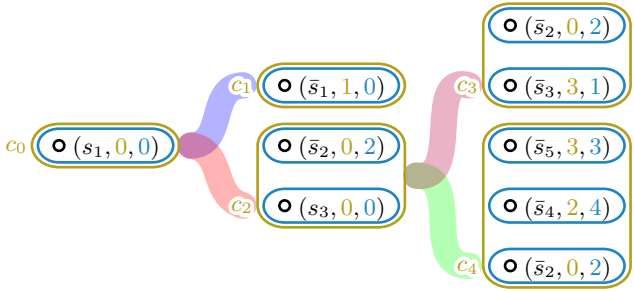

\centering
\begin{beliefspace}[yscale=0.8]
  \state (s0)      at (0,  0)   {$(s_1,\pl1{0},\pl2{0})$};
  \state (s2)      at (3, -1)   {$(s_3,\pl1{0},\pl2{0})$};
  \state (s001)    at (3,  0)   {$(\bar{s}_2,\pl1{0},\pl2{2})$};
  \state (s4)      at (6,  0)   {$(\bar{s}_5,\pl1{3},\pl2{3})$};
  \state (s03)     at (6, -1)   {$(\bar{s}_4,\pl1{2},\pl2{4})$};
  \state (s00001)  at (6, -2)   {$(\bar{s}_2,\pl1{0},\pl2{2})$};
  \state (s01)     at (3,  1)   {$(\bar{s}_1,\pl1{1},\pl2{0})$};
  \state (s02)     at (6,  1)   {$(\bar{s}_3,\pl1{3},\pl2{1})$};
  \state (s0001)   at (6,  2)   {$(\bar{s}_2,\pl1{0},\pl2{2})$};

  \informset<3->[player=1,label=-180, contour] (s0)      {$c_0$};
  \informset<3->[player=1,label=-180, contour] (s01)     {$c_1$};
  \informset<3->[player=1,label=-180, contour] (s2)--(s001) {$c_2$};
  \informset<4->[player=1,label=-180, contour] (s00001)--(s03)--(s4) {$c_4$};
  \informset<3->[player=1,label=-180, contour] (s02)--(s0001) {$c_3$};

  \informset<2->[player=2,label=-180, contour, radius=2.8mm] (s0);
  \informset<2->[player=2,label=-180, contour, radius=2.8mm] (s2);
  \informset<2->[player=2,label=-180, contour, radius=2.8mm] (s00001);
  \informset<2->[player=2,label=-180, contour, radius=2.8mm] (s0001);
  \informset<2->[player=2,label=-180, contour, radius=2.8mm] (s001);
  \informset<2->[player=2,label=-180, contour, radius=2.8mm] (s4);
  \informset<2->[player=2,label=-180, contour, radius=2.8mm] (s03);
  \informset<2->[player=2,label=-180, contour, radius=2.8mm] (s01);
  \informset<2->[player=2,label=-180, contour, radius=2.8mm] (s02);

  \node[coordinate] (frontX0) at ($(s0)+(1.75,0)$) {};
  \node[coordinate] (backX1)  at ($(s01)+(-.55,0)$) {};
  \node[coordinate] (frontX2) at ($(3, -0.5)+(1.75,0)$) {};
  \node[coordinate] (backX2)  at ($(s2)+(-.55,0)$) {};
  \node[coordinate] (backX4)  at ($(s00001)+(-.55,0)$) {};
  \node[coordinate] (backX3)  at ($(s02)+(-.55,0)$) {};

  \begin{pgfonlayer}{background}
    \draw[riverflow, color=blue]   (frontX0) to[out=0,   in=180] (backX1);
    \draw[riverflow, color=red]    (frontX0) to[out=-10, in=160] (backX2);
    \draw[riverflow, color=green]  (frontX2) to[out=0,   in=180] (backX4);
    \draw[riverflow, color=purple] (frontX2) to[out=-10, in=160] (backX3);
  \end{pgfonlayer}
\end{beliefspace}
\caption{Reduction of the centipede game (Figure~\ref{fig:centipede:game}) into the credible MDP prior to filtering. In this example, the filtering step has no effect. Each node is annotated with a tuple \((s_t, \rho_{\colorL}, \rho_{\colorF})\), where \(s_t\) denotes the current state, and \(\rho_{\colorL}, \rho_{\colorF}\) denote the cumulative discounted rewards for the leader and follower, respectively. Determinism ensures that each such tuple corresponds to a unique occupancy state. Brown contours indicate credible sets, each indexed by the leader’s decision-rule history. River arcs depict epistemically pessimistic transitions, capturing the leader’s choice under all admissible follower responses.}\label{fig:centipede:game:reformulation}
\end{figure}
\end{example}

The construction of \( M' \) equips the leader with a forward-searchable state space that faithfully captures all rational follower responses through credible sets. This reduction preserves the strategic essence of the original game while enabling algorithmic tractability through dynamic programming techniques. Importantly, it avoids the pitfalls of earlier formulations that relied on game states or belief states lacking temporal and strategic coherence. Having completed the definition of the reduced model, we now turn to the formal analysis of its computational and structural properties.

\subsection{Lossless Reduction via the Credible MDP}
\label{sec:reduction:correctness}

The concept of a \emph{lossless reduction} formalises when a surrogate model preserves all decision-relevant structure of an LF-GSSG, including planning semantics and equilibrium outcomes, without introducing spurious policies or additional informational power~\cite{sanjari2023isomorphism}.

\begin{definition}[Lossless reduction]
\label{def:lossless-reduction}
A reduction from an LF-GSSG \( M \) to a surrogate model \( M' \) is said to be \emph{lossless} if the following three conditions hold:  
(i) \emph{value preservation} — for every leader policy \( \pi_{\colorL} \in \Pi_{\colorL} \), there exists a policy \( \pi'_{\colorL} \in \Pi'_{\colorL} \) such that \( v^{\pi_{\colorL}}_{M,\colorL,0}(s_0) = v^{\pi'_{\colorL}}_{M',\colorL,0}(s_0) \);  
(ii) \emph{equilibrium correspondence} — there exists a \emph{surjective} map from the set of SSEs in \( M \) onto the set of SSEs in \( M' \), preserving leader and follower payoffs;  
(iii) \emph{information compatibility} — the surrogate model \( M' \) introduces no additional observability or decision flexibility beyond what is available in \( M \).
\end{definition}

The credible MDP \( M' \), defined in Section~\ref{sec:credible:mdp}, satisfies all three conditions above and thus serves as a faithful surrogate model for planning in LF-GSSGs.

\begin{theorem}[Proof in Appendix~\ref{appendix:sec:credible:mdp}]
\label{thm:lossless-reduction}
The credible MDP \( M' \) associated with the LF-GSSG \( M \) constitutes a lossless reduction in the sense of Definition~\ref{def:lossless-reduction}.
\end{theorem}

\begin{proof}[Proof sketch]
Every leader policy \( \pi_{\colorL} = (\delta_{\colorL,0}, \dots, \delta_{\colorL,\ell-1}) \) in \( M \) induces a unique sequence of decision rules, which can be executed identically in \( M' \) to obtain the same trajectory of credible sets and cumulative value. Conversely, each policy in \( M' \) is defined by a sequence of decision rules that correspond to a well-defined leader policy in \( M \), ensuring value preservation. Because the follower reactions embedded in \( M' \) are precisely the admissible best responses under SSE semantics, all equilibrium solutions in \( M' \) correspond to equilibria in \( M \), establishing surjectivity. Finally, both models rely exclusively on private histories and decision rules defined in \( M \), satisfying information compatibility.
\end{proof}

%% file: body/formal-results.tex
\section{Theoretical Results for Credible MDPs}
\label{sec:theoretical:results}

This section establishes the theoretical foundation for solving the original leader--follower stochastic game \( M \) through its lossless reduction to the credible MDP \( M' \).  
We present three main results, each addressing a fundamental aspect: computational complexity, dynamic programming structure, and geometric properties of the value function.  
Together, they provide both theoretical guarantees and structural insights that enable the design of scalable solution algorithms.

We first establish that computing an optimal leader policy is \emph{NP}-hard, even when restricted to Markov policies.
This result highlights the intrinsic computational intractability of synthesising equilibrium strategies, motivating the need for approximate approaches.

\begin{theorem}[Proof in Appendix \ref{appendix:thm:cmdp-reduction:1}]
\label{thm:cmdp-reduction}
    The deterministic memoryless SSE decision problem is NP-hard
    both in the finite-horizon and infinite-horizon cases.
\end{theorem}

To characterise the recursive structure of leader decision-making over credible sets, we first formalise the optimal state- and action-value functions, which together form the basis of Bellman's optimality equations.

\begin{definition}[Optimal state-value function]
    Let \(\pi_{\colorL} \doteq (\delta_{\colorL,0},\ldots,\delta_{\colorL,\ell-1})\) be a leader policy. 
    The leader value function \(v_{\colorL,t}^{\pi_{\colorL}}\colon \tilde{C} \to \mathbb{R}\) at stage \(t\) under  \(\pi_{\colorL}\) is given as follows: for any terminal credible set \(\tilde{c}\in \tilde{F}\), we have \(v_{\colorL}^{\pi_{\colorL}}(\tilde{c}) = R(\tilde{c})\) and for any transient credible set \(\tilde{c} \not\in \tilde{F}\) at stage \(t\),
    \begin{align*}
        v_{\colorL,t}^{\pi_{\colorL}}\colon \tilde{c} &\mapsto 
             R\big(\tilde{T}\big(\tilde{T}(\cdots \tilde{T}(\tilde{c}, \delta_{\colorL,t}) \cdots, \delta_{\colorL,\ell-2}), \delta_{\colorL,\ell-1}\big)\big).        
    \end{align*}
    The optimal (leader) state-value function \(v_{\colorL,t}^*\colon \tilde{C} \to \mathbb{R}\) at stage \(t\) of surrogate model \(M'\) is given: 
    \begin{align*}
        v_{\colorL,t}^*\colon \tilde{c} & \mapsto 
        \textstyle 
        \max_{\pi_{\colorL}\in \Pi_{\colorL}} v_{\colorL,t}^{\pi_{\colorL}}(\tilde{c}).
    \end{align*}
\end{definition}

\begin{definition}[Optimal action-value function]
Let \(\pi_{\colorL} \doteq (\delta_{\colorL,0},\ldots,\delta_{\colorL,\ell-1})\) be a leader policy. 
The optimal leader action-value function \( q_{\colorL,t}^{\pi_{\colorL}} \colon \tilde{C} \times \Delta_{\colorL} \to \mathbb{R} \)  at stage \(t\) 
under~\(\pi_{\colorL}\) is defined, for any transient credible set \( \tilde{c} \in \tilde{C} \) and leader decision rule \( \delta_{\colorL} \), as
\begin{align*}
q_{\colorL,t}^{\pi_{\colorL}}\colon (\tilde{c}, \delta_{\colorL}) &\mapsto v_{\colorL,t}^{\pi_{\colorL}}(\tilde{T}(\tilde{c}, \delta_{\colorL})),
\end{align*}
where \( \tilde{T}(\tilde{c}, \delta_{\colorL}) \) is the successor credible set induced by \( \delta_{\colorL} \) starting in credible set \( \tilde{c} \). The optimal (leader) action-value function \(q_{\colorL,t}^*\colon \tilde{C}\times \Delta_{\colorL} \to \mathbb{R}\) at stage \(t\) of surrogate model \(M'\) is given: 
\begin{align*}
q_{\colorL,t}^*\colon (\tilde{c}, \delta_{\colorL}) & \mapsto 
\textstyle 
\max_{\pi_{\colorL}\in \Pi_{\colorL}} q_{\colorL,t}^{\pi_{\colorL}} (\tilde{c}, \delta_{\colorL}).
\end{align*}

\end{definition}

We now show that these value functions satisfy Bellman optimality equations, formalising how optimal values propagate across decision stages.

\begin{theorem}[Proof in Appendix~\ref{appendix:thm:boe}]
\label{thm:boe}
The optimal leader state-value function \( v_{\colorL,t}^* \colon \tilde{C} \to \mathbb{R}\) at stage \(t<\ell\) satisfies: 
\begin{align*}
v_{\colorL,t}^*\colon \tilde{c} &\mapsto \textstyle 
\max_{\delta_{\colorL} \in \Delta_{\colorL}} v_{\colorL,t+1}^*\big(\tilde{T}(\tilde{c},\delta_{\colorL})\big).
\end{align*}
\end{theorem}

We then establish a key regularity result, showing that the optimal value function is piecewise linear and convex over marginal belief states.  
This property underpins the design of point-based approximation methods by guaranteeing that value representations can be compactly expressed.

Define, for any \( \colorI \in \{\colorL, \colorF\} \), the value function at stage \(t\) induced by the vector set \( \Gamma \) at an occupancy state \( o \) as
\[
v_{\colorI,t}^{\Gamma}(o) \doteq 
\rho_{\colorI}(o)
+
\gamma^t
\sum_{h_{\colorF}} 
\Pr(h_{\colorF} \mid o)
\max_{\alpha \in \bar{\Gamma}(o_{h_{\colorF}})}
\alpha_{\colorI}(o_{h_{\colorF}}),
\]
where \(\bar{\Gamma}(o_{h_{\colorF}}) \doteq \argmax_{\alpha \in \Gamma} \alpha_{\colorF}(o_{h_{\colorF}})\) is a filtered vector set induced by \(\Gamma\)  and \(o_{h_{\colorF}}\colon (s,h_{\colorL}) \mapsto \frac{o(s,h_{\colorL},h_{\colorF}) }{\Pr(h_{\colorF} \mid o)}\) is a conditional occupancy state.

\begin{theorem}[Proof in Appendix~\ref{appendix:thm:convexity}]
\label{thm:convexity}
The optimal leader value function \(v_{\colorL,t}^*\colon \tilde{C}\to\mathbb{R}\) at stage \(t\), solution of Bellman's optimality equations (Theorem \ref{thm:boe}), is uniformly continuous across credible sets. There exists a collection \(\Lambda_t\) of  finite sets \(\Gamma\) of value vectors \((\alpha_{\colorL},\alpha_{\colorF})\) linear across  conditional occupancy states at stage \(t\): 
\[
    v_{\colorL,t}^*\colon \tilde{c} \mapsto \max_{\Gamma\in \Lambda_t}
    \left\{ v_{\colorL,t}^{\Gamma}(o) \,\middle|\,  o \in\argmax_{o'\in \tilde{c}} v_{\colorF,t}^{\Gamma}(o')
    \right\}.
\]
\end{theorem}

Uniform continuity guarantees that small perturbations in credible sets lead to small changes in value, a prerequisite for approximating \(v_{\colorL,t}^*\) using a finite set of representative credible sets. To approximate optimal value function \(v_{\colorL,t}^*\), we introduce a point-based value iteration (PBVI) scheme \cite{pineau2003point,peralezsolving,peralez2024optimally}.  Building on the uniform continuity property, we now formalise how an approximate greedy leader decision rule can be computed efficiently for any given credible set, showing that it reduces to solving a mixed-integer linear program (MILP).

\begin{theorem}[Proof in Appendix~\ref{appendix:proof:thm:pbvi-backup}]
\label{thm:pbvi-backup}
Let \( \Lambda_{t+1} \) be a finite collection of finite sets \(\Gamma_{t+1}\) of linear functions approximating the optimal leader value function \( v_{\colorL}^* \) at stage \( t{+}1 \).  
Then, for any credible set \( \tilde{c} \) at stage \( t \), the greedy leader decision rule \( \delta_{\colorL}^{\tilde{c}} \) maximising the leader action-value is the solution of a MILP (see Appendix \ref{appendix:greedy:leader:dr}, Table \ref{appendix:eq:milp-backup}).
\end{theorem}

Intuitively, the MILP exploits the uniform continuity structure of the action-value functions to jointly select leader decisions and consistent follower best responses, using integer variables to encode max-selection constraints. Next, we leverage these insights to design practical algorithms for scalable equilibrium computation.

%% file: body/point-based-value-iteration.tex
\section{Point-Based Value Iteration}

For computing an SSE in $M$, we adapt point-based value iteration (PBVI, \emph{see} Algorithm~\ref{alg:pbvi-credible-mdp}) \citet{pineau2003point} to solving credible MDP \(M'\) induced by LF-GSSG \(M\).
The algorithm alternates between three
phases: expansion of credible sets alongside with (conditional) occupancy states, improvement of leader value function, and pruning to control complexity. The pseudo-code is given in \Cref{appendix:pbvi}.

The expansion phase (\emph{see} Algorithm~\ref{alg:pbvi-expand})  samples a finite set \(\tilde{C}' \subset \tilde{C}\) of representative credible sets through forward simulations. Starting with the initial state \(s_0\), each credible set \( \tilde{c} \in \tilde{C}' \) is expanded by selecting leader decisions \( \delta_{\colorL} \in \Delta_{\colorL} \), and populating occupancy states in \(\tilde{T}(\tilde{c},\delta_{\colorL})\) by sampling the corresponding  follower response \( \delta_{\colorF} \in \Delta_{\colorF}  \) starting in an occupancy state \(o\in \tilde{c}\), and applying the transition \( \tau(o,\delta_{\colorL}, \delta_{\colorF}) \) to generate new occupancy states. Newly reached occupancy states are retained only if their \(\ell_1\)-distance from the current credible set exceeds a fixed threshold, ensuring coverage. 
The improvement phase (\emph{see} Algorithm~\ref{alg:pbvi-backup}) computes, through point-based backups and mixed-integer linear programming, a finite set \(\Gamma^{\tilde{c}}\) of value vectors \((\alpha_{\colorL},\alpha_{\colorF})\) linear over conditional occupancy states for every credible sample set \(\tilde{c}\in \tilde{C}'\). The  collection  \(\Lambda \doteq \{\Gamma^{\tilde{c}} \mid \tilde{c}\in \tilde{C}'\}\) induces leader value function \(v_{\colorL}\)  which approximates optimal leader value function \(v_{\colorL}^{*}\).
This point-based update procedure iterates until convergence or computational resources are exhausted.  To further improve scalability, we incorporate two pruning strategies. The first (Algorithm~\ref{alg:pruning-vectors}) eliminates dominated sets of value vectors. The second (Algorithm~\ref{alg:pruning-states}) removes noncredible sets. Both procedures introduce approximation error, but yield significant computational savings.

To analyse the approximation quality of point-based value iteration (PBVI) in the credible MDP \(M'\), we define a sample-based approximation of the optimal leader value function over a finite set \(\tilde{C}' \subset \tilde{C}\) of sampled credible sets. Let \(\Lambda_0, \ldots, \Lambda_\ell\) be the collections of approximate value vectors constructed via Theorem~\ref{thm:pbvi-backup}, inducing an approximate leader value function \(v_{\colorL,t} \colon \tilde{C} \to \mathbb{R}\) at stage \(t\), generalised to arbitrary credible sets \(\tilde{c} \in \tilde{C}\) via nearest-neighbour evaluation:
\[
v_{\colorL,t}(\tilde{c}) = 
\min 
\left\{ 
v_{\colorL,t}(\tilde{c}^*) 
\,\middle|\, 
\tilde{c}^* \in \arg\min_{\tilde{c}' \in \tilde{C}'} d_H(\tilde{c}, \tilde{c}')
\right\}.
\]
The distance \(d_H(\tilde{c},\tilde{c}')\) denotes the \emph{Hausdorff distance} between finite subsets \(\tilde{c},\tilde{c}' \subseteq O\) of occupancy states:
\[
d_H(\tilde{c},\tilde{c}') \doteq \max\left\{
\sup_{o \in \tilde{c}} \inf_{\bar{o} \in \tilde{c}'} \| o - \bar{o} \|_1,\ 
\sup_{\bar{o} \in \tilde{c}'} \inf_{o \in \tilde{c}} \| \bar{o} - o \|_1
\right\}.
\]
Let \(v_{\colorL,t}^* \colon \tilde{C} \to \mathbb{R}\) be the optimal leader value function defined in Theorem~\ref{thm:boe}, and let \(m \doteq \max\{\|r_{\colorL}\|_\infty,\ \|r_{\colorF}\|_\infty\}\) denote a uniform bound on instantaneous payoffs. By Theorem~\ref{thm:convexity}, the mapping \(\tilde{c} \mapsto v_{\colorL,t}^*(\tilde{c})\) is uniformly continuous with respect to \(d_H\).
We now express the worst-case approximation error incurred by nearest-neighbour PBVI in terms of the sampling resolution:

\begin{theorem}
\label{thm:pbvi-exploitability}
Let \( \sigma \doteq \sup_{\tilde{c} \in \tilde{C}} \min_{\tilde{c}' \in \tilde{C}'} d_H(\tilde{c}, \tilde{c}') \) be the Hausdorff covering radius of the PBVI sample set \(\tilde{C}'\). Then the exploitability \(\varepsilon\) of the leader policy returned by PBVI satisfies:
\(
\varepsilon \leq  \textstyle \frac{2m\sigma}{(1 - \gamma)^2} \left[ 1 + \ell \gamma^{\ell+1} - (\ell + 1) \gamma^\ell \right]
\).
\end{theorem}

%% file: body/experiments.tex
\section{Experimental Evaluation}
\label{sec:experiments}

We evaluate the empirical performance of our dynamic programming framework for computing strong Stackelberg equilibria (SSEs) in leader–follower general-sum stochastic games (LF-GSSGs). The main objective is to quantify the impact of history-dependent planning on leader value and runtime across diverse strategic settings. Experiments were run on a machine with 64\,GB RAM and a 4.9\,GHz CPU, with a timeout of 10 minutes per run. Linear programs and mixed-integer programs were solved using ILOG CPLEX. Source code will be released publicly.

We compare six planning variants reported in Appendix~\ref{appendix:algorithms}:~\textbf{H} denotes point-based value iteration (PBVI) over credible sets induced by history-dependent policies; \textbf{S} restricts the leader to Markov (state-dependent) policies, which limit the expressiveness of the resulting credible sets compared to history-dependent planning; \textbf{BI} applies backward induction using the full follower value function; \textbf{MY} is a myopic variant that considers only stage-wise follower rewards;~\textbf{LP} solves a linear program per follower policy~\cite{conitzer2006computing}; and \textbf{MILP} encodes the full strategy space in a single mixed-integer program~\cite{3038794.3038831}.
The benchmark suite includes five domains: \textsc{Centipede}, \textsc{Match}, \textsc{Patrolling}, \textsc{Dec-Tiger}, and \textsc{MABC}. These domains span cooperative, competitive, and mixed-motive scenarios, with varying reliance on history. Formal game definitions are deferred to Appendix~\ref{appendix:benchmarks}.
Table~\ref{tab:compact-results} summarises, for each method and horizon length \( \ell \), the resulting SSE leader value (V) and runtime (T in seconds). Additional sensitivity analyses appear in Appendix~\ref{appendix:results}.

\begin{table}[t]
\scriptsize
\setlength{\tabcolsep}{3pt}
\renewcommand{\arraystretch}{1.1}
\scalebox{0.9}{
\begin{tabular}{c
  cc cc cc cc cc cc}
\toprule
\multicolumn{1}{c}{} 
& \multicolumn{2}{c}{\textbf{H}} 
& \multicolumn{2}{c}{\textbf{S}} 
& \multicolumn{2}{c}{\textbf{BI}} 
& \multicolumn{2}{c}{\textbf{MY}} 
& \multicolumn{2}{c}{\textbf{LP}} 
& \multicolumn{2}{c}{\textbf{MILP}} \\[-1mm]
\cmidrule(lr){2-3} \cmidrule(lr){4-5} \cmidrule(lr){6-7} \cmidrule(lr){8-9} \cmidrule(lr){10-11} \cmidrule(lr){12-13}
$\ell$ 
& V & T 
& V & T 
& V & T 
& V & T 
& V & T 
& V & T \\
\midrule
\multicolumn{13}{c}{\textsc{Centipede}  \cite{rosenthal1981games}} \\[-0.5mm]
\hline
1 & \pl1{\bf1} & 0.01 & \pl1{\bf1} & 0.01 & \pl1{\bf1} & 0.01 & \pl1{\bf1} & 0.01 & \pl1{\bf1} & 0.25 & \pl1{\bf1} & 0.14 \\
\myrowcolour
2 & \pl1{\bf1} & 0.01 & \pl1{\bf1} & 0.01 & \pl1{\bf1} & 0.01 & \pl1{\bf1} & 0.01 & \pl1{\bf1} & 20.21 & \pl1{\bf1} & 7.29 \\
3 & \pl1{\bf2} & 0.01 & \pl1{\bf2} & 0.01 & 1 & 0.01 & 1 & 0.01 & --- & --- & --- & --- \\
\myrowcolour
6 & \pl1{\bf4.67} & 3.68 & \pl1{\bf4.67} & 0.11 & 1 & 0.17 & 1 & 0.13 & --- & --- & --- & --- \\
\midrule
\multicolumn{13}{c}{\textsc{Match} \cite{10.5555/2900929.2900938}} \\
\hline
1 & \pl1{\bf-1000} & 0.01 & \pl1{\bf-1000} & 0.01 & \pl1{\bf-1000} & 0.01 & \pl1{\bf-1000} & 0.01 & \pl1{\bf-1000} & 0.2 & \pl1{\bf-1000} & 0.3 \\
\myrowcolour
2 & \pl1{\bf-1000} & 0.01 & \pl1{\bf-1000} & 0.01 & \pl1{\bf-1000} & 0.01 & \pl1{\bf-1000} & 0.01 & --- & --- & --- & --- \\
3 & \pl1{\bf0} & 0.01 & -1000 & 0.01 & -1000 & 0.04 & -1000 & 0.02 & --- & --- & --- & --- \\
\myrowcolour
6 & \pl1{\bf0} & 0.03 & -2000 & 0.03 & -2000 & 0.03 & -2000 & 0.03 & --- & --- & --- & --- \\
\midrule
\multicolumn{13}{c}{\textsc{Dec-Tiger} \cite{nair2003taming}} \\
\hline
1 & \pl1{\bf20} & 0.01 & \pl1{\bf20} & 0.01 & \pl1{\bf20} & 0.01 & \pl1{\bf20} & 0.01 & \pl1{\bf20} & 0.19 & \pl1{\bf20} & 0.22 \\
\myrowcolour
2 & \pl1{\bf40} & 0.01 & \pl1{\bf40} & 0.01 & \pl1{\bf40} & 0.01 & \pl1{\bf40} & 0.01 & \pl1{\bf40} & 1.07 & \pl1{\bf40} & 0.23 \\
3 & \pl1{\bf60} & 0.04 & \pl1{\bf60} & 0.03 & \pl1{\bf60} & 0.01 & \pl1{\bf60} & 0.01 & --- & --- & --- & --- \\
\myrowcolour
6 & \pl1{\bf120} & 0.45 & \pl1{\bf120} & 0.30 & \pl1{\bf120} & 0.20 & \pl1{\bf120} & 0.10 & --- & --- & --- & --- \\
\midrule
\multicolumn{13}{c}{\textsc{MABC} \cite{hansen2004dynamic}} \\
\hline
1 & \pl1{\bf1} & 0.01 & \pl1{\bf1} & 0.01 & \pl1{\bf1} & 0.06 & \pl1{\bf1} & 0.03 & \pl1{\bf1} & 0.20 & \pl1{\bf1} & 0.19 \\
\myrowcolour
2 & \pl1{\bf2} & 0.11 & \pl1{\bf2} & 0.05 & \pl1{\bf2} & 0.09 & \pl1{\bf2} & 0.08 & --- & --- & --- & --- \\
3 & \pl1{\bf2.99} & 0.73 & \pl1{\bf2.99} & 0.49 & \pl1{\bf2.99} & 0.05 & \pl1{\bf2.99} & 0.06 & --- & --- & --- & --- \\
\myrowcolour
6 & \pl1{\bf5.93} & 208.6 & \pl1{\bf5.93} & 0.59 & \pl1{\bf5.93} & 0.09 & \pl1{\bf5.93} & 0.08 & --- & --- & --- & --- \\
\midrule
\multicolumn{13}{c}{\textsc{Patrolling} \cite{10.5555/2900929.2900938}} \\
\hline
1 & \pl1{\bf0} & 0.01 & \pl1{\bf0} & 0.01 & \pl1{\bf0} & 0.01 & \pl1{\bf0} & 0.01 & \pl1{\bf0} & 0.1 & \pl1{\bf0} & 0.09 \\
\myrowcolour
2 & \pl1{\bf0} & 0.04 & \pl1{\bf0} & 0.02 & \pl1{\bf0} & 0.02 & \pl1{\bf0} & 0.02 & --- & --- & --- & --- \\
3 & \pl1{\bf0} & 3.41 & \pl1{\bf0} & 1.1 & \pl1{\bf0} & 0.10 & \pl1{\bf0} & 0.07 & --- & --- & --- & --- \\
\myrowcolour
6 & \pl1{\bf0} & 66.23 & \pl1{\bf0} & 4.0 & \pl1{\bf0} & 0.14 & \pl1{\bf0} & 0.26 & --- & --- & --- & --- \\
\bottomrule
\end{tabular}
}
\caption{SSE leader value (V) and runtime (T in seconds). 
\pl1{\bf Bold leader values} indicate the highest return across planning variants for each configuration.
``---'' indicates timeout. 
}
\label{tab:compact-results}
\end{table}

\textbf{Results and insights.}  
In history-sensitive games such as \textsc{Centipede} and \textsc{Match}, the history-aware PBVI variant~\textbf{H} consistently achieves the highest leader value, highlighting the benefit of planning over credible sets. 
By contrast,~\textbf{S} performs similarly in domains where history plays a limited strategic role (\eg \textsc{Dec-Tiger}, \textsc{MABC}, \textsc{Patrolling}).  
Both \textbf{BI} and \textbf{MY} often fail to induce valid SSEs because backward induction may ignore how a local policy deviation retroactively alters the follower’s interpretation of prior commitment, as formalised in Proposition~\ref{prop:wrong:pi}.
Both \textbf{LP} and \textbf{MILP} become computationally intractable as the horizon increases due to exponential growth in the space of deterministic history-dependent strategies.  
Finally, in fully cooperative or adversarial domains (\eg \textsc{MABC}, \textsc{Dec-Tiger}, \textsc{Patrolling}), 
all methods converge to the same SSE value under cooperative or adversarial symmetry, consistent with DP solutions for MDPs~\cite{Puterman1994} and zero-sum SGs~\cite{shapley_1953_stochastic}.

%% file: body/conclusion.tex
\section{Conclusion}
\label{sec:conclusion}

We presented a dynamic programming framework for computing strong Stackelberg equilibria (SSEs) in leader--follower general-sum stochastic games (LF-GSSGs). Our central contribution is a lossless reduction of the original game \( M \), which is multi-objective and non-Markovian, into a credible Markov decision process \( M' \) defined over an infinite state space of credible sets. Each credible set captures all follower occupancy states consistent with a fixed leader decision-rule history. This construction preserves optimal leader values, admits recursive decomposition via Bellman’s principle, and supports the transfer of point-based value iteration (PBVI) techniques. We formally proved that optimal value functions over credible sets are uniformly continuous, which supports greedy decision-rule extraction through mixed integer linear programming and enables bounded-error approximation via credible-set sampling.
Our PBVI algorithms (\textbf H and \textbf S) provably yield \(\varepsilon\)-optimal leader policies with respect to the exact SSE value, and outperform existing planning baselines on five LF-GSSG benchmarks, especially in mixed-motive settings requiring history-aware strategies.

Beyond these contributions, our reduction offers a foundation for extending exact dynamic programming to more expressive settings such as leader--follower general-sum partially observable stochastic games \cite{chang2021worstcase}. By explicitly modelling belief updates and credible best responses through occupancy-state transitions, our framework addresses a longstanding open question on the applicability of Bellman-style recursion to asymmetric general-sum stochastic games \cite{shapley_1953_stochastic}. As such, it paves the way for unified methods that integrate strategic planning, information asymmetry, and sample-efficient learning.

%% file: appendix/appendix.tex
\onecolumn
\appendix 
\input{appendix/notation_table}
\input{appendix/background}

\section{Credible Markov Decision Processes}
\label{appendix:sec:credible:mdp}

\input{appendix/lossless-reduction}

\section{Theoretical Results for Credible MDPs}
\label{appendix:sec:theoretical:results}
\input{appendix/complexity-reduction}

\input{appendix/theoretical-results}
\input{appendix/greedy-decision-rule}

\section{Point-Based Value Iteration for the Credible MDP}
\label{appendix:pbvi}
\input{appendix/pbvi}
\section{Baseline Algorithms}
\label{appendix:algorithms}
\input{appendix/algorithms}

\section{Benchmark Descriptions}
\label{appendix:benchmarks}
\input{appendix/benchmark}

\section{Detailed Experimental Results}
\label{appendix:results}
\input{appendix/results}

\section{Sensitivity analysis}
\label{appendix:sensitivity}
\input{appendix/sentitivity}

%% file: appendix/notation_table.tex
\section{Table of Notations}
\label{appendix:notations}

\begin{table}[H]
\centering
\renewcommand{\arraystretch}{1.2}
\small
\begin{tabular}{|l|p{8cm}|}
\hline
\multicolumn{2}{|c|}{\textbf{Leader–Follower General-Sum Stochastic Game (\( M \))}} \\ \hline

\multicolumn{2}{|l|}{\textbf{Leader and Follower}} \\ \hline
Leader's / Follower's action set & \( A_{\colorL} \) / \( A_{\colorF} \) \\ \hline
Leader's / Follower's action & \( a_{\colorL} \) / \( a_{\colorF} \) \\ \hline
Joint action & \( a = (a_{\colorL}, a_{\colorF}) \) \\ \hline
Leader's / Follower's reward function & \( r_{\colorL} \) / \( r_{\colorF} \) \\ \hline
Leader's / Follower's policy & \( \pi_{\colorL} \) / \( \pi_{\colorF} \) \\ \hline
Joint policy & \( \pi = (\pi_{\colorL}, \pi_{\colorF}) \) \\ \hline
Leader's / Follower's decision rule at stage \( t \) & \( \delta_{\colorL,t} \), \( \delta_{\colorF,t} \) \\ \hline
Joint decision rule & \( \delta = (\delta_{\colorL}, \delta_{\colorF}) \) \\ \hline

\multicolumn{2}{|l|}{\textbf{Game Structure}} \\ \hline
LF-GSSG model & \( M \) \\ \hline
Number of stages (horizon) & \( \ell \) \\ \hline
Stage index & \( t \in \{0, 1, \dots, \ell\} \) \\ \hline
State space & \( S \) \\ \hline
Environment state & \( s \in S \) \\ \hline
Transition function & \( p(s' \mid s, a_{\colorL}, a_{\colorF}) \) \\ \hline
Discount factor & \( \gamma \in [0,1) \) \\ \hline
Set of best responses to \( \pi_{\colorL} \) & \( \texttt{BR}(\pi_{\colorL}) \) \\ \hline

\multicolumn{2}{|c|}{\textbf{Credible MDP (\( M' \))}} \\ \hline
Leader's / Follower's private environment history & \( h_{\colorL,t} \), \( h_{\colorF,t} \) \\ \hline
Joint private history & \( h = (h_{\colorL}, h_{\colorF}) \) \\ \hline
Decision-rule history up to stage \( t \) & \( \theta_{\colorL,t} \), \( \theta_{\colorF,t} \) \\ \hline
Joint decision-rule history & \( \theta = (\theta_{\colorL}, \theta_{\colorF}) \) \\ \hline
Occupancy state (distribution over joint private histories) & \( o = o_{\theta} \) \\ \hline
Credible set & \( c_{\theta_{\colorL}} = \{ o_{(\theta_{\colorL}, \theta_{\colorF})} \} \) \\ \hline
Filtered credible set & \( \tilde{c} = \mathbb{F}(c_{\theta_{\colorL}}) \) \\ \hline
Set of filtered credible sets & \( \tilde{C} \) \\ \hline
Initial credible set & \( c_0 \) \\ \hline
Filtered transition function & \( \tilde{T}(\tilde{c}, \delta_{\colorL}) \) \\ \hline
Deterministic transition on occupancy states & \( \tau(o, \delta_{\colorL}, \delta_{\colorF}) \) \\ \hline
Expected cumulative reward (leader / follower) & \( \rho_{\colorL}(o) \), \( \rho_{\colorF}(o) \) \\ \hline
Marginal belief (over states) & \( b_o \in \Delta(S) \) \\ \hline
Reward function on \( \tilde{C} \) & \( R(\tilde{c}) \) \\ \hline

\multicolumn{2}{|c|}{\textbf{Value Functions}} \\ \hline
Leader / Follower value function at stage \( t \) under policy \( \pi \) & \( v^{\pi}_{\colorL,t} \), \( v^{\pi}_{\colorF,t} \) \\ \hline
Leader / Follower \( q \)-value at stage \( t \) & \( q^{\pi}_{\colorL,t} \), \( q^{\pi}_{\colorF,t} \) \\ \hline

\multicolumn{2}{|c|}{\textbf{Planning Objects}} \\ \hline
Set of \(\alpha\)-vector pairs & \( \Gamma =\{ (\alpha^{(k)}_\colorL,\alpha^{(k)}_\colorF) \}_{k=1}^K \) \\ \hline
Collection of such sets & \( \Lambda = \{ \Gamma^{(j)} \}_{j=1}^J \) \\ \hline
\end{tabular}
\caption{Summary of notations used throughout the paper. All terms are formally defined in the main text.}
\label{table:notations}
\end{table}

%% file: appendix/background.tex
\section{Naive Dynamic Programming: Why It Fails}

A naive dynamic programming approach may fail to compute strong Stackelberg equilibria (SSEs), even when applying seemingly rational local improvement steps. The following proposition illustrates this failure mode.

\begin{proposition}
\label{appendix:prop:wrong:pi}
Let \( \pi_{\colorL} \) be a Markov leader policy. Construct a new policy \( \pi'_{\colorL} \) by selecting, at each state \( s \in S \) and stage \( t < \ell \), a decision rule \( \delta'_{\colorL,t}(s) \) solving:
\[
\begin{aligned}
\max_{\delta_{\colorL} \in \Delta(A_{\colorL})} \; &\max_{a_{\colorF} \in A_{\colorF}} \;
\mathbb{E}_{a_{\colorL} \sim \delta_{\colorL}} \left[ q^{\pi_{\colorL}}_{\colorL,t}(s,s,s, a_{\colorL}, a_{\colorF}) \right] \\
\text{s.t.} \quad &
\mathbb{E}_{a_{\colorL} \sim \delta_{\colorL}} \left[ q^{\pi_{\colorL}}_{\colorF,t}(s,s,s, a_{\colorL}, a_{\colorF}) \right]
= v^{\pi_{\colorL}}_{\colorF,t}(s,s,s).
\end{aligned}
\]
Even if the solution satisfies \( \pi'_\colorL = \pi_\colorL\), the resulting policy \( \pi'_{\colorL} \) may fail to induce a strong Stackelberg equilibrium.
\end{proposition}

\begin{proof}
We construct a counterexample based on the finite-horizon centipede game (Figure~\ref{fig:centipede:game}). Consider a Markov leader policy \( \pi_{\colorL} \) that selects action \( \pl1{\mathtt{take}} \) in both states \( s_1 \) and \( s_3 \), inducing a follower best response \( \pi_{\colorF} \) in which the follower plays \( \pl2{\mathtt{take}} \) in both \( s_2 \) and \( s_4 \). The resulting joint policy \( \pi = (\pi_{\colorL}, \pi_{\colorF}) \) yields the following value vectors:
\[
\begin{aligned}
v_0^{\pi}(s_1) &= (1, 0), \\
v_1^{\pi}(s_2) &= (0, 2), \\
v_2^{\pi}(s_3) &= (3, 1), \\
v_3^{\pi}(s_4) &= (2, 4).
\end{aligned}
\]

Suppose we now apply the local improvement rule defined in the proposition to \( \pi_{\colorL} \). At both \( s_1 \) and \( s_3 \), the decision rule \( \delta_{\colorL,t}(s) \) corresponding to \( \pl1{\mathtt{take}} \) already satisfies the local optimisation objective, and hence the improvement step returns the same policy: \( \pi'_{\colorL} = \pi_{\colorL} \).

However, consider an alternative leader policy \( \pi''_{\colorL} \) that plays \( \pl1{\mathtt{continue}} \) in both \( s_1 \) and \( s_3 \). The follower best response under \( \pi''_{\colorL} \) is to play \( \pl2{\mathtt{continue}} \) in \( s_2 \) and \( \pl2{\mathtt{take}} \) in \( s_4 \), leading to the joint policy \( \pi'' = (\pi''_{\colorL}, \pi''_{\colorF}) \) with:
\[
v_0^{\pi''}(s_1) = (2, 4).
\]
This dominates the outcome under \( \pi'_{\colorL} \), which yielded only \( (1,0) \), thereby contradicting the assumption that \( \pi'_{\colorL} \) induces an SSE.

Hence, the local improvement rule may yield a fixed point that fails to satisfy the global optimality required by the strong Stackelberg equilibrium.
\end{proof}

%% file: appendix/lossless-reduction.tex
\subsection{Lossless Reduction} \label{appendix:proof:thm:lossless-reduction} 
\begin{lemma} \label{appendix:lem:filtered:equality} For every Markov leader policy \( \pi_{\colorL} \in \Pi_{\colorL} \) in LF-GSSG \( M \), let \( c_{\pi_{\colorL}} \) be the credible set induced by \( \pi_{\colorL} \) starting from the initial state \( s_0 \). Then, \(R(c_{\pi_{\colorL}}) = R\big(\mathbb{F}(c_{\pi_{\colorL}})\big)\), where \( \mathbb{F}(\cdot) \) denotes the filter operator. 
\end{lemma} 
\begin{proof} Let \(c\) be any credible set. Define $B_c \doteq \{ b_o \mid o \in c \}$ and, for each $b \in B_c$, define the belief class 
\[c_b \doteq \{ o \in c \mid b_o = b \}.\] 
Within each $c_b$, for any occupancy state \(o\in c_b\), define  \[ c^o_b\doteq \{ o' \in c_b \mid \forall i \in I, \rho_i(o') \geq \rho_i(o), \; \exists j \in I, \rho_j(o') > \rho_j(o) \}.\]  We then define the filtered set as \[\mathbb{F}(c) \doteq \cup_{b \in B_c} \unique(\{ o \in c_b \mid c^o_b = \emptyset \}),\] where $\unique(\cdot)$ retains exactly one representative per unique cumulative value-so-far vector $(\rho_{\colorL}, \rho_{\colorF})$. We proceed by contradiction. Assume the statement does not hold, that is, \[ R(c_{\pi_{\colorL}}) \neq R\big(\mathbb{F}(c_{\pi_{\colorL}})\big), \] which implies that there exists an occupancy state \( o \in c_{\pi_{\colorL}} \) whose removal by the filter operator \( \mathbb{F} \) changes the value: \[ R(c_{\pi_{\colorL}}) \neq R\big(c_{\pi_{\colorL}} \setminus \{o\}\big). \]  Let \( b \) be the marginal belief state associated with \( o \), and let \( c_b \subseteq c_{\pi_{\colorL}} \) be the subset of occupancy states in \( c_{\pi_{\colorL}} \) sharing the marginal belief \( b \).  For the removal of \( o \) to affect the value, it must satisfy: \[ \rho_{\colorF}(o) > \max_{o' \in \mathbb{F}(c_{\pi_{\colorL}})} \rho_{\colorF}(o') \quad \text{or} \quad \rho_{\colorF}(o) = \max_{o' \in \mathbb{F}(c_{\pi_{\colorL}})} \rho_{\colorF}(o') \wedge \rho_{\colorL}(o) > R\big(\mathbb{F}(c_{\pi_{\colorL}})\big). \] where \( \rho_{\colorL}(o) \) and \( \rho_{\colorF}(o) \) denote the expected cumulative rewards for the leader and follower, respectively, under occupancy state \( o \).  Since \( o \) was removed by \( \mathbb{F} \), there exists \( o' \in c^o_b \) such that \[ \forall i \in I, \rho_i(o') \geq \rho_i(o), \; \exists j \in I, \rho_j(o') > \rho_j(o).\]  Consequently, we have \[ R\big(\mathbb{F}(c_{\pi_{\colorL}})\big) > \rho_{\colorL}(o), \] which contradicts the earlier assumption that \[ \rho_{\colorF}(o) > \max_{o' \in \mathbb{F}(c_{\pi_{\colorL}})} \rho_{\colorF}(o') \quad \text{or} \quad \rho_{\colorF}(o) = \max_{o' \in \mathbb{F}(c_{\pi_{\colorL}})} \rho_{\colorF}(o') \wedge \rho_{\colorL}(o) > R\big(\mathbb{F}(c_{\pi_{\colorL}})\big). \]  Therefore, the assumption is false, and we conclude that \[ R(c_{\pi_{\colorL}}) = R\big(\mathbb{F}(c_{\pi_{\colorL}})\big). \] \end{proof} 

\begin{theorem}
\label{appendix:thm:lossless-reduction}
The credible MDP \( M' \) associated with the LF-GSSG \( M \) constitutes a lossless reduction in the sense of Definition~\ref{def:lossless-reduction}.
\end{theorem}

\begin{proof}
We verify each of the three criteria in Definition~\ref{def:lossless-reduction}.

\textbf{(i) Value preservation.}  
Let \( \pi_{\colorL} = (\delta_{\colorL,0}, \ldots, \delta_{\colorL,\ell-1}) \in \Pi_{\colorL} \) be a leader policy in \( M \).  
We define the corresponding policy \( \pi'_{\colorL} \in \Pi'_{\colorL} \) in \( M' \) as the sequence of decision rules \( \pi'_{\colorL} \doteq (\delta_{\colorL,0}, \ldots, \delta_{\colorL,\ell-1}) \), applied deterministically at each stage.  
Since the transitions in \( M' \) are deterministic and coincide with the application of decision rules over credible sets, this correspondence preserves execution paths.

The value in \( M' \) starting from initial credible set \( c_0 = \{s_0\} \) is
\[
v_{M',\colorL,0}^{\pi'_{\colorL}}(c_0) 
= R\big( \mathbb{F}(c_{\pi'_{\colorL}}) \big)
= R(c_{\pi'_{\colorL}}),
\]
where the second equality follows from Lemma~\ref{appendix:lem:filtered:equality}.

By definition of the terminal reward function, we then have:
\[
v_{M',\colorL,0}^{\pi'_{\colorL}}(c_0)
=
\max\left\{
\rho_{\colorL}(o) \;\middle|\;
o \in  \argmax_{o' \in c_{\pi'_{\colorL}}} \rho_{\colorF}(o')
\right\}.
\]
Each \( o \in c_{\pi'_{\colorL}} \) corresponds to an occupancy state induced by some follower response \( \pi_{\colorF} \in \mathtt{BR}(\pi_{\colorL}) \). Since \( \rho_{\colorI}(o) = v_{M,\colorI,0}^{\pi_{\colorL},\pi_{\colorF}}(s_0) \) by construction of occupancy states, this implies
\[
v_{M',\colorL,0}^{\pi'_{\colorL}}(c_0) = v_{M,\colorL,0}^{\pi_{\colorL}}(s_0).
\]
Conversely, given a policy \( \pi'_{\colorL} \in \Pi'_{\colorL} \), the sequence of decision rules defines a leader policy \( \pi_{\colorL} \in \Pi_{\colorL} \) in \( M \) that induces the same sequence of occupancy sets and the same terminal credible set \( c_{\pi_{\colorL}} = c_{\pi'_{\colorL}} \). Hence,
\[
v_{M,\colorL,0}^{\pi_{\colorL}}(s_0) = v_{M',\colorL,0}^{\pi'_{\colorL}}(c_0).
\]
This establishes value preservation in both directions.

\textbf{(ii) Equilibrium correspondence.}  
Let \( \pi = (\pi_{\colorL}, \pi_{\colorF}) \) be a strong Stackelberg equilibrium (SSE) in \( M \). Then \( \pi_{\colorF} \in \mathtt{BR}(\pi_{\colorL}) \), and all best responses induce occupancy states \( o \in c_{\pi_{\colorL}} \).  
As above, the corresponding leader policy \( \pi'_{\colorL} \) in \( M' \) induces the same credible set \( c_{\pi'_{\colorL}} = c_{\pi_{\colorL}} \), and the follower best response in \( M' \) coincides with selecting the occupancy state \( o \in c_{\pi_{\colorL}} \) maximising \( \rho_{\colorF} \).  
Hence, the value of the SSE is preserved, and the mapping \( \pi \mapsto \pi'_{\colorL} \) defines a surjection from SSEs in \( M \) to optimal leader policies in \( M' \). Multiple SSEs in \( M \) can map to the same policy in \( M' \) if they induce the same decision-rule sequence.

\textbf{(iii) Information compatibility.}  
The leader policies in both \( M \) and \( M' \) are defined over the same private state–action histories, and decision rules are chosen with respect to credible sets indexed by these histories.  
The surrogate model \( M' \) does not introduce any additional observability or decision-making capability. All transitions and outcomes are grounded in the environment and policy structure of the original model \( M \).

\medskip
Therefore, \( M' \) satisfies all three criteria of Definition~\ref{def:lossless-reduction} and is a lossless reduction of \( M \).
\end{proof}

%% file: appendix/complexity-reduction.tex
\subsection{Complexity Results}


We consider the following problem. 
\textsc{Infinite-Horizon Exact Value: }
Given a MDP $M$, and a rational number $v$, decide whether there is a
memoryless deterministic policy $\pi : S \rightarrow A$ whose value $V^\pi$ (from the initial state) is exactly $v$.
We also define \textsc{Finite Horizon Exact Value} which further accepts a horizon $\ell$.

\begin{theorem}
    \label{thm:exact-value-policy-np-hard-finite-horizon}
    The \textsc{Finite Horizon Exact Value} is NP-complete.
\end{theorem}
\begin{proof}
    By reduction from Hamiltonian path. 
    Given directed graph $G=(V,E)$, we build an MDP $M$ whose states are the vertices of $G$,
    and whose transitions are the edges of $G$. The initial state $s_0$ of $M$ is chosen to be an arbitrary vertex of $G$.
    Define $B = |V|+1$.
    We associate an index $0\leq i \leq |V|-1$ to each state of $M$.
    The reward when leaving the state with index $i$ is $B^{i-1}$.
    The discount factor is $\gamma = 1$. We set $v = 1 + B + ... B^{|V|-1}$.
    
    The graph $G$ admits a Hamiltonian path iff $(M, v, \ell=|V|)$ is a positive instance of the finite horizon exact value problem.
    
    \fbox{$\Rightarrow$} If there is a Hamiltonian path, consider policy $\pi$ that follows that path. Its value is $v = 1 + B + ... B^{|V|-1}$ because we reach exactly each state once.
    
    \fbox{$\Leftarrow$} Consider a policy $\pi$ whose value is $v = 1 + B + ... B^{|V|-1}$.
    Let $x_i$ be the number of times the state with index $i$ is visited.
    The value can also be written $v = x_1 + x_2 B + \times x_{|V|}B^{|V|-1}$.
    As $x_i < B$, the two writings of number $v$ in base $B$ are unique, thus $x_i = 1$ for all $i$.
    So $\pi$ follows a Hamiltonian path.
\end{proof}

The infinite-horizon case is also NP-hard but for this proof, we restrict to memoryless deterministic policies.

\begin{theorem}
    \label{thm:exact-value-policy-np-hard}
    \textsc{Infinite Horizon Exact Value} is NP-complete.
\end{theorem}

\begin{proof}
   The membership is proven as follows. Guess $\pi$. Compute the exact value in the induced Markov chain in polynomial time. Check that it is equal to $v$.

    We now focus on the NP-hardness, proven by reduction from Hamiltonian cycle. Let $G = (V, E)$ be a directed graph. We construct an MDP $M$ as follows. The state set is $V$, and the transitions in $M$ are given by the edges of $G$. Fix any vertex $s_0$ in $V$ to be the initial state of $M$. The reward is 1 when going out $s_0$ and 0 everywhere else. We fix the discount factor $\gamma = 1/2$.
    Let $v = \frac 1 {1 - \gamma^{|V|}}$.

    The instance $(M, v)$ is computable in poly-time from $G$. It remains to prove that $G$ has a Hamiltonian cycle iff $(M, v)$ is a positive instance.

    \fbox{$\Rightarrow$} Suppose $G$ has a Hamiltonian cycle. Consider $\pi$ to be the strategy that goes through the cycle forever. Every $|V|$ steps, we reach $s_0$ and the reward is 1. The gain is

    $$V^\pi(s_0) = 1 + \gamma^{|V|} + \gamma^{2|V|} + ... = v.$$

So $M$ has a strategy whose value is exactly $v$. So $(M, v)$ is a positive instance.

    \fbox{$\Leftarrow$} Suppose that $(M, v)$ is a positive instance, that is, $M$ has a strategy $\pi$ whose value is exactly $v$. Because $M$ is deterministic, a single infinite path in $G$ is generated under $\pi$.

    Suppose that under $\pi$, $s_0$ is never reached again after the first step. Then $V^\pi(s_0) = 1$, which contradicts $V^\pi(s_0) = v$.
    So $s_0$ must be reached again.
    Because $\pi$ is deterministic and memoryless,
    it induces a simple cycle that contains $s_0$, which is repeated indefinitely.
    Let $n$ be the number of states before reaching $s_0$ again, that is, the length of the simple cycle. We have 

    $$V^\pi(s_0) = 1 + \gamma^{n} + \gamma^{2n} + ... = \frac{1}{1 - \gamma^{n}}.$$

    As $V^\pi(s_0) = v$, we have $n = |V|$. Thus $\pi$ induces a Hamiltonian cycle.
\end{proof}

A \emph{constrained MDP (CMDP)} (we consider the single-cost version of CMDP given in
\cite{feinberg2000constrained})
is a tuple 
\[ (S, A, p, r, c, \nu, \gamma, s_0) \] where:
\begin{itemize}
  \item \( S \): Finite state space, \( s_0 \in S \) initial state
  \item \( A \): Finite action space
  \item \( p(s' \mid s, a) \): Transition probabilities
  \item \( r(s, a) \in \mathbb{Q} \): Reward function
  \item \( c(s, a) \in \mathbb{Q}_{\geq 0} \): Cost function 
  \item \( \gamma \in (0,1) \): Discount factor
  \item $\nu\in \mathbb{Q}$: Cost threshold
\end{itemize}
The \emph{infinite-horizon CMDP feasibility problem} is the the decision problem that, given a CMDP, and rational $R$, asks whether there exists a policy \( \pi \) satisfying:
\begin{align}
  \textstyle
  \mathbb{E}^{\pi} \left[ \sum_{t=0}^{\infty} \gamma^t r(s_t, a_t) \right] &\geq R, \label{eq:cmdp-reward} \\
  \textstyle
  \mathbb{E}^{\pi} \left[ \sum_{t=0}^{\infty} \gamma^t c(s_t, a_t) \right] &\leq \nu.  \label{eq:cmdp-constraint}
\end{align}

In words, we ask for the existence of a policy $\pi$ whose expected gain is greater than or equal to $R$, while expected cost is bounded by $\nu$.

We similarly define the \emph{finite-horizon CMDP feasibility problem}
by considering a horizon $\ell$. In the sequel, we ask for the existence of a \emph{deterministic} policy.

\begin{theorem}
    \label{thm:cmdp-feasbility-np-hard}
    The infinite-horizon and finite horizon CMDP feasability problems are NP-hard.
\end{theorem}
\begin{proof}
    Reduction from \textsc{Infinite Horizon Exact value policy}
    of Theorem \ref{thm:exact-value-policy-np-hard}. Consider an instance $(M, v)$ of \textsc{Infinite Horizon Exact value policy}. 
    The CMDP $M'$ is $M$, and the costs are identical to rewards: $c(s, a) = r(s, a).$
    
    Then, there is a policy in $M$ whose value is $v$ iff 
    there is a policy in $M'$ whose value is $\geq v$ and whose 
    cost is $\leq v$.

    For the finite horizon case, we similarly reduce from 
    the \textsc{Finite Horizon Exact value policy} using Theorem \ref{thm:exact-value-policy-np-hard-finite-horizon}.
\end{proof}

The \emph{``two-bound'' SSE decision problem} is the following: 
Given rationals \( R, \nu' \),
and LF-GSSG \( M \), determine if there is a 
joint policy \( \pi = (\pi_{\colorL}, \pi_{\colorF}) \) 
such that  
\(v_{\colorL}^{\pi}(s_0) \geq R\)
and \(v_{\colorF}^{\pi}(s_0) \geq \nu'\).
The \emph{memoryless deterministic} SSE decision problem
is the restriction of this problem to leader policies that are
memoryless and deterministic.

\begin{theorem}
    \label{appendix:thm:cmdp-reduction:1}
    The memoryless deterministic ``two-bound'' SSE decision problem is NP-hard
    both in the finite horizon and infinite horizon cases.
\end{theorem}
\begin{proof}
    We prove NP-hardness via a polynomial-time reduction from the 
    CMDP feasibility problems (Theorem~\ref{thm:cmdp-feasbility-np-hard})
            
    From a CMDP feasibility instance $M$, and bounds $R,\nu$, we construct an infinite-horizon LF-GSSG $M'$ with bounds $R, \nu'=-\nu$
    on the same state space as $M$.
    The leader's actions are identical to actions available in $M$,
    and the follower has a single dummy action which is ignored in all transitions. Furthermore, we have the following rewards:
    $r_{\colorL}(s,a) = r(s,a)$, and $r_{\colorF}(s,a) = -c(s,a)$
    for all $s,a$.
    
    Assume that \eqref{eq:cmdp-reward}-\eqref{eq:cmdp-constraint} hold
    for policy $\pi$. 
    Let $\pi_{\colorL} = \pi$, and $\pi_{\colorF}$ be an arbitrary policy. Then 
    \(v_{\colorL}^{\pi}(s_0) \geq R\) since the value for the leader is identical to the value in $M$, and 
    \(v_{\colorF}^{\pi}(s_0) \geq \nu'\) since the follower's value is minus the expected cost in $M$.
    
    Conversely, if 
    \(v_{\colorL}^{\pi}(s_0) \geq R\) and 
    \(v_{\colorF}^{\pi}(s_0) \geq \nu'\), then 
    by letting $\pi=\pi_{\colorL}$, we satisfy
    \eqref{eq:cmdp-reward}-\eqref{eq:cmdp-constraint}.

    The proof for the finite horizon case is analogous using
    Theorem~\ref{thm:cmdp-feasbility-np-hard}.
\end{proof}

The \emph{infinite-horizon SSE decision problem} is the following: 
Given rational \( R \), and LF-GSSG \( M \), determine if there is a 
joint policy \( \pi = (\pi_{\colorL}, \pi_{\colorF}) \) such that 
$\pi_{\colorF} \in \mathtt{BR}(\pi_{\colorL})$ and
\(v_{\colorL}^{\pi}(s_0) \geq R\).
We also define the \emph{finite-horizon SSE decision problem} which further takes
the horizon $\ell$ as input.

\begin{theorem}
    \label{appendix:thm:cmdp-reduction:2}
    The deterministic memoryless SSE decision problem is NP-hard
    both in the finite-horizon and infinite-horizon cases.
\end{theorem}
\begin{proof}
     We prove NP-hardness via a polynomial-time reduction from the finite-horizon 
    CMDP feasibility problem (Theorem~\ref{thm:cmdp-feasbility-np-hard}).

\begin{figure}
    \centering
    \begin{tikzpicture}
    \node[blue] (debut) {$s'_0$};
    \node[black] (reject) at (-2, -3) {\footnotesize 
    \begin{tabular}{l}
    leader gets a reward of $R-1$ \\
    follower gets 0
    \end{tabular}};
    \node[orange, draw, minimum height=2cm, minimum width=3cm] (game) at (2, -3) {simulation of $M$};
    \draw (debut) edge[-latex, blue] node[right] {accept} (game);
    \draw (debut) edge[-latex, blue] node[left] {refuse} (reject);
    \end{tikzpicture}
    \caption{The LF-GSSG $M'$ constructed from the CMDP feasibility instance $(M, R,  \nu)$ works as follows. In the initial state $s_0'$ the follower starts by either refusing or accepting. If the follower refuses, the leader gets a reward of $R-1$ while the follower gets 0. If the follower accepts, a simulation of $M$ starts where the leader plays. The leader receives the reward in $M$ while the reward of $F$ handles the constraint.}
    \label{figure:reductioncmdptosse}
\end{figure}
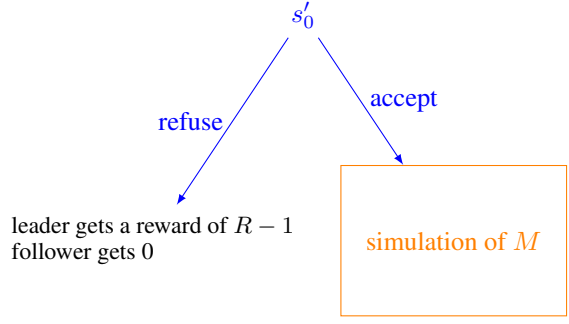

    Let us start with the finite horizon case.
    From a finite-horizon deterministic CMDP instance $(M, R, \nu, \ell)$ we construct the following finite-horizon deterministic SSE-instance $(M', R, \ell)$. The LF-GSSG $M'$ depicted in Figure~\ref{figure:reductioncmdptosse} is constructed as follows:
    \begin{itemize}
        \item We introduce a fresh initial state $s_0'$ (distinct from the states of $M$);
        \item In $s_0'$, the follower decides either to "accept" to simulate a play in $M$, or to "refuse".
        \item If the follower refuses, the leader obtains a reward of $R-1$ and the follower a reward of 0 and the LF-GSSG stops.
        \item If the follower accepts, we enter the initial state of a copy of $M$, in which only the leader determines the actions to be taken (actions played by the follower are not relevant), and where the reward of the leader mimics the reward in $M$, while the reward of $F$ takes the constraint into account as follows. At state $s$, when the leader plays $a$, the follower receives the reward:
        $$r_F(s, a) := \frac \nu \ell - c(s,a)$$
        where $c(s, a)$ is the cost of executing $a$ from state $s$ in $M$.
    \end{itemize}
    The parameter $R$ and $\ell$ are the same than in $(M, R, \nu, \ell)$.
    The construction of $(M', R, \ell)$ from $(M, R, \nu, \ell)$ can be performed in poly-time.

    It remains to prove that there is a policy $\pi$ in the CMDP $M$ such that the gain is $\geq R$ and the cost is $\leq \nu$ iff there is $\pi'$ a SSE in $M'$ with the leader-value $\geq R$.
    
    \fbox{$\Rightarrow$}
    Suppose there is such a $\pi$. We set $\pi'$ so that $\pi_L$ mimics $\pi$ in the $M$-part of $M'$, and $\pi_F$ accepts in $s_0$. The leader-value for $\pi'$ is $\geq R$. Meanwhile, as the cost is $\leq \nu$ in $M$ with $\pi$, the follower value is $\geq 0$. So $\pi_{\colorF}$ is a best response, and $\pi'$ is indeed a SSE.
    
    \fbox{$\Leftarrow$} Suppose there is a SSE $\pi'$ in $M'$ such that the leader-value is $\geq R$ with $\pi'$. Let $\pi$ the policy that mimics the leader-policy $\pi_L$ in the $M$-part. As the leader-value is $\geq R$, it means that the follower has accepted in $s_0$. 
    Thus, as the follower plays a best response, the follower-value is as good as the "refuse" choice, \textit{i.e.} greater than or equal to 0. 
    It follows that the total cost under $\pi$ in $M$ 
    is less than or equal to $\nu$.

    \bigskip
    For the infinite-horizon case, we consider the same reduction 
    with the following differences.
    In the initial state $s_0'$, if the follower refuses, the players receive $R-1$ and $0$ respectively as before,
    and we extend the game with a self-loop with reward $0$ for both players.
    In the copy of $M$, given state-action pair $(s,a)$, the leader receives the reward $r_{\colorL}(s,a) = r(s,a)/\gamma$,
    while the follower receives $r_{\colorF}(s,a) = \nu(1-\gamma)/\gamma - c(s,a)/\gamma$. Note that we divide the rewards and costs by $\gamma$
    to account for the additional step we introduced due to $s_0'$.

    Consider a joint policy $\pi$.
    If the follower refuses, then the leader receives $R-1$ and the follower $0$. If the follower accepts, then the leader's policy is played in the copy of $M$, and the leader receives
    \begin{align*}
        \mathbb{E}^{\pi} \left[ \sum_{t=1}^{\infty} \gamma^t r(s_t, a_t)/\gamma \right]
        = \mathbb{E}^{\pi} \left[ \sum_{t=0}^{\infty} \gamma^t r(s_t, a_t) \right].
    \end{align*}
    
    The follower receives
    \begin{align*}
      &\mathbb{E}^{\pi} \left[ \sum_{t=1}^{\infty} \gamma^t (\nu(1-\gamma)/\gamma - c(s_t, a_t)/\gamma) \right]\\
      & = \mathbb{E}^{\pi} \left[ \sum_{t=0}^{\infty} \gamma^t (\nu(1-\gamma)
        - c(s_t, a_t)) \right]\\
        & = \nu - \mathbb{E}^{\pi} \left[ \sum_{t=0}^{\infty} \gamma^t c(s_t, a_t) \right].
    \end{align*}

    The correctness of the reduction is then identical to the proof in the finite-horizon case.
\end{proof}

%% file: appendix/theoretical-results.tex
\subsection{Proof of Losslessness of the Filtering Operator}
\label{appendix:proof:lossless:filtering}

\begin{theorem}[Losslessness of the Filtering Operator]
\label{thm:lossless-filtering}
Let \( M \) be a finite-horizon leader--follower general-sum stochastic game (LF-GSSG) with horizon \( \ell \), and let \( \mathbb{F} \) denote the belief-based filtering operator.  
Fix any Markov leader policy \( \pi_{\colorL} \in \Pi_{\colorL} \), and let \( (c_0, \dots, c_\ell) \) denote the sequence of credible sets generated by forward simulation under \( \pi_{\colorL} \).  
Define the filtered sequence \( (\tilde{c}_0, \dots, \tilde{c}_\ell) \) recursively by:
\begin{itemize}
    \item \( \tilde{c}_0 \doteq \mathbb{F}(c_0) \),
    \item \( \tilde{c}_{t+1} \doteq \mathbb{F}(T(\tilde{c}_t, \delta_{\colorL,t})) \) for all \( t < \ell \),
\end{itemize}
where \( \delta_{\colorL,t} \) is the decision rule applied by \( \pi_{\colorL} \) at stage \( t \), and \( T \) denotes the transition function of the credible MDP.  
Then, the final cumulative leader reward is invariant under recursive filtering:
\[
R(c_\ell) = R(\tilde{c}_\ell).
\]
\end{theorem}

\begin{proof}
We proceed by induction on \( t \in \{0, \dots, \ell\} \), maintaining the following invariant.

\begin{quote}
\textbf{Filtered Reachability Preservation (FRP).} 
At every stage \( t \), the filtered set \( \tilde{c}_t \) contains, for each marginal belief \( b \), at least one occupancy state per value-so-far vector \( (\rho_{\colorL}, \rho_{\colorF}) \) that is reachable from \( c_0 \) under \( \pi_{\colorL} \), and that may lead to a final occupancy state attaining the optimal cumulative reward \( R(c_\ell) \).
\end{quote}

\textit{Base case (\( t = 0 \)).} The initial credible set \( c_0 \) is a singleton, and the filtering operator does not remove it: \( \tilde{c}_0 = \mathbb{F}(c_0) = c_0 \). The invariant holds trivially.

\textit{Inductive step.} Assume the invariant holds at stage \( t \). We must show it holds at \( t+1 \).  
Let \( o \in \tilde{c}_t \) be any occupancy state preserved by filtering. By the inductive hypothesis, every marginal belief and value-so-far vector relevant for optimal return is represented in \( \tilde{c}_t \).  
Since the leader policy is Markov and the environment is fully observable and Markovian, the successor occupancy states resulting from applying \( T(o, \delta_{\colorL,t}) \) depend only on the marginal belief \( b \) and the joint action, not on local histories.

Moreover, the transition function \( T \) deterministically propagates marginal beliefs and extends cumulative rewards additively. Thus, the set \( c_{t+1} \doteq T(\tilde{c}_t, \delta_{\colorL,t}) \) includes all marginal belief and cumulative reward pairs relevant for reaching optimal terminal values.  
Applying \( \mathbb{F} \) to \( c_{t+1} \) preserves one representative per such pair, hence the invariant holds at \( t+1 \).

\textit{Terminal step.} At stage \( t = \ell \), the credible MDP accumulates the total leader reward. By the inductive invariant, the final filtered set \( \tilde{c}_\ell \) contains all marginal belief and reward pairs relevant for achieving \( R(c_\ell) \).  
By Lemma~\ref{appendix:lem:filtered:equality}, which guarantees that filtering does not reduce the maximal cumulative reward,
\[
R(\tilde{c}_\ell) = R(c_\ell).
\]

\textit{Conclusion.} Filtering preserves all value-relevant paths under Markov leader policies. Thus, filtering may be applied at each stage without loss of optimality.
\end{proof}

\footnotetext{
The filtering operator \( \mathbb{F} \) retains exactly one representative per cumulative reward vector within each marginal belief class. Since transitions and rewards in the credible MDP depend only on these quantities under Markov policies, the tie-breaking rule used to select representatives has no impact on value.
}

\subsection{Bellman's Optimality Equations}
\label{appendix:proof:thm:boe}
\begin{theorem}
\label{appendix:thm:boe}
The optimal leader state-value function \( v_{\colorL,t}^* \colon \tilde{C} \to \mathbb{R}\) at stage \(t<\ell\) satisfies: 
\begin{align*}
v_{\colorL,t}^*\colon \tilde{c} &\mapsto \textstyle 
\max_{\delta_{\colorL} \in \Delta_{\colorL}} v_{\colorL,t+1}^*\big(\tilde{T}(\tilde{c},\delta_{\colorL})\big).
\end{align*}
\end{theorem}
\begin{proof}
The proof proceeds from backward induction for finite-horizon discrete-time Markov decision process \cite{Puterman1994}. We begin with the definition: for any transient credible set \(\tilde{c}\) at stage \(t<\ell\),
\begin{align*}
 v_{\colorL,t}^*(\tilde{c}) &= 
 \max_{\pi_{\colorL} \in \Pi_{\colorL}}
 v_{\colorL,t}^{\pi_{\colorL}}(\tilde{c}).
\end{align*}
By decomposing leader policy \(\pi_{\colorL} \doteq (\delta_{\colorL,t},\ldots, \delta_{\colorL,\ell-1})\), we obtain the following expression
\begin{align*}
 v_{\colorL,t}^*(\tilde{c}) &= 
 \max_{\delta_{\colorL,t} \in \Delta_{\colorL,t}}
 \cdots
 \max_{\delta_{\colorL,\ell-1} \in \Delta_{\colorL,\ell-1}}
 v_{\colorL,t}^{\pi_{\colorL}}(\tilde{c}).
\end{align*}
Since the policy tail \(\pi'_{\colorL} \doteq (\delta_{\colorL,t+1},\ldots, \delta_{\colorL,\ell-1})\) applies to the next credible set \(\tilde{T}(\tilde{c},\delta_{\colorL,t})\) reached by executing decision rule \(\delta_{\colorL,t}\) at \(\tilde{c}\), we obtain:
\begin{align*}
  v_{\colorL,t}^* (\tilde{c}) &=
  \max_{\delta_{\colorL,t}\in \Delta_{\colorL,t}}
  \max_{\pi'_{\colorL}\in \Delta_{\colorL,t+1:\ell-1}} v_{\colorL,t+1}^{\pi'_{\colorL}}(\tilde{T}(\tilde{c},\delta_{\colorL,t})).  
\end{align*}
Injecting the definition of the optimal value function yields:
\begin{align*}
  v_{\colorL,t}^* (\tilde{c}) &=
  \max_{\delta_{\colorL,t}\in \Delta_{\colorL,t}}
  v_{\colorL,t+1}^*(\tilde{T}(\tilde{c},\delta_{\colorL,t}))\\
  &=
  \max_{\delta_{\colorL,t}\in \Delta_{\colorL,t}}
  q_{\colorL,t}^*(\tilde{c},\delta_{\colorL,t}).
\end{align*}
This concludes the proof.
\end{proof}

\subsection{Uniform Continuity Properties}
We establish the uniform continuity of the leader’s value function over finite credible sets. For clarity, we decompose the proof into three lemmas leading to the main proposition. Before proceeding any further we start with the definition of uniform continuity and the Hausdorff distance across credible sets useful for establishing the uniform continuity properties.

\begin{definition}[Uniform Continuity]
    A function \(v\colon \tilde{C}\to \mathbb{R}\) is uniformly continuous on \(\tilde{C}\) if:
    \begin{align*}
        \forall \varepsilon > 0,
        \exists \delta > 0
        \text{ such that }
        \forall \tilde{c},\tilde{c}'\in \tilde{C},
        |\tilde{c}-\tilde{c}'| < \delta
        \implies 
        |v(\tilde{c})-v(\tilde{c}')| < \varepsilon.
    \end{align*}
\end{definition}

Next, we introduce a distance between two credible sets, \(\tilde{c}\) and \(\tilde{c}'\), drawn from the occupancy-state space \(O\), using the Hausdorff distance with the \(\ell_1\)-norm. The Hausdorff distance between two credible sets measures how far apart they are by capturing the worst-case minimal \(\ell_1\)-distance needed to match each occupancy state in one set to some occupancy state in the other. If every occupancy state in both credible sets is close to at least one in the other, the sets are considered behaviorally close.

\begin{definition}[Hausdorff Distance]
Let \( \tilde{c}, \tilde{c}' \subset O \) be finite credible sets, and equip the space of finite subsets of \( O \) with the Hausdorff distance \( d_H \) induced by the \( \ell_1 \)-norm:
\[
\textstyle 
d_H(\tilde{c},\tilde{c}') \doteq \max \left\{ \sup_{o \in \tilde{c}} \inf_{o' \in \tilde{c}'} \| o - o' \|_1,\; \sup_{o' \in \tilde{c}'} \inf_{o \in \tilde{c}} \| o' - o \|_1 \right\}.
\]
\end{definition}

\begin{lemma}[Stability of filtered maxima under Hausdorff perturbations]
\label{lem:stability:filtered:maxima}
Let \(v \colon O \to \mathbb{R}\) be uniformly continuous, and let \(x, \tilde{c}' \subset O\) be credible sets. Then for every \( \varepsilon > 0 \), there exists \( \delta > 0 \) such that if \( d_H(\tilde{c}, \tilde{c}') < \delta \), then
\[
\left| \max_{o \in x} v(o) - \max_{\bar{o} \in \tilde{c}'} v(\bar{o}) \right| < \varepsilon.
\]
\end{lemma}
\begin{proof}
Fix \( \varepsilon > 0 \). By uniform continuity of \( v \colon O \to \mathbb{R} \), there exists \( \delta > 0 \) such that:
\[
\|o - \bar{o}\|_1 < \delta \quad \Rightarrow \quad |v(o) - v(\bar{o})| < \varepsilon.
\]

Assume \( d_H(\tilde{c}, \tilde{c}') < \delta \). Then:
\begin{enumerate}
\item For every \( o \in \tilde{c} \), there exists \( \bar{o} \in \tilde{c}' \) with \( \|o - \bar{o}\|_1 < \delta \), hence \( |v(o) - v(\bar{o})| < \varepsilon \).
Taking maxima over \( o \in \tilde{c} \), we obtain:
\[
\max_{o \in \tilde{c}} v(o) < \max_{\bar{o} \in \tilde{c}'} v(\bar{o}) + \varepsilon.
\]

\item Similarly, for every \( \bar{o} \in \tilde{c}' \), there exists \( o \in \tilde{c} \) with \( \|\bar{o} - o\|_1 < \delta \), hence \( |v(\bar{o}) - v(o)| < \varepsilon \), and thus:
\[
\max_{\bar{o} \in \tilde{c}'} v(\bar{o}) < \max_{o \in \tilde{c}} v(o) + \varepsilon.
\]
\end{enumerate}

Combining both bounds:
\(
\left| \max_{o \in \tilde{c}} v(o) - \max_{\bar{o} \in \tilde{c}'} v(\bar{o}) \right| < \varepsilon
\).
\end{proof}

\begin{lemma}[Leader value function uniform continuity]
\label{lem:continuity:filtered:sets}
Let \( K \) and \( N \) be finite index sets, and let \( O \) be a compact metric space with distance \( \| \cdot \| \). For each \( k \in K \) and \( n \in N \), let \( v_{kn}, u_{kn} \colon O \to \mathbb{R} \) be uniformly continuous functions. For any credible set \( \tilde{c} \subset O \), define the filtered set
\[
\tilde{c}_{kn}^* \doteq \big\{ o \in \tilde{c} \mid u_{kn}(o) = \max_{o' \in \tilde{c}} \max_{n' \in N} u_{kn'}(o') \big\},
\]
and define the composite value function
\[
v^*(\tilde{c}) \doteq \max_{k \in K} \max_{n \in N} \max_{o \in \tilde{c}_{kn}^*} v_{kn}(o).
\]
Then \( v^* \) is uniformly continuous over credible sets equipped with the Hausdorff distance \( d_H \) induced by \( \| \cdot \| \).
\end{lemma}

\begin{proof}
Let \( \varepsilon > 0 \) be arbitrary.

By uniform continuity of each \( u_{kn} \) and by Lemma~\ref{lem:stability:filtered:maxima},  
there exists \( \delta_{kn}^u > 0 \) such that if \( d_H(\tilde{c}, \tilde{c}') < \delta_{kn}^u \),  
then the filtered sets satisfy:
\[
d_H(\tilde{c}_{kn}^*, \tilde{c}_{kn}'^*) < \delta_{kn}^v,
\]
where \( \delta_{kn}^v \) will be chosen next.

By uniform continuity of each \( v_{kn} \) and again by Lemma~\ref{lem:stability:filtered:maxima},  
there exists \( \delta_{kn}^v > 0 \) such that if \( d_H(\tilde{c}_{kn}^*, \tilde{c}_{kn}'^*) < \delta_{kn}^v \), then:
\[
\big| \max_{o \in \tilde{c}_{kn}^*} v_{kn}(o) - \max_{o' \in \tilde{c}_{kn}'^*} v_{kn}(o') \big| < \varepsilon.
\]

Define:
\[
\delta \doteq \min_{k, n} \{ \delta_{kn}^u, \delta_{kn}^v \}.
\]

Then, under \( d_H(\tilde{c}, \tilde{c}') < \delta \),  
each index pair satisfies:
\[
\big| \max_{o \in \tilde{c}_{kn}^*} v_{kn}(o) - \max_{o' \in \tilde{c}_{kn}'^*} v_{kn}(o') \big| < \varepsilon,
\]
and the global maximum over finitely many \( (k, n) \) stays controlled:
\[
|v^*(\tilde{c}) - v^*(\tilde{c}')| < \varepsilon.
\]

Therefore, \( v^* \) is uniformly continuous.
\end{proof}

In demonstrating the uniform continuity of \(v_{\colorL}^*\colon \tilde{C}\to\mathbb{R}\) across credible sets, it will prove useful to show the underlying structure of value function \(v^*_{\colorL}\) across credible sets.

\begin{lemma}
\label{appendix:lem:convexity}
The optimal leader state-value function \( v_{\colorL,t}^* \colon \tilde{C} \to \mathbb{R} \) at stage \(t\) is such that: there exists a collection \(\Lambda_t\) of finite sets \( \Gamma_t \) of linear function pairs \( (\alpha_{\colorL}, \alpha_{\colorF}) \), each defined over occupancy states, such that:
\[
    v_{\colorL,t}^*(\tilde{c})  
    = 
    \max_{\Gamma\in \Lambda_t}
    \left\{ 
    \rho_{\colorL}(o)
    +
    \gamma^t
    \sum_{h_{\colorF}}
    \Pr(h_{\colorF} \mid o)
    \cdot \alpha^{o_{h_{\colorF}}}_{\colorL}(o_{h_{\colorF}})
    \mid 
          (o, \alpha^{\cdot})
            \in 
            \argmax_{o'\in \tilde{c}, \alpha'\in \Gamma^{|\mathbb{C}_{\colorF}(\tilde{c})|}}
           \left(
           \rho_{\colorF}(o')
           +
           \gamma^t
    \sum_{h_{\colorF}}
    \Pr(h_{\colorF} \mid o)
    \cdot             \alpha'^{o_{h_{\colorF}}}_{\colorF}(o'_{h_{\colorF}})
            \right)
    \right\}.
\]
\end{lemma}
\begin{proof}
The proof follows directly from the definition of the optimal leader state-value function at stage \(t\): for any arbitrary credible set \(\tilde{c}\in \tilde{C}\), we have that:
\begin{align}
    v_{\colorL,t}^*(\tilde{c}) &\doteq 
    \textstyle
    \max_{\pi_{\colorL}\in \Pi_{\colorL}}
    v_{\colorL,t}^{\pi_{\colorL}}(\tilde{c}) 
    \label{eqn:appendix:thm:convexity:0}.
\end{align}
Let  \(\tilde{c}_{\pi_{\colorL}} \doteq \tilde{T}( \tilde{T}( \ldots (\tilde{T}(\tilde{c}, \delta_{\colorL,t}), \delta_{\colorL,t+1}) \ldots), \delta_{\colorL,\ell-1})\) denote a terminal credible set induced by \(\tilde{c}\) and leader policy \(\pi_{\colorL} \doteq (\delta_{\colorL,t}, \delta_{\colorL,t+1}, \ldots,\delta_{\colorL,\ell-1})\) from stage \(t\) onward. Then, it follows from (\ref{eqn:appendix:thm:convexity:0}) that: 
\begin{align}
    v_{\colorL,t}^*(\tilde{c})  &= 
    \textstyle
    \max_{\pi_{\colorL}\in \Pi_{\colorL}}
    R(\tilde{c}_{\pi_{\colorL}})
    \label{eqn:appendix:thm:convexity:1}.
\end{align}
Let \(o_{\pi_{\colorL}, \pi_{\colorF}}\) denote a terminal occupancy state induced by occupancy state \(o\in \tilde{c}\) and leader policy \(\pi_{\colorL} \doteq (\delta_{\colorL,t}, \delta_{\colorL,t+1}, \ldots,\delta_{\colorL,\ell-1})\) and follower policy \(\pi_{\colorF} \doteq (\delta_{\colorF,t}, \delta_{\colorF,t+1}, \ldots,\delta_{\colorF,\ell-1})\) from stage \(t\) onward.
The application of the definition of the reward in a credible Markov decision process into (\ref{eqn:appendix:thm:convexity:1}) leads us to the following expression:
\begin{align}
    v_{\colorL,t}^*(\tilde{c})  
    &= 
    \max_{\pi_{\colorL}\in \Pi_{\colorL}}
    \max_{o_{\pi_{\colorL}, \pi_{\colorF}}\in \tilde{c}_{\pi_{\colorL}}}
    \left\{ 
    \rho_{\colorL}(o_{\pi_{\colorL}, \pi_{\colorF}})
    \mid 
           \rho_{\colorF}(o_{\pi_{\colorL}, \pi_{\colorF}})
            =
            \max_{o'_{\pi_{\colorL}, \pi'_{\colorF}}\in \tilde{c}_{\pi_{\colorL}}}
           \rho_{\colorF}(o'_{\pi_{\colorL}, \pi'_{\colorF}})
    \right\}
    \label{eqn:appendix:thm:convexity:2}.
\end{align}
By the application of the definition of the state-value function \(v_{\colorI,t}^{\pi_{\colorL}, \pi_{\colorF}} \colon S\times H_{\colorL}\times H_{\colorF} \to \mathbb{R}\) under leader--follower policies \((\pi_{\colorL}, \pi_{\colorF})\) from stage \(t\) onward into (\ref{eqn:appendix:thm:convexity:2}) we obtain:
\begin{align}
    v_{\colorL,t}^*(\tilde{c})  
    &= 
    \max_{\pi_{\colorL}\in \Pi_{\colorL}}
    \max_{o\in \tilde{c}}
    \left\{ 
    \rho_{\colorL}(o)
    +
    \gamma^t
    v_{\colorL,t}^{\pi_{\colorL}, \pi_{\colorF}}(o)
    \mid 
           \rho_{\colorF}(o)
           +
           \gamma^t
           v_{\colorF,t}^{\pi_{\colorL}, \pi_{\colorF}}(o)
            =
            \max_{o'\in \tilde{c}}
           \left(
           \rho_{\colorF}(o')
           +
           \gamma^t
            v_{\colorF,t}^{\pi_{\colorL}, \pi'_{\colorF}}(o')
            \right)
    \right\}
    \label{eqn:appendix:thm:convexity:3}.
\end{align}
If we let 
\[ 
\Gamma_{\pi_{\colorL,t:\ell-1}} \doteq \{(v_{\colorL,t}^{\pi_{\colorL,t:\ell-1}, \pi_{\colorF,t:\ell-1}(h_{\colorF,t})}, v_{\colorF,t}^{\pi_{\colorL,t:\ell-1}, \pi_{\colorF,t:\ell-1}(h_{\colorF,t})}) \mid \pi_{\colorF,t:\ell-1}\in \Pi_{\colorF,t:\ell-1}, h_{\colorF,t}\in H_{\colorF,t}\} 
\] 
where \(\pi_{\colorF,t:\ell-1}(h_{\colorF,t})\) is a follower policy tree from stage \(t\) onward, 
and \(\Lambda_t \doteq \{\Gamma_{\pi_{\colorL,t:\ell-1}} \mid \pi_{\colorL,t:\ell-1}\in \Pi_{\colorL,t:\ell-1}\} \), then (\ref{eqn:appendix:thm:convexity:3}) becomes:
\begin{align*}
    v_{\colorL,t}^*(\tilde{c})  
    &= 
    \max_{\Gamma\in \Lambda_t}
    \left\{ 
    \rho_{\colorL}(o)
    +
    \gamma^t
    \sum_{h_{\colorF}}
    \Pr(h_{\colorF} \mid o)
    \cdot \alpha^{o_{h_{\colorF}}}_{\colorL}(o_{h_{\colorF}})
    \mid 
          (o, \alpha^{\cdot})
            \in 
            \argmax_{o'\in \tilde{c}, \alpha'\in \Gamma^{|\mathbb{C}_{\colorF}(\tilde{c})|}}
           \left(
           \rho_{\colorF}(o')
           +
           \gamma^t
    \sum_{h_{\colorF}}
    \Pr(h_{\colorF} \mid o)
    \cdot             \alpha'^{o_{h_{\colorF}}}_{\colorF}(o'_{h_{\colorF}})
            \right)
    \right\}.
\end{align*}
Which ends the proof.
\end{proof}

\begin{theorem}
\label{appendix:thm:convexity}
The optimal leader state-value function \( v_{\colorL,t}^* \colon \tilde{C} \to \mathbb{R} \) at stage \(t\) is uniformly continuous across credible sets.  
\end{theorem}
\begin{proof}
Let \(v^{\Gamma}_{\alpha^{\cdot}}\colon o \mapsto \rho_{\colorL}(o) + \gamma^t \sum_{h_{\colorF}} \Pr(h_{\colorF}\mid o)\cdot \alpha_{\colorL}^{o_{h_{\colorF}}}(o_{h_{\colorF}})\) and \(u^{\Gamma}_{\alpha^{\cdot}}\colon o \mapsto \rho_{\colorF}(o) + \gamma^t \sum_{h_{\colorF}} \Pr(h_{\colorF}\mid o)\cdot \alpha_{\colorF}^{o_{h_{\colorF}}}(o_{h_{\colorF}})\). Both are uniformly continuous over occupancy states \(O\) linearity of \(\rho_{\colorI}\) and \(\alpha^{\cdot}_{\colorI}\) for all \(\colorI\in \{\colorL,\colorF\}\). The result then follows directly from Lemma~\ref{appendix:lem:convexity} and  Lemma~\ref{lem:continuity:filtered:sets}, which ensures the uniform continuity of  
\[
        v^*\colon \tilde{c}\mapsto \max_{\Gamma\in \Lambda}\{v^\Gamma_{\alpha^{\cdot}}(o) \mid (o,\alpha^{\cdot}) \in \argmax_{o'\in \tilde{c},\alpha'^{\cdot}\in \Gamma^{|\mathbb{C}_\colorF(\tilde{c})|}} u^{\Gamma}_{\alpha'^{\cdot}}(o') \}
.\]
\end{proof}

\begin{theorem}
\label{appendix:thm:pbvi-exploitability}
Let \( \sigma \doteq \sup_{\tilde{c} \in \tilde{C}} \min_{\tilde{c}' \in \tilde{C}'} d_H(\tilde{c}, \tilde{c}') \) be the Hausdorff covering radius of the PBVI sample set \(\tilde{C}'\). Then the exploitability \(\varepsilon\) of the leader policy returned by PBVI satisfies:
\[
\varepsilon \leq \textstyle \frac{2m\sigma}{(1 - \gamma)^2} \left[ 1 + \ell \gamma^{\ell+1} - (\ell + 1) \gamma^\ell \right].
\]
\end{theorem}
\begin{proof}
    Let \(\pi_\colorL \doteq \delta_{\colorL,0:\ell-1}\) be an optimal leader policy with value \(v^*_{\colorL}(\tilde{c}_0)\). 
    Let \((\tilde{c}_0, \ldots, \tilde{c}_\ell)\) denote the sequence of credible sets induced by \(\pi_\colorL\) and the follower (best-response) policy \(\pi_\colorF\doteq \delta_{\colorF,0:\ell-1}\), 
    for which PBVI yields the worst estimate.
    Let \((\tilde{c}'_0, \ldots, \tilde{c}'_\ell)\) be the closest sequence of credible sets to \((\tilde{c}_0, \ldots, \tilde{c}_\ell)\) in Hausdorff distance \(d_H(\cdot,\cdot)\), induced by the sampled credible sets \(\tilde{C}'_{0:\ell}\).
    As a consequence, the following inequality holds \(d_H(\tilde{c}_t,\tilde{c}'_t) \leq \sigma\) for any stage \(t\).
    Let \(v_\colorL\) be the approximate value function,  \(\pi'_\colorL \doteq \delta'_{\colorL,0:\ell-1}\) the induced leader policy computed by PBVI over \(\tilde{C}'_{0:\ell}\) and the follower (best-response) policy  \(\pi'_\colorF \doteq \delta'_{\colorF,0:\ell-1}\).
    Let \(v_{\colorL,0:\ell}^{\pi_\colorL,\pi_\colorF}\) and \(v_{\colorL,0:\ell}^{\pi'_\colorL,\pi'_\colorF}\) be the value functions induced by pairs of behavioural strategies, each uniformly continuous over credible sets, such that \(v_{\colorL,0}^*(\tilde{c}_0) = v_{\colorL,0}^{\pi_\colorL,\pi_\colorF}(s_0)\) and \(v_{\colorL,0}(\tilde{c}_0) = v_{\colorL,0}^{\pi'_\colorL,\pi'_\colorF}(s_0)\), respectively. 
    Then, 
\begin{align*}
\varepsilon &\textstyle\doteq
v_{\colorL,0}^*(\tilde{c}_0) - v_{\colorL,0}^{\pi'_\colorL, \pi'_\colorF}(s_0).
\end{align*}
Let \(\nu_\colorL(o_t, \delta_{\colorL,t}, \delta_{\colorF,t}) \doteq \mathbb{E}[r(s_t, a_{\colorL,t}, a_{\colorF,t}, s_{t+1})\mid o_t, \delta_{\colorL,t}, \delta_{\colorF,t}]\) denotes the expected immediate leader payoff upon taking decision rules \((\delta_{\colorL,t}, \delta_{\colorF,t})\) in occupancy state \(o_t \doteq o_{\delta_{\colorL,0:t-1}\delta_{\colorF,0:t-1}} \in \tilde{c}_t\).
Given \(v_{\colorL,0}^*(\tilde{c}_0) = v_{\colorL,0}^{\pi_\colorL,\pi_\colorF}(s_0)\) and \(v_{\colorL,0}(\tilde{c}_0) = v_{\colorL,0}^{\pi'_\colorL,\pi'_\colorF}(s_0)\), it follows that:
\begin{align*}
    v_{\colorL,0}^*(\tilde{c}_0) - v_{\colorL,0}^{\pi'_\colorL, \pi'_\colorF}(s_0)
    &= 
    v_{\colorL,0}^{\pi_\colorL,\pi_\colorF}(s_0) - v_{\colorL,0}^{\pi'_\colorL, \pi'_\colorF}(s_0)\\
    &\textstyle= \left(\textcolor{sthlmGreen}{\sum_{t=0}^{\ell-1} \gamma^t \cdot \nu_\colorL(o_t, \delta_{\colorL,t}, \delta_{\colorF,t})} \right) - v_{\colorL,0}^{\pi'_\colorL, \pi'_\colorF}(s_0) \quad \text{(by definition)} \\
    &\textstyle= \left(\sum_{t=0}^{\ell-1} \gamma^t \cdot \nu_\colorL(o_t, \delta_{\colorL,t}, \delta_{\colorF,t})\right) \textcolor{sthlmGreen}{ - \sum_{t=1}^{\ell} \gamma^t (v_{\colorL,t}^{\pi'_\colorL, \pi'_\colorF}(o_t) - v_{\colorL,t}^{\pi'_\colorL, \pi'_\colorF}(o_t))} - v_{\colorL,0}^{\pi'_\colorL, \pi'_\colorF}(s_0) \quad \text{(adding zero)}.
\end{align*}
Using the convention \(v_{\colorL,\ell}^{\pi'_\colorL, \pi'_\colorF}(\cdot) \doteq 0\), we rearrange terms:
\begin{align*}
&\textstyle= \sum_{t=0}^{\ell-1} \gamma^t \cdot \nu_\colorL(o_t, \delta_{\colorL,t}, \delta_{\colorF,t})
+ \left( \gamma^{\ell} \cdot v_{\colorL,\ell}^{\pi'_\colorL, \pi'_\colorF}(o_{\ell}) + \sum_{t=1}^{\ell-1} \gamma^t \cdot v_{\colorL,t}^{\pi'_\colorL, \pi'_\colorF}(o_t) \right)
- \left( \gamma^0 \cdot v_{\colorL,0}^{\pi'_\colorL, \pi'_\colorF}(s_0) + \sum_{t=1}^{\ell-1} \gamma^t \cdot v_{\colorL,t}^{\pi'_\colorL, \pi'_\colorF}(o_t) \right) \\
&\textstyle= \sum_{t=0}^{\ell-1} \gamma^t \cdot \nu_\colorL(o_t, \delta_{\colorL,t}, \delta_{\colorF,t})
+ \sum_{t=0}^{\ell-1} \gamma^{t+1} \cdot v_{\colorL,t}^{\pi'_\colorL, \pi'_\colorF}(o_{t+1}) - \sum_{t=0}^{\ell-1} \gamma^t \cdot v_{\colorL,t}^{\pi'_\colorL, \pi'_\colorF}(o_t) \\
&\textstyle= \sum_{t=0}^{\ell-1} \gamma^t \left(  \nu_\colorL(o_t, \delta_{\colorL,t}, \delta_{\colorF,t}) + \gamma \cdot v_{\colorL,t}^{\pi'_\colorL, \pi'_\colorF}(o_{t+1}) - v_{\colorL,t}^{\pi'_\colorL, \pi'_\colorF}(o_t) \right) \\
&\textstyle= \sum_{t=0}^{\ell-1} \gamma^t \left( q_{\colorL,t}^{\pi'_\colorL, \pi'_\colorF}(o_t, \delta'_{\colorL,t}, \delta'_{\colorF,t}) - v_{\colorL,t}^{\pi'_\colorL, \pi'_\colorF}(o_t) \right).
\end{align*}
Now substitute \(o'_t  \in \tilde{c}'_t\) such that \(\|o_t-o'_t\|_1 \le \sigma\) in place of \(o_t\). Notice that \(o'_t\) must exist since \(d_H(\tilde{c}_t,\tilde{c}'_t) \le \sigma\). Thus,
\begin{align*}
&\textstyle= \sum_{t=0}^{\ell-1} \gamma^t \left( q_{\colorL,t}^{\pi'_\colorL, \pi'_\colorF}(o_t, \delta'_{\colorL,t}, \delta'_{\colorF,t}) \textcolor{sthlmGreen}{ - q_{\colorL,t}^{\pi'_\colorL, \pi'_\colorF}(o'_t, \delta'_{\colorL,t}, \delta'_{\colorF,t}) + q_{\colorL,t}^{\pi'_\colorL, \pi'_\colorF}(o'_t, \delta'_{\colorL,t}, \delta'_{\colorF,t})} - v_{\colorL,t}^{\pi'_\colorL, \pi'_\colorF}(o_t) \right).
\end{align*}
Because the greedy rule for \(q_{\colorL,t}^{\pi'_\colorL, \pi'_\colorF}(o'_t, \cdot, \cdot)\) achieves value \(v_{\colorL,t}^{\pi'_\colorL, \pi'_\colorF}(o'_t)\), we have:
\begin{align*}
&\textstyle\leq \sum_{t=0}^{\ell-1} \gamma^t \left( q_{\colorL,t}^{\pi'_\colorL, \pi'_\colorF}(o_t, \delta'_{\colorL,t}, \delta'_{\colorF,t})  - q_{\colorL,t}^{\pi'_\colorL, \pi'_\colorF}(o'_t, \delta'_{\colorL,t}, \delta'_{\colorF,t}) + v_{\colorL,t}^{\pi'_\colorL, \pi'_\colorF}(o'_t) - v_{\colorL,t}^{\pi'_\colorL, \pi'_\colorF}(o_t) \right)\\
&\textstyle=
\sum_{t=0}^{\ell-1} \gamma^t \left[ (q_{\colorL,t}^{\pi'_\colorL, \pi'_\colorF}(o_t, \delta'_{\colorL,t}, \delta'_{\colorF,t})  - v_{\colorL,t}^{\pi'_\colorL, \pi'_\colorF}(o_t))  - ( q_{\colorL,t}^{\pi'_\colorL, \pi'_\colorF}(o'_t, \delta'_{\colorL,t}, \delta'_{\colorF,t}) - v_{\colorL,t}^{\pi'_\colorL, \pi'_\colorF}(o'_t)) \right]\\
&\textstyle=
\sum_{t=0}^{\ell-1} \gamma^t \left[q_{\colorL,t}^{\pi'_\colorL, \pi'_\colorF}(\cdot, \delta'_{\colorL,t}, \delta'_{\colorF,t})  - v_{\colorL,t}^{\pi'_\colorL, \pi'_\colorF}(\cdot) \right] \cdot [o_t -o'_t],
~\text{(by linearity over occupancy states)}.
\end{align*}
Applying Hölder’s inequality, using the definition of \(\sigma\), and the fact that payoffs are bounded:
\begin{align*}
v_{\colorL,0}^*(\tilde{c}_0) - v_{\colorL,0}^{\pi'_\colorL, \pi'_\colorF}(s_0)
&\textstyle\leq \sum_{t=0}^{\ell-1} \gamma^t \cdot \left\| q_{\colorL,t}^{\pi'_\colorL, \pi'_\colorF}(\cdot, \delta'_{\colorL,t}, \delta'_{\colorF,t})  - v_{\colorL,t}^{\pi'_\colorL, \pi'_\colorF}(\cdot) \right\|_\infty \cdot \|o_t -o'_t\|_1 \\
&\textstyle\leq \sigma \sum_{t=0}^{\ell-1} \gamma^t \cdot \left\| q_{\colorL,t}^{\pi'_\colorL, \pi'_\colorF}(\cdot, \delta'_{\colorL,t}, \delta'_{\colorF,t})  - v_{\colorL,t}^{\pi'_\colorL, \pi'_\colorF}(\cdot) \right\|_\infty \\
&\textstyle\leq 2m\sigma \sum_{t_0=0}^{\ell-1} \gamma^{t_0} \sum_{t_1=t_0}^{\ell-1} \gamma^{t_1 - t_0},~\text{(where \(q_{\colorL,t}^{\pi'_\colorL, \pi'_\colorF}\) and \(v_{\colorL,t}^{\pi'_\colorL, \pi'_\colorF}\) linear over occupancy states)} \\
&\textstyle= 2m\sigma \sum_{t_0=0}^{\ell-1}\sum_{t_1=t_0}^{\ell-1} \gamma^{t_1} \\
&\textstyle= 2m\sigma \sum_{t_0=0}^{\ell-1} \frac{\gamma^{t_0} - \gamma^{\ell}}{1 - \gamma} \\
&\textstyle= \frac{2m\sigma}{1 - \gamma} \sum_{t_0=0}^{\ell-1} (\gamma^{t_0} - \gamma^{\ell}) \\
&\textstyle= \frac{2m\sigma}{(1 - \gamma)^2} \left[ 1 + \ell \gamma^{\ell+1} - (\ell + 1) \gamma^\ell \right].
\end{align*}
Which ends the proof.
\end{proof}

%% file: appendix/greedy-decision-rule.tex
\subsection{Greedy Leader Decision Rules}
\label{appendix:greedy:leader:dr}

\begin{definition}
Let \(o\) be an occupancy state.
A conditional occupancy state \(o_{h_{\colorF}}\) induced by occupancy state \(o\) and follower history \(h_{\colorF}\) is given by: for any state \(s\in S\) and leader history \(h_{\colorL}\in H_{\colorL}\),
\[
o_{h_{\colorF}}(s,h_{\colorL}) \doteq \frac{ o(s,h_{\colorL},h_{\colorF}) }{ \sum_{s,h_{\colorL}} o(s,h_{\colorL},h_{\colorF})}.
\]
\end{definition}
Each occupancy state \( o \) admits a convex decomposition of the form
\[
\textstyle o \;=\; \sum_{h_{\colorF}} \Pr\bigl(h_{\colorF}\mid o\bigr)\,\cdot \bigl( o_{h_{\colorF}} \otimes \pmb{e}_{h_{\colorF}} \bigr),
\]
where \( \pmb{e}_{h_{\colorF}} \) is the one-hot vector over follower histories and \( \otimes \) denotes the Kronecker product.
\begin{lemma}
\label{lem:conditional:occupancy:state}
    Let \(o\) be an occupancy state, \(o_{h_{\colorF}}\) be a conditional occupancy state induced by occupancy state \(o\) and follower history \(h_{\colorF}\). 
    Let \(\delta_{\colorL}\) be a leader decision rule, and \(a_{\colorF}\) be a follower action.
    The next conditional occupancy state \(o_{h_{\colorF}a_{\colorF}s'} \doteq \tau_{s'}(o_{h_{\colorF}},\delta_{\colorL},a_{\colorF})\) upon receiving next state \(s'\) after taking leader-follower decisions \((\delta_{\colorL},a_{\colorF})\) starting in conditional occupancy state \(o_{h_{\colorF}}\) is given by: for any state \(s'\) and leader history \(h'_{\colorL}\doteq h_{\colorL}a_{\colorL}s'\),
    \begin{align*}
        o_{h_{\colorF}a_{\colorF}s'}(s',h'_{\colorL}) &\propto \textstyle\sum_s o_{h_{\colorF}}(s,h_{\colorL}) \cdot \delta_{\colorL}(a_{\colorL}|h_{\colorL}) \cdot p(s'|s,a_{\colorL},a_{\colorF}),
    \end{align*}
    with normalising constant \(\eta_{s'}(o_{h_{\colorF}},\delta_{\colorL},a_{\colorF})\).
\end{lemma}

\begin{table}[!ht]
\begin{align}
    &
    \delta_{\colorL}^{\Gamma,o}
    \in  
    \textstyle 
    \argmax_{\delta_{\colorL}\in \Delta_{\colorL}} 
    q_{\colorL,t}^{\Gamma}(o,\delta_{\colorL}) 
    \label{appendix:milp:constraint:1}\\
    &
    \text{subject to}\colon
    \textstyle 
    q_{\colorF,t}^{\Gamma}(o,\delta_{\colorL}) 
    \ge
    q^{\Gamma}_{\colorF,t}(o\prime,\delta_{\colorL}),
    ~
    \forall o\prime\in \tilde{c}
    \label{appendix:milp:constraint:2}\\
    &q_{\colorL,t}^{\Gamma}(o,\delta_{\colorL}) =
    \textstyle
    \rho_{\colorL}(o)
    +
    \gamma^t
    \sum_{h_{\colorF}} 
    \Pr(h_{\colorF}\mid o)
    \cdot 
    g_{\colorL,t}^{\Gamma}(o_{h_{\colorF}},\delta_{\colorL})
   \label{appendix:milp:constraint:3}\\
    &q_{\colorF,t}^{\Gamma}(o\prime,\delta_{\colorL}) =
    \textstyle
    \rho_{\colorF}(o\prime)
    +
    \gamma^t
    \sum_{h_{\colorF}} 
    \Pr(h_{\colorF}\mid o\prime)
    \cdot 
    g_{\colorF,t}^{\Gamma}(o\prime_{h_{\colorF}},\delta_{\colorL}),
    \forall
     o\prime\in \tilde{c}
    \label{appendix:milp:constraint:4} \\
    &
    g^\Gamma_{\colorF,t}(o\prime_{h_{\colorF}},\delta_\colorL) 
    \ge  
    \sum_{s'\in S} 
    \beta^\Gamma_{\colorF,t}(o\prime_{h_{\colorF}},\delta_\colorL,a_\colorF,s'),
    ~
    \forall a_\colorF \in A_\colorF, \forall o\prime_{h_{\colorF}}\in \mathbb{C}_\colorF(\tilde{c})
    \label{appendix:milp:constraint:5}\\
    &
    g^\Gamma_{\colorF,t}(o\prime_{h_{\colorF}},\delta_\colorL) 
    \le  
    \sum_{s'\in S} 
    \beta^\Gamma_{\colorF,t}(o\prime_{h_{\colorF}},\delta_\colorL,a_\colorF,s') 
    +
    \mathtt{M}
    \cdot  
    (1 - \kappa^{\Gamma,o}_\colorF(o\prime_{h_{\colorF}},a_{\colorF})),
    ~
    \forall a_{\colorF}\in A_{\colorF}, o\prime_{h_{\colorF}}\in \mathbb{C}_\colorF(\tilde{c})
    \label{appendix:milp:constraint:6}\\
    &
    \sum_{a_{\colorF}\in A_{\colorF}}
    \kappa^{\Gamma,o}_\colorF(o\prime_{h_{\colorF}},a_{\colorF})
    = 1,
    ~\forall o\prime_{h_{\colorF}}\in \mathbb{C}_\colorF(\tilde{c})
    \label{appendix:milp:constraint:7}\\
    &
    \beta_{\colorF,t}^{\Gamma}(o\prime_{h_{\colorF}},\delta_{\colorL},a_{\colorF},s') 
    \ge
    \nu_\colorF(o\prime_{h_\colorF}, \delta_\colorL, a_\colorF, s')
    +
    \gamma 
     \alpha_{\colorF}(\tau_{s'}(o\prime_{h_{\colorF}},\delta_{\colorL},a_{\colorF})),
     ~
     \forall 
     s'\in S, a_{\colorF}\in A_{\colorF}, o\prime_{h_{\colorF}}\in  \mathbb{C}_{\colorF}(\tilde{c}), \alpha\in \Gamma
     \label{appendix:milp:constraint:8}\\
    &
    \beta_{\colorF,t}^{\Gamma}(o\prime_{h_{\colorF}},\delta_{\colorL},a_{\colorF},s') 
    \le
    \nu_\colorF(o\prime_{h_\colorF}, \delta_\colorL, a_\colorF,s')
    +
    \gamma 
     \alpha_{\colorF}(\tau_{s'}(o\prime_{h_{\colorF}},\delta_{\colorL},a_{\colorF}))
     +
    \mathtt{M}
    \cdot  
    (1 - w^{\Gamma,o}_\colorF(o\prime_{h_{\colorF}},s',\alpha))
    \label{appendix:milp:constraint:9}\\
    &\qquad
     \forall 
     s'\in S, a_{\colorF}\in A_{\colorF}, o\prime_{h_{\colorF}}\in  \mathbb{C}_{\colorF}(\tilde{c}), \alpha\in \Gamma \nonumber
     \\
    &
    \sum_{\alpha\in \Gamma}
    w^{\Gamma,o}_\colorF(o\prime_{h_{\colorF}},s',\alpha)
    = 1,
    ~\forall s'\in S, o\prime_{h_{\colorF}}\in \mathbb{C}_\colorF(\tilde{c})
    \label{appendix:milp:constraint:11}\\
    &
    g_{\colorL,t}^{\Gamma}(o_{h_{\colorF}},\delta_{\colorL}) 
    \ge
     \sum_{s'}
     \beta_{\colorL,t}(o_{h_{\colorF}},\delta_{\colorL},a_{\colorF},s')
     -
     \mathtt{M}
     \cdot 
     (1 - \kappa^{\Gamma,o}_\colorF(o_{h_{\colorF}},a_{\colorF}))
     ,
     ~
     \forall
     a_{\colorF}\in A_{\colorF},
     o_{h_{\colorF}}\in  \mathbb{C}_{\colorF}(o)
     \label{appendix:milp:constraint:12}\\
    &
    g_{\colorL,t}^{\Gamma}(o_{h_{\colorF}},\delta_{\colorL}) 
    \le
     \sum_{s'}
     \beta_{\colorL,t}(o_{h_{\colorF}},\delta_{\colorL},a_{\colorF},s')
     +
     \mathtt{M}
     \cdot 
     (1 - \kappa^{\Gamma,o}_\colorF(o_{h_{\colorF}},a_{\colorF}))
     ,
     ~
     \forall
     a_{\colorF}\in A_{\colorF},
     o_{h_{\colorF}}\in  \mathbb{C}_{\colorF}(o)
     \label{appendix:milp:constraint:13}\\
     &
     \beta_{\colorL,t}^{\Gamma}(o_{h_{\colorF}},\delta_{\colorL},a_{\colorF},s') 
    \ge
    \nu_\colorL(o_{h_\colorF}, \delta_\colorL, a_\colorF,s')
    +
    \gamma
     \alpha_{\colorL}(\tau_{s'}(o_{h_{\colorF}},\delta_{\colorL},a_{\colorF})) 
     -
     \mathtt{M}
     \cdot 
     (1 - w^{\Gamma,o}_\colorF(o_{h_{\colorF}},s',\alpha)),
     \label{appendix:milp:constraint:14}\\
    &\qquad
     \forall s'\in S, a_{\colorF}\in A_{\colorF},  o_{h_{\colorF}}\in  \mathbb{C}_{\colorF}(o), \alpha\in \Gamma. \nonumber\\
     &
     \beta_{\colorL,t}^{\Gamma}(o_{h_{\colorF}},\delta_{\colorL},a_{\colorF},s') 
    \le
    \nu_\colorL(o_{h_\colorF}, \delta_\colorL, a_\colorF,s')
    +
    \gamma
     \alpha_{\colorL}(\tau_{s'}(o_{h_{\colorF}},\delta_{\colorL},a_{\colorF})) 
     +
     \mathtt{M}
     \cdot 
     (1 - w^{\Gamma,o}_\colorF(o_{h_{\colorF}},s',\alpha)),
     \label{appendix:milp:constraint:15}\\
    &\qquad
     \forall s'\in S, a_{\colorF}\in A_{\colorF},  o_{h_{\colorF}}\in  \mathbb{C}_{\colorF}(o), \alpha\in \Gamma. \nonumber      
\end{align}
\caption{This mixed-integer linear program  computes the greedy leader decision rule 
\( \delta_{\colorL}^{\tilde{c}} \) 
maximising the stage-\(t\) leader action-value as described in Theorem \ref{appendix:proof:thm:pbvi-backup}. 
\(\mathtt{M}\)
is a sufficiently large constant for big-M constraints, 
\(w_{\colorF}^{\Gamma}(o\prime_{h_{\colorF}},s', \alpha) \) and \(\kappa^{\Gamma,o}_\colorF(o\prime_{h_\colorF}, a_\colorF)\)
are binary selection variables, and \(\delta_{\colorL}(\cdot|\cdot)\)
are free variables denoting altogether the leader decision rule, and 
\(q_{\colorI,t}^{\Gamma}(o,\delta_{\colorL})\), \(g_{\colorI,t}^\Gamma(o_{h_\colorF, \delta_\colorL})\), \(\beta^\Gamma_{\colorI,t}(o_{h_\colorF}, \delta_\colorL, a_\colorF, s\prime)\)
are free variables. 
Note that \(\mathbb{C}_{\colorF}(o) \doteq \{ o_{h_{\colorF}} \mid h_{\colorF} \in H_{\colorF}, \sum_{s,h_{\colorL}} o(s,h_{\colorL},h_{\colorF}) > 0 \} \) and \(\mathbb{C}_{\colorF}(\tilde{c}) \doteq \cup_{o'\in \tilde{c}}~ \mathbb{C}_{\colorF}(o') \).}
\label{appendix:eq:milp-backup}
\end{table}

\subsection{Interpretation of the MILP for Greedy Backup}
\label{appendix:interpretation:milp}

This section explains the mixed-integer linear program (MILP) in \Cref{appendix:eq:milp-backup}, which performs a point-based greedy backup for the leader in a dynamic programming framework for solving leader--follower general-sum stochastic games (LF-GSSGs).

\paragraph{Problem setting.} 
Given:
\begin{itemize}
    \item a filtered credible set \( \tilde{c} \subseteq \Delta(S \times H_{\colorL} \times H_{\colorF}) \),
    \item a specific occupancy state \( o \in \tilde{c} \),
    \item and a finite set \( \Gamma \) of \(\alpha\)-vector pairs \( (\alpha_{\colorL}, \alpha_{\colorF}) \),
\end{itemize}
the goal is to compute a leader decision rule \( \delta^{\Gamma,o}_{\colorL} \) that maximises the stage-\(t\) expected return of the leader at \(o\), under the constraint that the follower best-responds in each conditional occupancy state \( o_{h_{\colorF}} \in \mathbb{C}_{\colorF}(\tilde{c}) \), and that the value approximations respect the \(\alpha\)-vector representation.

\paragraph{High-level structure.}
The MILP maximises the leader’s action-value function while enforcing:
\begin{enumerate}
    \item follower best response at state \(o\) (Stackelberg consistency),
    \item consistent value approximation via binary selection of \(\alpha\)-vectors,
    \item alignment of leader and follower values under these selections.
\end{enumerate}
This is encoded using continuous variables (for probabilities and expected values) and binary variables (for action and vector selection).

\paragraph{Decision variables.}
The MILP defines the following key variables:
\begin{itemize}
    \item \textbf{Leader decision rule:} 
    \( \delta_{\colorL}(a_{\colorL} \mid h_{\colorL}) \), defined for all \( a_{\colorL} \in A_{\colorL} \) and \( h_{\colorL} \in \bigcup_{o'\in\tilde{c}} H_{\colorL}(o') \), where
    \[
    H_{\colorL}(o') \doteq \left\{ h_{\colorL} \in H_{\colorL} ~\middle|~ \sum_{s,h_{\colorF}} o'(s,h_{\colorL},h_{\colorF}) > 0 \right\}.
    \]
    
    \item \textbf{Expected values:}
    \begin{itemize}
        \item \( g^{\Gamma}_{\colorL,t}(o_{h_{\colorF}}, \delta_{\colorL}) \): expected leader return at \( o_{h_{\colorF}} \in \mathbb{C}_{\colorF}(o) \),
        \item \( g^{\Gamma}_{\colorF,t}(o'_{h_{\colorF}}, \delta_{\colorL}) \): expected follower return at \( o'_{h_{\colorF}} \in \mathbb{C}_{\colorF}(\tilde{c}) \),
        \item \( \beta^{\Gamma}_{\colorI,t}(o_{h_{\colorF}}, \delta_{\colorL}, a_{\colorF}, s') \): expected return of player \( \colorI \in \{ \colorL, \colorF \} \) when executing \(a_{\colorF}\) and transitioning to state \(s' \in S\).
    \end{itemize}
    
    \item \textbf{Binary selectors:}
    \begin{itemize}
        \item \( \kappa^{\Gamma,o}_\colorF(o_{h_{\colorF}}, a_{\colorF}) \in \{0,1\} \): indicates selected follower action in each \( o_{h_{\colorF}} \in \mathbb{C}_{\colorF}(\tilde{c}) \),
        \item \( w^{\Gamma,o}_\colorF(o_{h_{\colorF}}, s', \alpha) \in \{0,1\} \): selects the \(\alpha\)-vector used to approximate value continuation at state \(s'\).
    \end{itemize}
\end{itemize}

\paragraph{Constraint semantics.}
We now explain how the constraints enforce the structure of the backup:

\begin{itemize}
    \item \textbf{Leader value maximisation (Constraint \ref{appendix:milp:constraint:1}):}  
    \[
    \max_{\delta_{\colorL}} q_{\colorL,t}^{\Gamma}(o, \delta_{\colorL}).
    \]
    This defines the objective of the MILP.

    \item \textbf{Follower best response (Constraint \ref{appendix:milp:constraint:2}):}  
    \[
    q_{\colorF,t}^{\Gamma}(o, \delta_{\colorL}) \ge q_{\colorF,t}^{\Gamma}(o', \delta_{\colorL}) \quad \forall o' \in \tilde{c}.
    \]
    The follower must prefer the policy induced at \(o\), enforcing Stackelberg equilibrium consistency.

    \item \textbf{Follower value representation (Constraints \ref{appendix:milp:constraint:5}--\ref{appendix:milp:constraint:7}):}
    These jointly enforce:
    \[
    g^{\Gamma}_{\colorF,t}(o_{h_{\colorF}}, \delta_{\colorL}) = \max_{a_{\colorF} \in A_{\colorF}} \sum_{s'} \beta^{\Gamma}_{\colorF,t}(o_{h_{\colorF}}, \delta_{\colorL}, a_{\colorF}, s').
    \]
    This is implemented via:
    \begin{itemize}
        \item Constraint \ref{appendix:milp:constraint:6} (lower bound, always active),
        \item Constraint \ref{appendix:milp:constraint:7} (upper bound, only active for selected action),
        \item Constraint \ref{appendix:milp:constraint:8} (exactly one action selected).
    \end{itemize}

    \item \textbf{Follower value via \(\alpha_{\colorF}\)-vectors (Constraints \ref{appendix:milp:constraint:9}--\ref{appendix:milp:constraint:11}):}
    These define:
    \[
    \beta^{\Gamma}_{\colorF,t}(o_{h_{\colorF}}, \delta_{\colorL}, a_{\colorF}, s') = \max_{\alpha \in \Gamma} \left\{ \nu_{\colorF} + \gamma \alpha_{\colorF}(\tau) \right\},
    \]
    where \(\tau = \tau_{s'}(o_{h_{\colorF}}, \delta_{\colorL}, a_{\colorF})\). Constraint \ref{appendix:milp:constraint:11} enforces that exactly one \(\alpha\)-vector is used per state.

    \item \textbf{Leader value under consistent follower response (Constraints \ref{appendix:milp:constraint:12}--\ref{appendix:milp:constraint:13}):}
    These encode:
    \[
    g^{\Gamma}_{\colorL,t}(o_{h_{\colorF}}, \delta_{\colorL}) = \max_{a_{\colorF}} \left\{
    \sum_{s'} \beta^{\Gamma}_{\colorL,t}(o_{h_{\colorF}}, \delta_{\colorL}, a_{\colorF}, s')
    ~\middle|~
    a_{\colorF} \text{ selected for } g^{\Gamma}_{\colorF,t}
    \right\}.
    \]

    \item \textbf{Leader value via matched \(\alpha_{\colorL}\)-vectors (Constraints \ref{appendix:milp:constraint:14}--\ref{appendix:milp:constraint:15}):}
    These enforce:
    \[
    \beta^{\Gamma}_{\colorL,t}(o_{h_{\colorF}}, \delta_{\colorL}, a_{\colorF}, s') =
    \max_{\alpha \in \Gamma} \left\{
    \nu_{\colorL} + \gamma \alpha_{\colorL}(\tau)
    ~\middle|~
    \text{matching } \alpha_{\colorF}(\tau) \text{ in follower value}
    \right\}.
    \]
\end{itemize}

\paragraph{Conclusion.}
This MILP operationalises the greedy dynamic programming step for the leader in an LF-GSSG. It combines leader optimisation, follower rationality, and joint value approximation via \(\alpha\)-vector selection. The binary variables encode deterministic follower responses and consistent value decompositions, enabling scalable and interpretable value backups over filtered credible sets.

\begin{theorem}
\label{appendix:proof:thm:pbvi-backup}
Let \( \Lambda_{t+1} \) be a finite collection of sets of functions \(\Gamma\) approximating the optimal leader value function \( v_{\colorL,t+1}^*\colon \tilde{C}  \to \mathbb{R}\) at stage \( t{+}1 \).  
Then, for any credible set \( \tilde{c} \) at stage \( t \), the greedy leader decision rule \( \delta_{\colorL}^{\tilde{c}} \) maximising the stage-\(t\) leader action-value is the solution of a mixed-integer linear program (see Table \ref{appendix:eq:milp-backup}).
\end{theorem}

\begin{proof}
The proof starts with the definition of the greedy leader decision-rule \(\delta_{\colorL}^{\tilde{c}}\) selection at credible set \(\tilde{c}\), assuming we have access to collection \(\Lambda_{t+1}\) of sets \(\Gamma\) approximating \(v_{\colorL}^*\) as leader value function \(v_{\colorL}\). That is,
\begin{align*}
    \delta_{\colorL}^{\tilde{c}}
    &\in  
    \argmax_{\delta_{\colorL}\in \Delta_{\colorL}} 
    v_{\colorL,t+1}(\tilde{T}(\tilde{c},\delta_{\colorL}))\\ 
    &\in \argmax_{\delta_{\colorL}\in \Delta_{\colorL}} 
    q_{\colorL,t}(\tilde{c},\delta_{\colorL}).
\end{align*}
By the application of uniform continuity property, see Lemma \ref{appendix:lem:convexity}, the following holds:
\begin{align}
    \delta_{\colorL}^{\tilde{c}}
    &
    \in 
    \argmax
    \{ 
    q_{\colorL,t}^{\Gamma}(o,\delta_{\colorL}) \mid  
    \{
    q_{\colorF,t}^{\Gamma}(o,\delta_{\colorL}) 
    \geq 
    q^{\Gamma}_{\colorF,t}(o\prime,\delta_{\colorL})
    \mid 
    \forall o\prime\in \tilde{c}
    \},
    \forall o\in \tilde{c},
    \delta_{\colorL}\in \Delta_{\colorL},
    \Gamma\in \Lambda_{t+1}
    \} \nonumber
    \\
    &q_{\colorI,t}^{\Gamma}(o,\delta_{\colorL}) 
    \doteq
    \textstyle 
    \rho_{\colorI}(o)
    +
    \gamma^t
    \sum_{h_{\colorF}} 
    \Pr(h_{\colorF}\mid o) 
    \cdot
    g_{\colorI,t}^{\Gamma}(o_{h_{\colorF}},\delta_{\colorL}),
    ~\forall \colorI \in I
    \nonumber\\
    &
    g_{\colorI,t}^{\Gamma}(o_{h_{\colorF}},\delta_{\colorL}) \doteq 
    \max_{a_{\colorF}\in \bar{A}_{\colorF}(o_{h_{\colorF}},\delta_{\colorL})}
     \sum_{s'\in S}
    \beta_{\colorI,t}^{\Gamma}(o_{h_{\colorF}},\delta_{\colorL},a_{\colorF},s'),
    ~\forall \colorI \in I.
    \nonumber\\
    &
    \beta_{\colorI,t}^{\Gamma}(o_{h_{\colorF}},\delta_{\colorL}, a_{\colorF},s') \doteq 
    \mathbb{E}[r_\colorI(s, a_{\colorL},a_{\colorF},s')\mid o_{h_{\colorF}},\delta_{\colorL},a_{\colorF},s']
    +
    \gamma
    \max_{\alpha\in \bar{\Gamma}(\tau_{s'}(o_{h_{\colorF}},\delta_{\colorL},a_{\colorF}))}
    \alpha_{\colorI}(\tau_{s'}(o_{h_{\colorF}},\delta_{\colorL},a_{\colorF})),
    ~\forall \colorI \in I.
    \nonumber     %
\end{align}
where \(\bar{\Gamma}(o_{h_{\colorF}}) \doteq \argmax_{\alpha'\in \Gamma} \alpha'_{\colorF}(o_{h_{\colorF}}) \) and \(\bar{A}_{\colorF}(o_{h_{\colorF}},\delta_{\colorL}) \doteq \argmax_{a'_{\colorF}\in A_{\colorF}} \sum_{s'}\beta_{\colorF,t}^{\Gamma}(o_{h_{\colorF}},\delta_{\colorL}, a'_{\colorF},s')\).
Define the expected immediate payoff for agent \(\colorI\), conditional occupancy state \(o_{h_\colorF}\), leader--follower decisions \((\delta_\colorL, a_\colorF)\) as follows:
\begin{align*}
\nu_\colorI(o_{h_\colorF}, \delta_\colorL, a_\colorF,s')
&\doteq
\mathbb{E}[r_\colorI(s,a_{\colorL},a_{\colorF},s') \mid o_{h_{\colorF}},\delta_{\colorL},a_{\colorF},s']
\nonumber\\
&=
\sum_{s\in S} 
\sum_{h_\colorL \in H_\colorL(o_{h_\colorF})}
o_{h_\colorF}(s,h_\colorL)
\sum_{a_\colorL\in A_\colorL}
\delta_\colorL(a_\colorL \mid h_\colorL) 
\cdot
p(s'\mid s, a_\colorL, a_\colorF) \cdot
r_i(s,a_\colorL,a_\colorF,s').
\nonumber
\end{align*}
Define the expected payoff for agent \(\colorI\), "value function" \(\alpha\), current conditional occupancy \(o_{h_\colorF}\), leader-follower decisions (\(\delta_\colorL, a_\colorF\)), next state \(s'\) as follows:
\begin{align}
    &
    \alpha_i(\tau_{s'}(o_{h_\colorF}, \delta_\colorL, a_\colorF)) \doteq
    \sum_{h_\colorL \in H_\colorL(o_{h_\colorF})}
    \sum_{a_\colorL \in a_\colorL}
    \tau_{s'}(o_{h_\colorF}, \delta_\colorL,a_\colorF)(s',(h_\colorL, a_\colorL, s')) \cdot \alpha_i(s',(h_\colorL, a_\colorL, s')),
    \nonumber\\
    &
    \forall \alpha \in \Gamma, \forall \colorI \in I, \forall s' \in S, \forall \tau_{s'}(o_{h_\colorF}, \delta_\colorL, a_\colorF) \in O.
    \nonumber\\
\end{align}
Where the probability of being in next state \(s'\), leader next history being (\(h_\colorL, a_\colorL, s'\)) coming from any conditional occupancy state \(o_{h_\colorF}\) is the following
\begin{align}
    &
    \tau_{s'}(o_{h_\colorF}, \delta_\colorL, a_\colorF)(s',(h_\colorL, a_\colorL, s'))
    =
    \sum_{s\in S}
    o_{h_\colorF}(s(h_\colorL),h_\colorL)
    \cdot
    \delta_\colorL(a_\colorL \mid h_\colorL)
    \cdot 
    p(s' \mid s(h_\colorL),a_\colorL, a_\colorF)
    \nonumber\\
    &, \forall h_\colorL \in H_\colorL(o_{h_\colorF}), \forall a_\colorL \in A_\colorL, \forall s' \in S.
    \nonumber\\
\end{align}
Re-arranging terms results in the following :
\begin{align}
    &
    \delta_{\colorL}^{\tilde{c}}
    \in  
    \textstyle 
    \arg\max_{\delta_{\colorL}\in \Delta_{\colorL}} 
    \max_{o\in \tilde{c}}
    \max_{\Gamma\in \Lambda_{t+1}}
    q_{\colorL,t}^{\Gamma}(o,\delta_{\colorL}) \nonumber
    \\
    &
    \text{subject to}\colon
    \textstyle 
    q_{\colorF,t}^{\Gamma}(o,\delta_{\colorL}) 
    =
    \max_{o\prime\in \tilde{c}}
    q^{\Gamma}_{\colorF,t}(o\prime,\delta_{\colorL}),
    \nonumber
    \\
    &q_{\colorL,t}^{\Gamma}(o,\delta_{\colorL}) =
    \textstyle
    \rho_{\colorL}(o)
    +
    \gamma^t
    \sum_{h_{\colorF}} 
    \Pr(h_{\colorF}\mid o)
    \cdot 
    g_{\colorL,t}^{\Gamma}(o_{h_{\colorF}},\delta_{\colorL}) 
   \nonumber
    \\
    &q_{\colorF,t}^{\Gamma}(o\prime,\delta_{\colorL}) =
    \textstyle
    \rho_{\colorF}(o\prime)
    +
    \gamma^t
    \sum_{h_{\colorF}}
    \Pr(h_{\colorF}\mid o\prime)
    \cdot 
    g_{\colorF,t}^{\Gamma}(o\prime_{h_{\colorF}},\delta_{\colorL}) 
    \nonumber
    ~
    \forall o\prime \in \tilde{c}
    \nonumber
    \\
    &
    g_{\colorF,t}^{\Gamma}(o\prime_{h_{\colorF}},\delta_{\colorL}) =
    \max_{a_{\colorF}\in A_{\colorF}}
    \{
    \sum_{s'}
     \beta_{\colorF,t}^\Gamma(o\prime_{h_{\colorF}},\delta_{\colorL},a_{\colorF},s')
     \},
     ~
     \forall 
     o\prime_{h_{\colorF}}\in  \mathbb{C}_{\colorF}(\tilde{c})
     \nonumber\\
    &
    \beta_{\colorF,t}^{\Gamma}(o\prime_{h_{\colorF}},\delta_{\colorL},a_{\colorF},s') =
    \max_{\alpha\in \Gamma}
    \{
    \nu_\colorF(o\prime_{h_\colorF}, \delta_\colorL, a_\colorF,s')
    +
    \gamma  
     \alpha_{\colorF}(\tau_{s'}(o\prime_{h_{\colorF}},\delta_{\colorL},a_{\colorF}))
     \}
     ,
     ~
     \forall 
     s'\in S, o\prime_{h_{\colorF}}\in  \mathbb{C}_{\colorF}(\tilde{c})
     \nonumber\\
    &
    g_{\colorL,t}^{\Gamma}(o_{h_{\colorF}},\delta_{\colorL}) =
    \max_{a_{\colorF}\in A_{\colorF}}  
     \left\{ 
     \sum_{s'}
     \beta_{\colorL,t}^\Gamma(o_{h_{\colorF}},\delta_{\colorL},a_{\colorF},s')
     \mid
     \sum_{s'}
     \beta_{\colorF,t}^\Gamma(o_{h_{\colorF}},\delta_{\colorL},a_{\colorF},s')
     =
     g_{\colorF,t}^{\Gamma}(o_{h_{\colorF}},\delta_{\colorL})
     \right\},
     ~
     \forall 
     o_{h_{\colorF}}\in  \mathbb{C}_{\colorF}(o)
     \nonumber\\
    &\beta_{\colorL,t}^{\Gamma}(o_{h_{\colorF}},\delta_{\colorL},a_{\colorF},s') 
    =
    \max_{\alpha\in \Gamma}
    \left\{ 
    \nu_\colorL(o_{h_\colorF}, \delta_\colorL, a_\colorF,s')
    +
    \gamma
     \alpha_{\colorL}(\tau_{s'}(o_{h_{\colorF}},\delta_{\colorL},a_{\colorF})) 
     \mid 
     \beta_{\colorF,t}^{\Gamma}(o_{h_{\colorF}},\delta_{\colorL},a_{\colorF},s')
    \right. 
    \nonumber\\
    &
    \qquad \qquad
    \left.
         =  
         \nu_\colorF(o_{h_\colorF}, \delta_\colorL, a_\colorF,s')
    +
    \gamma
     \alpha_{\colorF}(\tau_{s'}(o_{h_{\colorF}},\delta_{\colorL},a_{\colorF})) 
    \right\},
         ~
     \forall s'\in S, a_{\colorF}\in A_{\colorF},  o_{h_{\colorF}}\in  \mathbb{C}_{\colorF}(o). \nonumber
\end{align}

Fix occupancy state \(o\in \tilde{c}\) and value-vector set \(\Gamma\in \Lambda_{t+1}\), the greedy leader decision rule \(\delta_{\colorL}^{\Gamma,o}\). Let \[\{\kappa^{\Gamma,o}_\colorF(o\prime_{h_{\colorF}},a_{\colorF}) \mid o\prime_{h_{\colorF}} \in \mathbb{C}_\colorF(\tilde{c}), a_{\colorF}\in A_{\colorF}\}\] be a set of binary variables \(\kappa^{\Gamma,o}_\colorF(o\prime_{h_{\colorF}},a_{\colorF})\) representing the follower decision rule for occupancy state \(o\in \tilde{c}\) and value-vector set \(\Gamma\in \Lambda_{t+1}\). Similarly, let \[\{w^{\Gamma,o}_\colorF(o_{h_{\colorF}},s',\alpha) \mid o_{h_{\colorF}} \in \mathbb{C}_\colorF(\tilde{c}), s'\in S, \alpha\in \Gamma\}\] be the set of binary variables \(w^{\Gamma,o}_\colorF(o_{h_{\colorF}},s',\alpha)\) representing the follower policy upon seeing \(s'\in S\) after taking an action in conditional occupancy state \( o_{h_{\colorF}} \in \mathbb{C}_\colorF(\tilde{c})\). These binary variables help reformulating maximization as inequality constraints:
\begin{align}
    &
    \delta_{\colorL}^{\Gamma,o}
    \in  
    \textstyle 
    \argmax_{\delta_{\colorL}\in \Delta_{\colorL}} 
    q_{\colorL,t}^{\Gamma}(o,\delta_{\colorL}) \nonumber
    \\
    &
    \text{subject to}\colon
    \textstyle 
    q_{\colorF,t}^{\Gamma}(o,\delta_{\colorL}) 
    \ge
    q^{\Gamma}_{\colorF,t}(o\prime,\delta_{\colorL}),
    ~
    \forall o\prime\in \tilde{c}
    \nonumber
    \\
    &q_{\colorL,t}^{\Gamma}(o,\delta_{\colorL}) =
    \textstyle
    \rho_{\colorL}(o)
    +
    \gamma^t
    \sum_{h_{\colorF}} 
    \Pr(h_{\colorF}\mid o)
    \cdot 
    g_{\colorL,t}^{\Gamma}(o_{h_{\colorF}},\delta_{\colorL})
   \nonumber\\
    &q_{\colorF,t}^{\Gamma}(o\prime,\delta_{\colorL}) =
    \textstyle
    \rho_{\colorF}(o\prime)
    +
    \gamma^t
    \sum_{h_{\colorF}} 
    \Pr(h_{\colorF}\mid o\prime)
    \cdot 
    g_{\colorF,t}^{\Gamma}(o\prime_{h_{\colorF}},\delta_{\colorL}),
    \forall
     o\prime\in \tilde{c}
     \nonumber\\
    &
    g^\Gamma_{\colorF,t}(o\prime_{h_{\colorF}},\delta_\colorL) 
    \ge  
    \sum_{s'\in S} 
    \beta^\Gamma_{\colorF,t}(o\prime_{h_{\colorF}},\delta_\colorL,a_\colorF,s'),
    ~
    \forall a_\colorF \in A_\colorF, \forall o\prime_{h_{\colorF}}\in \mathbb{C}_\colorF(\tilde{c})
    \nonumber\\
    &
    g^\Gamma_{\colorF,t}(o\prime_{h_{\colorF}},\delta_\colorL) 
    \le  
    \sum_{s'\in S} 
    \beta^\Gamma_{\colorF,t}(o\prime_{h_{\colorF}},\delta_\colorL,a_\colorF,s') 
    +
    \mathtt{M}
    \cdot  
    (1 - \kappa^{\Gamma,o}_\colorF(o\prime_{h_{\colorF}},a_{\colorF})),
    ~
    \forall a_{\colorF}\in A_{\colorF}, o\prime_{h_{\colorF}}\in \mathbb{C}_\colorF(\tilde{c})
    \nonumber\\
    &
    \sum_{a_{\colorF}\in A_{\colorF}}
    \kappa^{\Gamma,o}_\colorF(o\prime_{h_{\colorF}},a_{\colorF}))
    = 1,
    ~\forall o\prime_{h_{\colorF}}\in \mathbb{C}_\colorF(\tilde{c})
    \nonumber\\
    &
    \beta_{\colorF,t}^{\Gamma}(o\prime_{h_{\colorF}},\delta_{\colorL},a_{\colorF},s') 
    \ge
    \nu_\colorF(o\prime_{h_\colorF}, \delta_\colorL, a_\colorF, s')
    +
    \gamma 
     \alpha_{\colorF}(\tau_{s'}(o\prime_{h_{\colorF}},\delta_{\colorL},a_{\colorF})),
     ~
     \forall 
     s'\in S, a_{\colorF}\in A_{\colorF}, o\prime_{h_{\colorF}}\in  \mathbb{C}_{\colorF}(\tilde{c}), \alpha\in \Gamma
     \nonumber\\
    &
    \beta_{\colorF,t}^{\Gamma}(o\prime_{h_{\colorF}},\delta_{\colorL},a_{\colorF},s') 
    \le
    \nu_\colorF(o\prime_{h_\colorF}, \delta_\colorL, a_\colorF,s')
    +
    \gamma 
     \alpha_{\colorF}(\tau_{s'}(o\prime_{h_{\colorF}},\delta_{\colorL},a_{\colorF}))
     +
    \mathtt{M}
    \cdot  
    (1 - w^{\Gamma,o}_\colorF(o\prime_{h_{\colorF}},s',\alpha))
    \nonumber\\
    &\qquad
     \forall 
     s'\in S, a_{\colorF}\in A_{\colorF}, o\prime_{h_{\colorF}}\in  \mathbb{C}_{\colorF}(\tilde{c}), \alpha\in \Gamma
     \nonumber\\
    &
    \sum_{\alpha\in \Gamma}
    w^{\Gamma,o}_\colorF(o\prime_{h_{\colorF}},s',\alpha)
    = 1,
    ~\forall s'\in S, o\prime_{h_{\colorF}}\in \mathbb{C}_\colorF(\tilde{c})
    \nonumber\\
    &
    g_{\colorL,t}^{\Gamma}(o_{h_{\colorF}},\delta_{\colorL}) 
    \ge
     \sum_{s'}
     \beta_{\colorL,t}(o_{h_{\colorF}},\delta_{\colorL},a_{\colorF},s')
     -
     \mathtt{M}
     \cdot 
     (1 - \kappa^{\Gamma,o}_\colorF(o_{h_{\colorF}},a_{\colorF}))
     ,
     ~
     \forall
     a_{\colorF}\in A_{\colorF},
     o_{h_{\colorF}}\in  \mathbb{C}_{\colorF}(o)
     \nonumber\\
    &
    g_{\colorL,t}^{\Gamma}(o_{h_{\colorF}},\delta_{\colorL}) 
    \le
     \sum_{s'}
     \beta_{\colorL,t}(o_{h_{\colorF}},\delta_{\colorL},a_{\colorF},s')
     +
     \mathtt{M}
     \cdot 
     (1 - \kappa^{\Gamma,o}_\colorF(o_{h_{\colorF}},a_{\colorF}))
     ,
     ~
     \forall
     a_{\colorF}\in A_{\colorF},
     o_{h_{\colorF}}\in  \mathbb{C}_{\colorF}(o)
     \nonumber\\
     &
     \beta_{\colorL,t}^{\Gamma}(o_{h_{\colorF}},\delta_{\colorL},a_{\colorF},s') 
    \ge
    \nu_\colorL(o_{h_\colorF}, \delta_\colorL, a_\colorF,s')
    +
    \gamma
     \alpha_{\colorL}(\tau_{s'}(o_{h_{\colorF}},\delta_{\colorL},a_{\colorF})) 
     -
     \mathtt{M}
     \cdot 
     (1 - w^{\Gamma,o}_\colorF(o_{h_{\colorF}},s',\alpha)),
     \nonumber\\
    &\qquad
     \forall s'\in S, a_{\colorF}\in A_{\colorF},  o_{h_{\colorF}}\in  \mathbb{C}_{\colorF}(o), \alpha\in \Gamma. \nonumber\\
     &
     \beta_{\colorL,t}^{\Gamma}(o_{h_{\colorF}},\delta_{\colorL},a_{\colorF},s') 
    \le
    \nu_\colorL(o_{h_\colorF}, \delta_\colorL, a_\colorF,s')
    +
    \gamma
     \alpha_{\colorL}(\tau_{s'}(o_{h_{\colorF}},\delta_{\colorL},a_{\colorF})) 
     +
     \mathtt{M}
     \cdot 
     (1 - w^{\Gamma,o}_\colorF(o_{h_{\colorF}},s',\alpha)),
     \nonumber\\
    &\qquad
     \forall s'\in S, a_{\colorF}\in A_{\colorF},  o_{h_{\colorF}}\in  \mathbb{C}_{\colorF}(o), \alpha\in \Gamma.      \nonumber
\end{align}
Which ends the proof.
\end{proof}

For each filtered credible set \(\tilde{c} \in \tilde{C}'\) from a sample set of filtered credible sets \(\tilde{C}'\subset \tilde{C}\), the greedy leader decision rule is extracted from the solution of the mixed-integer linear program in Theorem \ref{appendix:proof:thm:pbvi-backup} as follows \(\delta_{\colorL}^{\tilde{c}} \in \arg\max_{\delta_{\colorL}\in \Delta_{\colorL}}   \max_{o\in \tilde{c}, \Gamma\in \Lambda_{t+1}} 
q_{\colorL,t}^{\Gamma}(o,\delta_{\colorL})\).
Once  \( \delta_{\colorL}^{\tilde{c}} \) has been computed for each filtered credible set \(\tilde{c}\), we now can updated the collection \(\Lambda_t\) with the set \(\Gamma^{\tilde{c}}\) computing induced by filtered credible set \(\tilde{c}\) and leader decision rule \(\delta_{\colorL}^{\tilde{c}}\). 

\begin{corollary}
Let \(v_{\colorL,t}\) denote the current leader value function at stage \(t\), represented by the collection \(\Lambda_t\). Let \(\tilde{C}' \subset \tilde{C}\) a sample set of filtered credible sets. Let \(\tilde{c}\in \tilde{C}'\) be a filtered credible set, and let \(( \delta_{\colorL}^{\tilde{c}}, \Gamma^{\tilde{c}}, w_{\colorF}^{\Gamma^{\tilde{c}}}, \kappa_{\colorF}^{\Gamma^{\tilde{c}}})\) be the solution to the mixed-integer linear programme in Theorem~\ref{appendix:proof:thm:pbvi-backup} at filtered credible set \(\tilde{c}\). Define the updated leader value function \(v'_{\colorL,t}\) at stage \(t\) by augmenting \(\Lambda_t\) with a new set \(\Gamma\) of value vectors \((\alpha_{\colorL}^{o_{h_{\colorF}}}, \alpha_{\colorF}^{o_{h_{\colorF}}})\), each linear over conditional occupancy states \(o_{h_{\colorF}} \in \mathbb{C}_\colorF(\tilde{c})\), where: \(\forall \colorI\in I, s\in S, h_{\colorL}\in \cup_{o\in \tilde{c}} H_{\colorL}(o)\),
\begin{align*}
    \alpha_{\colorI}^{o_{h_{\colorF}}}(s,h_{\colorL}) &\doteq 
    \sum_{s'\in S}
    \sum_{a_{\colorF} \in A_{\colorF}}
    \kappa_{\colorF}^{\Gamma^{\tilde{c}}}(o_{h_{\colorF}},a_{\colorF})
    \sum_{a_{\colorL}\in A_{\colorL}}
    \delta_{\colorL}^{\tilde{c}}(a_{\colorL} | h_{\colorL})
    \cdot 
    p(s'\mid s, a_{\colorL}, a_{\colorF}) 
    \nonumber\\
    &\qquad 
    \left(
    r_\colorI(s, a_{\colorL}, a_{\colorF},s')
    +
    \gamma 
    \sum_{\alpha \in \Gamma^{\tilde{c}}}
    w^{\Gamma^{\tilde{c}}}_{\colorF}(o_{h_{\colorF}}, s', \alpha) \cdot 
    \alpha_{\colorI}(s', h_{\colorL} a_{\colorL} s')
    \right).
\end{align*}
Then the updated value function satisfies \(v'_{\colorL,t}(\tilde{c}) \geq v_{\colorL,t}(\tilde{c})\) for all filtered credible sets \(\tilde{c}\in \tilde{C}'\), and \(v'_{\colorL,t}(\tilde{c}) > v_{\colorL,t}(\tilde{c})\) for at least one such \(\tilde{c}\) whenever the greedy update yields a strict improvement.
\end{corollary}

%% file: appendix/pbvi.tex
This appendix provides the algorithmic details for computing an approximate strong Stackelberg equilibrium (SSE) in the original leader–follower game \( M \) using point-based value iteration (PBVI) over the reduced credible MDP \( \mathcal{M}' \). This method performs dynamic programming over filtered credible sets, constructing value function approximations using vector pairs that encode leader and follower returns while maintaining Stackelberg-consistent best responses at every step.

\subsection{Overview of the PBVI Algorithm}

\Cref{alg:pbvi-credible-mdp} outlines the PBVI loop. The algorithm maintains:

\begin{itemize}
    \item a sequence of filtered credible sets \( \tilde{C}'_0, \ldots, \tilde{C}'_\ell \), representing reachable belief supports up to horizon \( \ell \),
    \item corresponding value approximations \( \Lambda_0, \ldots, \Lambda_\ell \), where each \( \Lambda_t \) is a collection of vector sets \( \Gamma \subseteq \mathbb{R}^{|\mathbb{C}_{\colorF}(\tilde{c})|} \times \mathbb{R}^{|\mathbb{C}_{\colorF}(\tilde{c})|} \) approximating the leader and follower value functions over conditional occupancy states.
\end{itemize}

\begin{algorithm}[H]
  \caption{Computing an approximation of a SSE in $M$ via Point-Based Value Iteration (PBVI) for \(M'\).}
  \label{alg:pbvi-credible-mdp}
  \begin{algorithmic}[1]
    \STATE {\bf function} \(\mathtt{solveSSE}(M)\)
    \STATE \(\tilde{C}'_0 \gets \{c_0\}\) (initial set of credible sets)
    \STATE \(\tilde{C}'_1,\ldots,\tilde{C}'_\ell \gets \emptyset\) (sample sets of credible sets) 
    \STATE \(\Lambda_0,\ldots,\Lambda_\ell \gets \emptyset\) (collections of sets of vector pairs)
    \WHILE{not converged}
      \FOR{\(t = 0\) to \(\ell-1\)} \STATE \(\tilde{C}'_{t+1} \gets \mathtt{expand}(\tilde{C}'_t, \tilde{C}'_{t+1})\) 
      \ENDFOR
      \FOR{\(t = \ell-1\) to \(0\)} 
        \STATE \(\Lambda_{t} \gets  \mathtt{backup}(\tilde{C}'_t, \Lambda_{t+1})\) 
      \ENDFOR
      \STATE \((\Lambda_{t}, \tilde{C}'_t)\gets \mathtt{pruning}(\Lambda_{t}, \tilde{C}'_t)\)
    \ENDWHILE
    \RETURN \(\mathtt{extractSSE}(\Lambda_0,\ldots,\Lambda_\ell)\).
  \end{algorithmic}
\end{algorithm}

Each PBVI iteration alternates between three phases:
\begin{enumerate}
    \item \textbf{Expansion} (\Cref{alg:pbvi-expand}): simulate reachable belief sets,
    \item \textbf{Backup} (\Cref{alg:pbvi-backup}): update value approximations using MILPs,
    \item \textbf{Pruning} (\Cref{alg:pruning-vectors}, \Cref{alg:pruning-states}): reduce redundant vectors and belief points.
\end{enumerate}

The process terminates once value approximations stabilise and an SSE policy can be extracted.

\subsection{Phase 1: Expansion of Reachable Credible Sets}
\label{appendix:pbvi:expand}

\Cref{alg:pbvi-expand} defines the expansion routine \( \mathtt{expand} \), which constructs the next-stage set \( \tilde{C}'_{t+1} \) by sampling possible transitions from each \( \tilde{c} \in \tilde{C}'_t \).

\begin{algorithm}[H]
  \caption{The \(\mathtt{expand}\) phase for PBVI in \( M' \)}
  \label{alg:pbvi-expand}
  \begin{algorithmic}[1]
    \STATE {\bf function} \(\mathtt{expand}(\tilde{C}'_t, \tilde{C}'_{t+1})\) 
    \FOR{each \(\tilde{c} \in \tilde{C}'_t\)}
        \STATE Filtered credible set \(\tilde{c}' \gets \emptyset\) and set \(\mathbb{C}_\colorF(\tilde{c}') \gets \emptyset\)
        \STATE Sample leader decision rule \(\delta_{\colorL} \in \Delta_{\colorL}\)
        \FOR{each occupancy state \(o \in \tilde{c}\)}
            \STATE Initialise \(\delta_{\colorF}\gets \emptyset\)
            \STATE \(\isInteresting{\delta_{\colorF}} \gets \FALSE \)
            \FOR{each \(h_{\colorF} \in H_{\colorF}\) reachable from \(o\)}
                \STATE Sample  \(\delta_{\colorF}(h_{\colorF}) \sim \mathtt{uniform}(A_{\colorF})\) 
                \FOR{each \(s' \in S\)}
                    \STATE  \(o' \gets \tau_{s'}(o_{h_{\colorF}}, \delta_{\colorL}, \delta_{\colorF}(h_{\colorF}))\)
                    \IF{ \(\mathbb{C}_\colorF(\tilde{c}')=\emptyset\) or \(\max_{\bar{o}'\in \mathbb{C}_\colorF(\tilde{c}')} \| \bar{o}' - o' \|_1 > \varepsilon\)}
                        \STATE \(\isInteresting{\delta_{\colorF}} \gets \TRUE \)
                        \STATE \(\mathbb{C}_\colorF(\tilde{c}') \gets \mathbb{C}_\colorF(\tilde{c}') \cup \{o'\} \)
                    \ENDIF
                \ENDFOR
            \ENDFOR
            \IF{\(\isInteresting{\delta_{\colorF}}  \)}   
                \STATE \(\tilde{c}' \gets \tilde{c}' \cup \{\tau(o, \delta_{\colorL}, \delta_{\colorF})\}\)
            \ENDIF
        \ENDFOR
        \IF{\(\tilde{C}'_{t+1} = \emptyset\) or \( \max_{\tilde{c}''\in \tilde{C}'_{t+1}} d_H(\tilde{c}'', \tilde{c}')>\varepsilon\)}
            \STATE \(\tilde{C}'_{t+1} \gets \tilde{C}'_{t+1} \cup \{\tilde{c}'\} \)
        \ENDIF
    \ENDFOR
    \RETURN \(\tilde{C}'_{t+1}\)
  \end{algorithmic}
\end{algorithm}

For each \( o \in \tilde{c} \):
\begin{itemize}
    \item A leader decision rule \( \delta_{\colorL} \in \Delta_{\colorL} \) is sampled.
    \item Follower policies \( \delta_{\colorF} \) are sampled by selecting actions uniformly at random for each reachable follower-private history \( h_{\colorF} \).
    \item The resulting conditional occupancy states \( o' \in \mathbb{C}_{\colorF}(\tilde{c}') \) are computed using the environment transition \( \tau_{s'}(o_{h_{\colorF}}, \delta_{\colorL}, \delta_{\colorF}(h_{\colorF})) \) for all \( s' \in S \).
    \item Only sufficiently novel conditional states \( o' \) are retained (based on \(\ell_1\)-norm or Hausdorff distance), ensuring diversity of belief support.
\end{itemize}

This phase efficiently expands the reachable belief space while avoiding redundancy.

\subsection{Phase 2: Backing Up the Value Function}
\label{appendix:pbvi:backup}

\Cref{alg:pbvi-backup} performs value backups at each filtered credible set \( \tilde{c} \in \tilde{C}'_t \), using the current approximation \( \Lambda_{t+1} \).

\begin{algorithm}[H]
  \caption{The \(\mathtt{backup}\) phase for PBVI in \( \mathcal{M}' \)}
  \label{alg:pbvi-backup}
  \begin{algorithmic}[1]
    \STATE {\bf function} \(\mathtt{backup}(\tilde{C}'_t,\Lambda_{t+1})\)
    \FOR{ each filtered credible set \( \tilde{c} \in \tilde{C}'_t \)}
        \FOR{each occupancy state \(o \in \tilde{c}\)}
            \FOR{each occupancy state \(\Gamma \in \Lambda_{t+1}\)}
                \STATE \(
                    (\delta_\colorL^{\Gamma,o}, w_\colorF^{\Gamma,o},\kappa_\colorF^{\Gamma,o}, v^{\Gamma,o})
                        \gets 
                    \mathtt{MILP}(\tilde{c}, o,\Gamma) \), (cf. Table \ref{appendix:eq:milp-backup})
                \STATE \(\Phi
                    \gets \Phi  \cup
                    \{  (\delta_\colorL^{\Gamma,o}, w_\colorF^{\Gamma,o},\kappa_\colorF^{\Gamma,o}, v^{\Gamma',o})\} \)
            \ENDFOR
        \ENDFOR
        \STATE \((\delta_\colorL^{\tilde{c}}, w_\colorF^{\tilde{c}},\kappa_\colorF^{\tilde{c}} ,v^{\tilde{c}}) \gets \argmax_{(\delta_\colorL^{\Gamma,o}, w_\colorF^{\Gamma,o}, \kappa_\colorF^{\Gamma,o}, v^{\Gamma,o}) \in \Phi}  v^{\Gamma,o}\)
        \STATE \(\Gamma^{\tilde{c}} \gets \mathtt{improve}(\tilde{c}, \delta_\colorL^{\tilde{c}}, w_\colorF^{\tilde{c}},\kappa_\colorF^{\tilde{c}} )\)
      \STATE \(\Lambda_t \gets \Lambda_t \cup \{\Gamma^{\tilde{c}}\} \)
    \ENDFOR
    \STATE pruning of \(\Lambda_t \) and  \(\tilde{C}'_t \)
    \RETURN \(\Lambda_t \)
  \end{algorithmic}
\end{algorithm}

For each \( o \in \tilde{c} \) and \( \Gamma \in \Lambda_{t+1} \):
\begin{itemize}
    \item The MILP in \Cref{appendix:eq:milp-backup} is solved to obtain:
    \begin{itemize}
        \item a greedy leader decision rule \( \delta_{\colorL}^{\Gamma,o} \),
        \item a consistent follower best response \( w_{\colorF}^{\Gamma,o} \),
        \item and the resulting leader value \( v^{\Gamma,o} \).
    \end{itemize}
    \item The best such triple is selected and used to construct a new vector set \( \Gamma^{\tilde{c}} \) via \( \mathtt{improve} \), which is then added to \( \Lambda_t \).
\end{itemize}

This step updates the value function while preserving SSE rationality through explicit best-response modelling.

\subsection{Phase 3: Pruning Redundant Representations}
\label{appendix:pbvi:pruning}

To ensure computational efficiency, two pruning steps are applied:

\paragraph{(a) Value pruning.}  
\Cref{alg:pruning-vectors} retains only vector sets \( \Gamma \in \Lambda \) that are optimal for some filtered credible set \( \tilde{c} \in \tilde{C}' \). This ensures that all retained \(\Gamma\) contribute to the current upper envelope of the value function.

\begin{algorithm}[H]
  \caption{Bounded Set Pruning}
  \label{alg:pruning-vectors}
  \begin{algorithmic}[1]
    \STATE {\bf function} $\mathtt{BoundedPruning}(\Lambda, \tilde{C}')$
    \FOR{each $\Gamma \in \Lambda$}
      \STATE $\isKeptInPruning\Gamma \gets \FALSE$
    \ENDFOR
    \FOR{each $\tilde{c} \in \tilde{C}'$}
      
      \STATE $\Gamma_x \gets  \argmax_{\Gamma \in \Lambda}  v_{\colorL}^{\Gamma}(\tilde{c})
      $
      \STATE $\isKeptInPruning{\Gamma_{\tilde{c}}} \gets \TRUE$
    \ENDFOR
    \STATE \textbf{return} $\{\Gamma \in \Lambda \mid \isKeptInPruning\Gamma\}$
  \end{algorithmic}
\end{algorithm}

\paragraph{(b) Belief support pruning.}  
\Cref{alg:pruning-states} removes filtered credible sets that are well approximated by existing ones. Specifically, a set \( \tilde{c} \) is retained only if the leader value induced by its optimal \(\Gamma_{\tilde{c}}\) differs from all others by more than a threshold \( \epsilon \).

\begin{algorithm}[H]
  \caption{Pruning Uncredible Sets}
  \label{alg:pruning-states}
  \begin{algorithmic}[1]
    \STATE {\bf function} $\mathtt{PruneUncredibleSets}(\Lambda, \tilde{C}', \epsilon)$
    \STATE Initialise $\tilde{C}^{\circ} \gets \emptyset$
    \FOR{each $\tilde{c} \in \tilde{C}'$}
      \STATE $\Gamma_{\tilde{c}} \gets \argmax_{\Gamma \in \Lambda}  v_{\colorL}^{\Gamma}(\tilde{c})$
    \ENDFOR
    \FOR{each $\tilde{c} \in \tilde{C}'$}
      \IF{{for all $\tilde{c}' \in \tilde{C}^{\circ}$, $\left| v_{\colorL}^{\Gamma_{\tilde{c}}}(\tilde{c}) -  v_{\colorL}^{\Gamma_{\tilde{c}'}}(\tilde{c}) \right| > \epsilon$}}
        \STATE $\tilde{C}^{\circ} \gets \tilde{C}^{\circ} \cup \{\tilde{c}\}$
      \ENDIF
    \ENDFOR
    \STATE \textbf{return} $\tilde{C}^{\circ}$
  \end{algorithmic}
\end{algorithm}

Together, these operations keep the sampled support compact without compromising coverage.

\subsection{Extracting a Strong Stackelberg Policy}
\label{appendix:pbvi:extract}

Once PBVI converges, the sequence \( \Lambda_0, \ldots, \Lambda_\ell \) contains all the information required to reconstruct an approximate SSE policy.

\paragraph{Recursive representation.}  
Each vector pair \( (\alpha_{\colorL}, \alpha_{\colorF}) \in \Gamma \) stores:
\begin{itemize}
    \item the leader decision rule \( \delta_{\colorL} \) applied at its parent filtered credible set,
    \item the follower action(s) selected in each conditional occupancy state,
    \item pointers to the next-stage vector pairs \( (\alpha'_{\colorL}, \alpha'_{\colorF}) \in \Gamma' \in \Lambda_{t+1} \), for each possible environment transition \( s' \in S \).
\end{itemize}

\paragraph{Unrolling the policy.}  
Policy extraction proceeds as follows:
\begin{enumerate}
    \item Identify the initial filtered credible set \( \tilde{c}_0 \) and the vector pair \( (\alpha^*_{\colorL}, \alpha^*_{\colorF}) \in \Gamma^* \in \Lambda_0 \) maximising the leader value.
    \item At each time step \( t \), record the stored \( \delta_{\colorL} \) as the leader policy, and the corresponding best-response mapping from conditional occupancy states to \( a_{\colorF} \).
    \item Follow the transition pointers to the next-stage vector pairs \( (\alpha'_{\colorL}, \alpha'_{\colorF}) \) and repeat until horizon \(\ell\).
\end{enumerate}

\paragraph{Result.}  
This process reconstructs:
\begin{itemize}
    \item a full sequence of stage-wise leader decision rules \( \delta_{\colorL,0:\ell-1} \),
    \item best-response mappings for the follower consistent with SSE semantics,
    \item policies grounded in value approximations stored during PBVI.
\end{itemize}

The resulting leader–follower policy profile is thus consistent with the approximate value function and satisfies the SSE assumptions by construction. This avoids re-solving any decision problem at execution time.

%% file: appendix/algorithms.tex
This appendix provides algorithmic details for the baseline methods used in our experiments.

\subsection{PBVI Baselines}
\label{appendix:baselines:pbvi}

We consider two variations of the Point-Based Value Iteration (PBVI) method introduced in Appendix~\ref{appendix:pbvi}:

\paragraph{PBVI-H.} This variant computes policy decision rules that map private decision-rule histories to actions. It allows policies to be fully history-dependent.

\paragraph{PBVI-S.} This variant restricts the leader's policies to be Markovian, i.e., dependent only on the current state. This constraint enables state-based filtering of credible sets.

\subsection{Backward Induction Baselines}
\label{appendix:baselines:backward}

We evaluate two versions of backward induction (BI), differing in the criterion used to compute follower responses.

\paragraph{BI.} This method corresponds to the formulation in Proposition~\ref{prop:wrong:pi}, where the follower response maximises the complete expected cumulative reward.

\paragraph{MY.} In this variant, the follower is assumed to be \emph{myopic}, selecting actions solely to maximise immediate reward. This results in a myopic best response that disregards long-term value.

\subsection{Normal-Form Game (NFG) Baselines}
\label{appendix:baselines:nfg}

These baselines reformulate the sequential game as a one-shot normal-form game (NFG), where each row and column of the payoff matrix corresponds to a deterministic policy of the leader and follower, respectively.

Although normal-form representations typically model memoryless (Markovian) interactions, we adapt this formulation to support history-dependent policies. Let \( \Pi^{\mathrm{det}}_\colorL \) and \( \Pi^{\mathrm{det}}_\colorF \equiv \Pi_\colorF \) denote the sets of deterministic policies for the leader and follower, respectively.

\paragraph{LP Baseline.} We use the OpenSpiel implementation of the LP-based SSE solver introduced by \citet{conitzer2006computing}. For each fixed deterministic follower policy \( \pi^{\mathrm{det}}_\colorF \in \Pi_\colorF \), the following linear program identifies a stochastic leader policy \( \pi_\colorL \in \Pi_\colorL \) that maximises the leader’s expected value under the constraint that \( \pi^{\mathrm{det}}_\colorF \) remains a best response:

\begin{align*}
\max_{\pi_\colorL\in \Pi_\colorL} \quad &
\sum_{\pi^{\mathrm{det}}_\colorL} \pi_\colorL(\pi^{\mathrm{det}}_\colorL) \cdot 
v^{\pi^{\mathrm{det}}_\colorL, \pi^{\mathrm{det}}_\colorF}_{\colorL, 0}(s_0, s_0, s_0) \\
 \text{Subject to}\colon \quad &
 \sum_{\pi^{\mathrm{det}}_\colorL} \pi_\colorL(\pi^{\mathrm{det}}_\colorL) \cdot 
 v^{\pi^{\mathrm{det}}_\colorL, \pi^{\mathrm{det}}_\colorF}_{\colorF, 0}(s_0, s_0, s_0)
 \geq
 \sum_{\pi^{\mathrm{det}}_\colorL} \pi_\colorL(\pi^{\mathrm{det}}_\colorL) \cdot 
 v^{\pi^{\mathrm{det}}_\colorL, \bar{\pi}^{\mathrm{det}}_\colorF}_{\colorF, 0}(s_0, s_0, s_0), 
 \quad \forall \bar{\pi}^{\mathrm{det}}_\colorF \in \Pi_\colorF.
\end{align*}

The optimal policy pair \((\pi_\colorL, \pi^{\mathrm{det}}_\colorF)\) is selected as the SSE.

\paragraph{MILP Baseline.} We also consider the MILP formulation by \citet{3038794.3038831}, which encodes follower tie-breaking constraints explicitly. Let \( \pi : \Pi^{\mathrm{det}}_\colorL \times \Pi_\colorF \to [0,1] \) be a joint distribution over deterministic policy pairs. The MILP is defined as:

\begin{align*}
\max_{\pi} &
\sum_{\pi^{\mathrm{det}}_\colorL} \sum_{\pi^{\mathrm{det}}_\colorF} \pi(\pi^{\mathrm{det}}_\colorL, \pi^{\mathrm{det}}_\colorF) \cdot 
v^{\pi^{\mathrm{det}}_\colorL, \pi^{\mathrm{det}}_\colorF}_{\colorL, 0}(s_0, s_0, s_0) \\
\text{s.t.} \quad &
\sum_{\pi^{\mathrm{det}}_\colorF} \pi(\pi^{\mathrm{det}}_\colorL, \pi^{\mathrm{det}}_\colorF) = \pi_\colorL(\pi^{\mathrm{det}}_\colorL), 
\quad \forall \pi^{\mathrm{det}}_\colorL \in \Pi^{\mathrm{det}}_\colorL, \\
& \sum_{\pi^{\mathrm{det}}_\colorL} \pi(\pi^{\mathrm{det}}_\colorL, \pi^{\mathrm{det}}_\colorF) = \pi_\colorF(\pi^{\mathrm{det}}_\colorF), 
\quad \forall \pi^{\mathrm{det}}_\colorF \in \Pi_\colorF, \\
& \sum_{\pi^{\mathrm{det}}_\colorL} \sum_{\pi^{\mathrm{det}}_\colorF} \pi(\pi^{\mathrm{det}}_\colorL, \pi^{\mathrm{det}}_\colorF) \cdot 
v^{\pi^{\mathrm{det}}_\colorL, \pi^{\mathrm{det}}_\colorF}_{\colorF, 0}(s_0, s_0, s_0) 
\geq 
\sum_{\pi^{\mathrm{det}}_\colorL} \pi_\colorL(\pi^{\mathrm{det}}_\colorL) \cdot 
v^{\pi^{\mathrm{det}}_\colorL, \bar{\pi}^{\mathrm{det}}_\colorF}_{\colorF, 0}(s_0, s_0, s_0), \quad \forall \bar{\pi}^{\mathrm{det}}_\colorF \in \Pi_\colorF\\
&
\sum_{\pi^{\mathrm{det}}_\colorL} \sum_{\pi^{\mathrm{det}}_\colorF}
\pi(\pi^{\mathrm{det}}_\colorL, \pi^{\mathrm{det}}_\colorF)  = 1\\
& \pi(\pi^{\mathrm{det}}_\colorL, \pi^{\mathrm{det}}_\colorF) \geq 0, \pi_\colorL(\pi^{\mathrm{det}}_\colorL) \geq 0,
\quad \forall \pi^{\mathrm{det}}_\colorL\in \Pi^{\mathrm{det}}_\colorL, \pi^{\mathrm{det}}_\colorF\in \Pi_\colorF\\
&
\pi_\colorF(\pi^{\mathrm{det}}_\colorF) \in \{0,1\},
\quad \forall  \pi^{\mathrm{det}}_\colorF\in \Pi_\colorF.
\end{align*}

%% file: appendix/benchmark.tex
This appendix describes the benchmark environments used in our experiments. Each benchmark is defined by a finite state space \(S\), leader and follower action spaces \( A_{\colorL} \) and \( A_{\colorF} \), transition function \(T\), and reward functions \( r_{\colorL} \) and \( r_{\colorF} \). All benchmark files, including state definitions, action sets, transitions, and reward functions, are provided in the supplementary code.

\paragraph{Centipede \cite{rosenthal1981games}.}
This benchmark adapts the classical Centipede game. Players alternate moves over a finite horizon \(\ell\): the leader acts on odd stages, the follower on even ones. At each stage, the current player may either \emph{stop}, ending the game and claiming a disproportionate share of the accumulated rewards, or \emph{continue}, increasing the pool for future stages. The benchmark tests commitment and strategic foresight in a setting where early defection yields short-term gain but limits cumulative reward.

\paragraph{Match \cite{10.5555/2900929.2900938}.}
The Match benchmark is derived from a counterexample demonstrating the suboptimality of Markov policies in general-sum stochastic games. The follower first chooses between actions that may penalise the leader (e.g., yielding a reward \(\varepsilon\) to the follower but cost \(-\mathbf{M}\) to the leader). The game then transitions deterministically to states encoding the follower’s past action. The leader’s subsequent decision can reset the game or impose penalties. This benchmark illustrates how history-dependent leader policies can constrain follower behaviour more effectively than Markov ones.

\paragraph{Patrolling \cite{10.5555/2900929.2900938}.}
The Patrolling benchmark is a zero-sum security game. The leader (defender) and the follower (attacker) move across a graph of targets. At each timestep, the attacker selects a target to attack, while the defender attempts to intercept. The game is adversarial: the attacker gains if the target is undefended; otherwise, the defender succeeds. Transitions are deterministic based on agent movement. This benchmark tests planning under spatial constraints, imperfect information, and adversarial reasoning.

\paragraph{Dec-Tiger \cite{nair2003taming}.}
This benchmark adapts the fully observable version of the classic Dec-Tiger problem. At each stage, both agents observe the location of a tiger behind one of two doors. They independently choose which door to open. Opening the door without the tiger yields a positive reward; opening the tiger door yields a penalty. The tiger’s location is uniformly resampled at each stage. The game tests simultaneous decision-making and partial action coordination under shared state information.

\paragraph{MABC \cite{hansen2004dynamic}.}
The Multi-Agent Broadcast Channel (MABC) benchmark is a cooperative communication problem. Multiple agents share a single broadcast channel and must coordinate their transmissions to avoid collisions. At each timestep, agents independently choose whether to transmit. If exactly one agent transmits, all receive a reward of 1; otherwise, the reward is 0. The current state encodes whether the previous transmission succeeded or failed. This benchmark tests decentralised coordination through implicit signalling under communication constraints.

\subsection{Scalability Limitations}
\label{appendix:baselines:scalability}

The number of deterministic history-dependent policies for agent \( \colorI \) grows exponentially with the action set \( A_i \), state space \( S \), and time horizon \( \ell \). The total number of possible deterministic policies is:

\[
|\Pi^{\mathrm{det}}_{\colorI}| = |A_{\colorI}|^{\sum_{t=0}^{\ell-1} |A_{\colorI} \times S|^t}.
\]

\begin{table}[ht]
\centering
\begin{tabular}{lcccc}
\toprule
\textbf{Benchmark} & \textbf{Horizon 1} & \textbf{Horizon 2} & \textbf{Horizon 3} & \textbf{Horizon 4} \\
\midrule
Centipede   & \( 2^1 \)   & \( 2^5 \)     & \( 2^{73} \)    & \( 2^{1111} \) \\
Match       & \( 2^1 \)   & \( 2^9 \)     & \( 2^{73} \)    & \( 2^{585} \)  \\
Tiger       & \( 2^1 \)   & \( 2^5 \)     & \( 2^{21} \)    & \( 2^{85} \)   \\
MABC        & \( 2^1 \)   & \( 2^9 \)     & \( 2^{73} \)    & \( 2^{585} \)  \\
Patrolling  & \( 5^1 \)   & \( 5^{31} \)  & \( 5^{931} \)   & \( 5^{27931} \) \\
\bottomrule
\end{tabular}
\caption{Number of deterministic policies per agent per planning horizon, illustrating exponential blow-up.}
\label{table:nb_hist_det_policies}
\end{table}

As illustrated in Table~\ref{table:nb_hist_det_policies}, these baselines are exact but become computationally intractable beyond small horizons due to the combinatorial explosion in policy space.

%% file: appendix/results.tex
This appendix provides additional experimental results on the benchmark problems using baseline algorithms over extended horizons, in order to analyse the scalability of the proposed method. We also report the size of the value function for the PBVI baselines, included as an additional metric in the comprehensive result table (Table~\ref{tab:complete-results}).

\begin{table}[H]
\centering
\setlength{\tabcolsep}{3pt}
\renewcommand{\arraystretch}{1.1}
\scalebox{1}{
\begin{tabular}{c
  ccc ccc ccc ccc ccc ccc}
\toprule
\multicolumn{1}{c}{} 
& \multicolumn{3}{c}{\textbf{H}} 
& \multicolumn{3}{c}{\textbf{S}} 
& \multicolumn{2}{c}{\textbf{BI}} 
& \multicolumn{2}{c}{\textbf{MY}} 
& \multicolumn{2}{c}{\textbf{LP}} 
& \multicolumn{2}{c}{\textbf{MILP}} \\[-1mm]
\cmidrule(lr){2-4} \cmidrule(lr){5-7} \cmidrule(lr){8-9} \cmidrule(lr){10-11} \cmidrule(lr){12-13} \cmidrule(lr){14-15}
$\ell$ 
& V & VF size & T 
& V & VF size & T 
& V & T 
& V & T 
& V & T 
& V & T \\
\midrule
\multicolumn{13}{c}{\textsc{Centipede}  \cite{rosenthal1981games}} \\[-0.5mm]
\hline
1 & \pl1{\bf1} & 2 & 0.01 & \pl1{\bf1} & 2 & 0.01 & \pl1{\bf1} & 0.01 & \pl1{\bf1} & 0.01 & \pl1{\bf1} & 0.25 & \pl1{\bf1} & 0.14 \\
\myrowcolour
2 & \pl1{\bf1} & 4& 0.01 & \pl1{\bf1} & 3 & 0.01 & \pl1{\bf1} & 0.01 & \pl1{\bf1} & 0.01 & \pl1{\bf1} & 20.21 & \pl1{\bf1} & 7.29 \\
3 & \pl1{\bf2} & 4 & 0.01 & \pl1{\bf2} & 4 & 0.01 & 1 & 0.01 & 1 & 0.01 & --- & --- & --- & --- \\
\myrowcolour
4 & \pl1{\bf2.67} & 7 & 1.56 & \pl1{\bf2.67} & 5 & 0.02 & 1 & 0.05 & 1 & 0.06 & --- & --- & --- & --- \\
5 & \pl1{\bf4} & 9 & 2.69 & \pl1{\bf4} & 6 & 0.04 & 1 & 0.10 & 1 & 0.10 & --- & --- & --- & --- \\
\myrowcolour
6 & \pl1{\bf4.67} & 9 & 3.68 & \pl1{\bf4.67} & 7 & 0.11 & 1 & 0.17 & 1 & 0.13 & --- & --- & --- & --- \\
7 & \pl1{\bf6} & 10 & 5.34 & \pl1{\bf6} & 8 & 0.20 & 1 & 0.18 & 1 & 0.15 & --- & --- & --- & --- \\
\myrowcolour
8 & \pl1{\bf7} & 14 & 8.18 & \pl1{\bf6.67} & 9 & 0.46 & 1 & 0.20 & 1 & 0.49 & --- & --- & --- & --- \\
9 & \pl1{\bf8} & 13 & 11.53 & \pl1{\bf8} & 10 & 0.71 & 1 & 0.21 & 1 & 0.25 & --- & --- & --- & --- \\
\myrowcolour
10 & \pl1{\bf9} & 15 & 13.00 & \pl1{\bf8.67} & 11 & 1.32 & 1 & 0.3 & 1 & 0.27 & --- & --- & --- & --- \\
\midrule
\multicolumn{13}{c}{\textsc{Match} \cite{10.5555/2900929.2900938}} \\
\hline
1 & \pl1{\bf-1000} & 2 & 0.01 & \pl1{\bf-1000} & 2 & 0.01 & \pl1{\bf-1000} & 0.01 & \pl1{\bf-1000} & 0.01 & \pl1{\bf-1000} & 0.2 & \pl1{\bf-1000} & 0.3 \\
\myrowcolour
2 & \pl1{\bf-1000} & 3 & 0.01 & \pl1{\bf-1000} & 3 & 0.01 & \pl1{\bf-1000} & 0.01 & \pl1{\bf-1000} & 0.01 & --- & --- & --- & --- \\
3 & \pl1{\bf0}& 4 & 0.01 & -1000 & 4 & 0.01 & -1000 & 0.04 & -1000 & 0.02 & --- & --- & --- & --- \\
\myrowcolour
4 & \pl1{\bf0}& 5 & 0.02 & -2000 & 5 & 0.02 & -2000 & 0.05 & -2000 & 0.05 & --- & --- & --- & --- \\
5 & \pl1{\bf0}& 6 & 0.03 & -2000 & 6 & 0.02 & -2000 & 0.05 & -2000 & 0.06 & --- & --- & --- & --- \\
\myrowcolour
6 & \pl1{\bf0} & 7 & 0.05 & -2000 & 7 & 0.11 & -2000 & 0.07 & -2000 & 0.08 & --- & --- & --- & --- \\
7 & \pl1{\bf0}& 8 & 0.07 & -3000 & 8 & 0.15 & -3000 & 0.08 & -3000 & 0.10 & --- & --- & --- & --- \\
\myrowcolour
8 & \pl1{\bf0}& 9 & 0.11 & -3000 & 9 & 0.04 & -3000 & 0.10 & -3000 & 0.10 & --- & --- & --- & --- \\
9 & \pl1{\bf0}& 10 & 0.19 & -3000 & 10 & 0.10 & -3000 & 0.11 & -3000 & 0.12 & --- & --- & --- & --- \\
\myrowcolour
10 & \pl1{\bf0}& 11 & 0.40 & -4000 & 11 & 0.22 & -4000 & 0.11 & -4000  & 0.11 & --- & --- & --- & --- \\
11 & \pl1{\bf0}& 12 & 0.96 & -4000 & 12 & 0.04 & -4000 & 0.12 & -4000 & 0.15 & --- & --- & --- & --- \\
\myrowcolour
12 & \pl1{\bf0}& 13 & 1.99 & -4000 & 13 & 0.26 & -4000 & 0.13 & -4000  & 0.15 & --- & --- & --- & --- \\
13 & \pl1{\bf0}& 14 & 5.68 & -5000 & 14 & 0.23 & -5000 & 0.14 & -5000 & 0.16 & --- & --- & --- & --- \\
\myrowcolour
14 & \pl1{\bf0}& 15 & 16.5 & -5000 & 15 & 0.53 & -5000 & 0.17 & -5000 & 0.15 & --- & --- & --- & --- \\
15 & \pl1{\bf0}& 16 & 63.52 & -5000 & 16 & 0.42 & -5000 & 0.19 & -5000 & 0.17 & --- & --- & --- & --- \\
\midrule
\multicolumn{13}{c}{\textsc{Dec-Tiger} \cite{nair2003taming}} \\
\hline
1 & \pl1{\bf20} & 2 & 0.01 & \pl1{\bf20} & 2 & 0.01 & \pl1{\bf20} & 0.01 & \pl1{\bf20} & 0.01 & \pl1{\bf20} & 0.19 & \pl1{\bf20} & 0.22 \\
\myrowcolour
2 & \pl1{\bf40} & 4 & 0.01 & \pl1{\bf40} & 3 & 0.01 & \pl1{\bf40} & 0.01 & \pl1{\bf40} & 0.01 & \pl1{\bf40} & 1.07 & \pl1{\bf40} & 0.23 \\
3 & \pl1{\bf60} & 6 & 0.04 & \pl1{\bf60} & 4 & 0.03 & \pl1{\bf60} & 0.01 & \pl1{\bf60} & 0.01 & --- & --- & --- & --- \\
\myrowcolour
4 & \pl1{\bf80} & 8 & 0.06 & \pl1{\bf80} & 5 & 0.04 & \pl1{\bf80} & 0.05 & \pl1{\bf80} & 0.09 & --- & --- & --- & --- \\
5 & \pl1{\bf100} & 10 & 0.10 & \pl1{\bf100} & 6 & 0.12 & \pl1{\bf100} & 0.14 & \pl1{\bf100} & 0.11 & --- & --- & --- & --- \\
\myrowcolour
6 & \pl1{\bf120} & 12 & 0.45 & \pl1{\bf120} & 7 & 0.30 & \pl1{\bf120} & 0.20 & \pl1{\bf120} & 0.10 & --- & --- & --- & --- \\
7 & \pl1{\bf120} & 14 & 3.94 & \pl1{\bf120} & 8 & 0.96 & \pl1{\bf120} & 0.22 & \pl1{\bf120} & 0.14 & --- & --- & --- & --- \\
\myrowcolour
8 & \pl1{\bf160} & 16 & 83.42 & \pl1{\bf160} & 9 & 3.30 & \pl1{\bf160} & 0.22 & \pl1{\bf160} & 0.11 & --- & --- & --- & --- \\
\midrule
\multicolumn{13}{c}{\textsc{MABC} \cite{hansen2004dynamic}} \\
\hline
1 & \pl1{\bf1} & 2 & 0.01 & \pl1{\bf1} & 2 & 0.01 & \pl1{\bf1} & 0.06 & \pl1{\bf1} & 0.03 & \pl1{\bf1} & 0.20 & \pl1{\bf1} & 0.19 \\
\myrowcolour
2 & \pl1{\bf2} & 4 & 0.11 & \pl1{\bf2} & 3 & 0.05 & \pl1{\bf2} & 0.09 & \pl1{\bf2} & 0.08 & --- & --- & --- & --- \\
3 & \pl1{\bf2.99} & 7 & 0.73 & \pl1{\bf2.99} & 4 & 0.49 & \pl1{\bf2.99} & 0.05 & \pl1{\bf2.99} & 0.06 & --- & --- & --- & --- \\
\myrowcolour
4 & \pl1{\bf3.97} & 9 & 0.96 & \pl1{\bf3.97} & 5 & 0.53 & \pl1{\bf3.97} & 0.05 & \pl1{\bf3.97} & 0.06 & --- & --- & --- & --- \\
5 & \pl1{\bf4.95} & 28 & 3.60 & \pl1{\bf4.95} & 6 & 0.57 & \pl1{\bf4.95} & 0.05 & \pl1{\bf4.95} & 0.06 & --- & --- & --- & --- \\
\myrowcolour
6 & \pl1{\bf5.93} & 74 & 208.6 & \pl1{\bf5.93} & 7 & 0.59 & \pl1{\bf5.93} & 0.09 & \pl1{\bf5.93} & 0.08 & --- & --- & --- & --- \\
\midrule
\multicolumn{13}{c}{\textsc{Patrolling} \cite{10.5555/2900929.2900938}} \\
\hline
1 & \pl1{\bf0} & 2 & 0.01 & \pl1{\bf0} & 2 & 0.01 & \pl1{\bf0} & 0.01 & \pl1{\bf0} & 0.01 & \pl1{\bf0} & 0.1 & \pl1{\bf0} & 0.09 \\
\myrowcolour
2 & \pl1{\bf0} & 6 & 0.04 & \pl1{\bf0} & 3 & 0.02 & \pl1{\bf0} & 0.02 & \pl1{\bf0} & 0.02 & --- & --- & --- & --- \\
3 & \pl1{\bf0} & 9 & 3.41 & \pl1{\bf0} & 4 & 1.1 & \pl1{\bf0} & 0.10 & \pl1{\bf0} & 0.07 & --- & --- & --- & --- \\
\myrowcolour
4 & \pl1{\bf0} & 18 & 5.07 & \pl1{\bf0} & 5 & 1.58 & \pl1{\bf0} & 0.01 & \pl1{\bf0} & 0.01 & --- & --- & --- & ---  \\
5 & \pl1{\bf0} & 26 & 17.26 & \pl1{\bf0} & 6 & 3.19 & \pl1{\bf0} & 0.01 & \pl1{\bf0} & 0.01 & --- & --- & --- & ---  \\
\myrowcolour
6 & \pl1{\bf0} & 39 & 66.23 & \pl1{\bf0} & 7 & 4.0 & \pl1{\bf0} & 0.14 & \pl1{\bf0} & 0.26 & --- & --- & --- & --- \\
\bottomrule
\end{tabular}
}
\caption{SSE leader value (V), runtime (T in seconds), value function size (VF size in number of hyperplane pairs). 
\pl1{\bf Bold leader values} indicate the highest return across planning variants for each configuration.
``---'' indicates timeout. 
}
\label{tab:complete-results}
\end{table}

%% file: appendix/sentitivity.tex
This appendix presents a sensitivity analysis of the PBVI baselines. The experimental results for both \textbf{H} and \textbf{S} are influenced by the configuration of PBVI parameters. The primary source of variability lies in the sampling subprocess, where several interdependent parameters can be tuned. Our analysis focuses on two key parameters: the number of credible sets and the maximum number of occupancy states allowed per credible set.

\subsection{Sensitivity to the Number of Occupancy States}

We begin by examining the sensitivity of methods \textbf{H} and \textbf{S} to the number of occupancy states. The plots below report the average results over five runs, with the number of credible sets fixed to 1 and the planning horizon set to 4. The number of occupancy states per credible set varies from 1 to 5. The leader’s expansion policy is uniform.

\begin{figure}[h]
    \centering
    \includegraphics[width=\linewidth]{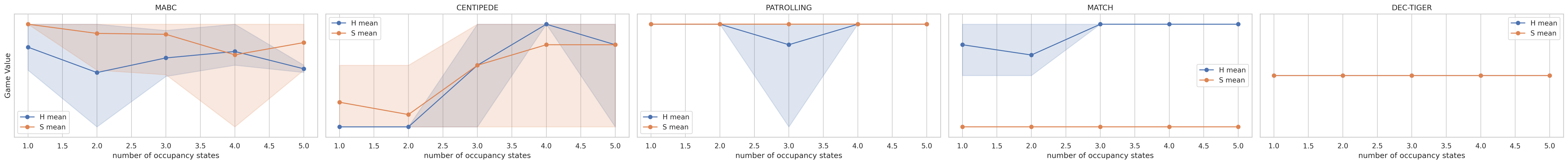}
    \caption{Value sensitivity to the number of occupancy states.}
    \label{fig:occupancy-sensitivity}
\end{figure}

Increasing the number of occupancy states generally improves the value estimates, particularly in environments such as \textsc{Centipede} and \textsc{Match}. However, the magnitude of improvement is not strictly monotonic; for example, in \textsc{Centipede}, performance plateaus despite additional occupancy states. In \textsc{Dec-Tiger} and \textsc{Patrolling}, optimal values are often achieved with a single occupancy state. Moreover, in some cases (e.g., \textsc{Patrolling} with method \textbf{H}), adding redundant occupancy states introduces noise, slightly degrading value estimates.

\subsection{Sensitivity to the Number of Credible Sets}

We now analyse the effect of varying the number of credible sets, with the number of occupancy states per set fixed to 5. As before, we fix the planning horizon to 4 and use a uniform leader expansion policy. Results are averaged across five independent runs.

\begin{figure}[h]
    \centering
    \includegraphics[width=\linewidth]{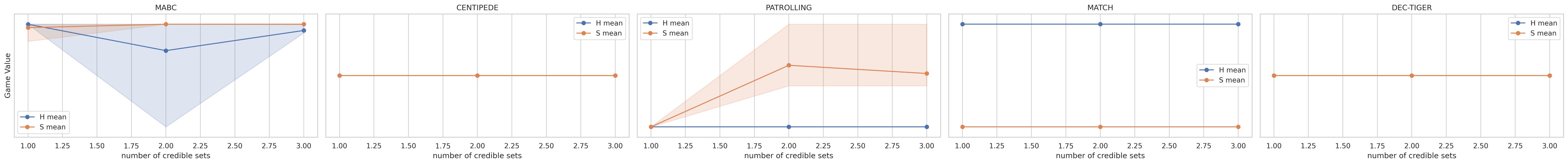}
    \caption{Value sensitivity to the number of credible sets.}
    \label{fig:credible-sensitivity}
\end{figure}

In \textsc{Match}, \textsc{Centipede}, and \textsc{Dec-Tiger}, optimal values are typically reached with the first expanded credible set, and additional sets do not yield further improvements. In contrast, in \textsc{MABC}, adding a second credible set causes value variation for method \textbf{S}. For \textsc{Patrolling}, expanding the set of beliefs allows method \textbf{H} to exploit artefacts in the follower's best responses, leading to overestimated leader policies.